\documentclass[letterpaper,11pt]{article}
\usepackage{comment}
\usepackage[utf8]{inputenc}
\usepackage{enumerate}
\usepackage[OT4]{fontenc}
\usepackage{graphicx}
\usepackage{amssymb}
\usepackage{amsmath}
\usepackage{amsfonts}
\usepackage{amsthm}
\usepackage{tikz}
\usepackage{caption}
\usepackage{tikzscale}
\usepackage{authblk}
\usepackage{enumitem}
\usepackage{thmtools}
\usepackage{thm-restate}
\usepackage[colorinlistoftodos,textsize=tiny,textwidth=2cm,color=red!25!white,obeyFinal]{todonotes}
\usepackage[hypertexnames=false,colorlinks=true,citecolor=olive]{hyperref}
\usetikzlibrary{backgrounds}
\usetikzlibrary{patterns}
\usetikzlibrary{decorations.pathmorphing}
\usetikzlibrary{decorations.pathreplacing}
\usepackage{xspace}
\usepackage{breqn}

\usepackage[letterpaper,left=1in,right=1in,top=1in,bottom=1in]{geometry}
\newtheorem{theorem}{Theorem}[section]
\newtheorem{lemma}[theorem]{Lemma}

\newtheorem{fact}[theorem]{Fact}
\newtheorem{definition}[theorem]{Definition}
\newtheorem{corollary}[theorem]{Corollary}

\newtheorem{remark}[theorem]{Remark}

\def\DEBUG{true}

\ifdefined\DEBUG

  \def\rem#1{{\marginpar{\raggedright\scriptsize #1}}}
  \newcommand{\adar}[1]{\rem{\textcolor{blue}{$\bullet$ #1}}}
  \newcommand{\nikr}[1]{\rem{\textcolor{red}{$\bullet$ #1}}}
  \newcommand{\pinr}[1]{\rem{\textcolor{green}{$\bullet$ #1}}}
\else
  
  \newcommand{\adar}[1]{}
  
  \newcommand{\nikr}[1]{}
  
  \newcommand{\pinr}[1]{}
\fi

\newcommand{\Ot}{\ensuremath{\widetilde{O}}}
\newcommand{\nnote}[1]{\todo{Nikos: #1}\xspace}

\newcommand{\anote}[1]{\todo{Adam: #1}\xspace}
\newcommand{\rev}[1]{\ensuremath{{#1}^{\text{R}}}}

\newcommand{\enter}[1]{\ensuremath{enter({#1})}}
\newcommand{\leave}[1]{\ensuremath{leave({#1})}}
\newcommand{\mnv}[1]{\ensuremath{min({#1})}}
\newcommand{\mxv}[1]{\ensuremath{max({#1})}}

\newcommand{\pord}{ord}
\newcommand{\pordr}{ord^R}
\newcommand{\tsz}{size}
\newcommand{\tszr}{size^R}

\def\polylog{\operatorname{polylog}}

\author[1]{Giuseppe F. Italiano\thanks{Giuseppe F. Italiano is partially supported by MUR, the Italian Ministry for University and Research, under PRIN Project AHeAD (Efficient Algorithms for HArnessing Networked Data).}}
\author[2]{Adam Karczmarz\thanks{Supported by ERC Consolidator
Grant 772346 TUgbOAT, the Polish National Science Centre 2018/29/N/ST6/00757 grant, and by the Foundation for Polish Science (FNP) via the START programme.}}
\author[3]{Nikos Parotsidis\thanks{This work was partially done while the author was employed at the University of Copenhagen supported by the Grant Number 16582, Basic Algorithms Research Copenhagen (BARC), from the VILLUM Foundation.}}

\affil[1]{LUISS University, Rome, Italy}
\affil[ ]{\texttt{gitaliano@luiss.it}\medskip}

\affil[2]{Institute of Informatics, University of Warsaw, Poland}
\affil[ ]{\texttt{a.karczmarz@mimuw.edu.pl}\medskip}

\affil[3]{Google Research}
\affil[ ]{\texttt{nikosp@google.com}}

\begin{document}

\setitemize{itemsep=-1pt}

\begin{titlepage}
\date{}
  \title{Planar Reachability Under Single Vertex or Edge Failures}
  \maketitle
  \begin{abstract}
    In this paper we present an efficient reachability oracle under single-edge or single-vertex failures for planar directed graphs.
Specifically, we show that a planar digraph $G$ can be preprocessed
in $O(n\log^2{n}/\log\log{n})$ time, producing an $O(n\log{n})$-space
data structure that can answer in $O(\log{n})$ time whether $u$ can reach $v$ in $G$ if the vertex $x$ (the edge~$f$) is removed from $G$, for any query vertices $u,v$ and failed vertex $x$ (failed edge $f$).
To the best of our knowledge, this 
is
the first data
structure for planar directed graphs with nearly optimal preprocessing
time that answers
all-pairs queries under any kind
of failures in polylogarithmic time.

We also consider 2-reachability problems, where we are given a planar  digraph $G$ and we wish to determine if there are two vertex-disjoint (edge-disjoint)
paths from $u$ to $v$, for query vertices $u,v$. In this setting we provide a nearly optimal 2-reachability oracle, which is the existential variant of the reachability oracle under single failures, with the following bounds.
We can construct in $O(n\polylog{n})$ time an $O(n\log^{3+o(1)}{n})$-space
data structure that can check in $O(\log^{2+o(1)}{n})$ time for any query vertices $u,v$ whether $v$ is 2-reachable from $u$, 
  or otherwise find some separating vertex (edge) $x$ lying on all paths from $u$ to $v$ in $G$. 

To obtain our results, we follow the general recursive approach
of Thorup for reachability
in planar graphs [J.~ACM~'04] and we present new
data structures which generalize 
dominator trees and previous  data structures for strong-connectivity under failures
[Georgiadis et al., SODA~'17]. Our new data structures work also for general digraphs and may be of independent interest.

  \end{abstract}
  \thispagestyle{empty}
\end{titlepage}

\section{Introduction}
Computing reachability is perhaps one of the most fundamental problems in directed graphs. 
Let $G=(V,E)$ be a directed graph with $n$ vertices and $m$ edges.
The transitive closure (i.e., all-pairs reachability) problem consists of computing  whether there is a directed path from $u$ to $v$, for all pairs of vertices $u,v\in V$. 
The single-source reachability variant asks for each $v\in V$ whether there exists a path from $s$ to $v$, where $s\in V$ is fixed.
While single-source reachability can be solved in 
optimal $O(m)$ time, the fastest algorithm for
computing transitive closure runs in
$\Ot(\min(n^\omega,nm))$ time, where $\omega<2.38$ is the matrix multiplication exponent.
Notice that for solving the all-pairs reachability problem one needs 
$O(n^2)$ space to store the information for all pairs of vertices.

In the oracle variant of all-pairs reachability, we wish to preprocess the input graph and build a data structure that can answer reachability queries between any pair of vertices, while trying to minimize the query time, the preprocessing time, as well as
the size of the data structure.
Henzinger et al. \cite{henzinger2017conditionalLowerBounds} gave conditional lower bounds for ``combinatorial''\footnote{That is, not relying on fast matrix multiplication algorithms, which are often considered impractical.} constructions for this problem. 
Specifically, they showed that there is no all-pairs reachability oracle that simultaneously requires $O(n^{3-\epsilon})$ time preprocessing and supports queries in $O(n^{2-\epsilon})$ time (for all $m$), unless there is a truly ``combinatorial'' algorithm that can multiply two $n\times n$ boolean matrices in $O(n^{3-\epsilon})$ time.
%

However, non-trivial reachability oracles  are known for a few important graph classes. 
Most notably, for planar digraphs, for which $m=O(n)$,
the first reachability oracle with near-linear preprocessing
and polylogarithmic query time was obtained by Thorup~\cite{Thorup04},
whereas a decade later
Holm et al. \cite{HolmRT15} presented an asymptotically optimal oracle, with $O(n)$ space and preprocessing time and constant query time.
For graph classes admitting balanced separators of size $s(n)$ (which include graphs with treewidth $O(s(n))$, and minor-free graphs
for $s(n)=O(\sqrt{n})$), an $\Ot(n\cdot s(n))$-space reachability oracle
with query time $\Ot(s(n))$ exists\footnote{To obtain such an oracle it is enough to precompute single-source reachability from/to all the $O(s(n))$ vertices
of the separator and recurse on the components of $G$ after removing the separator.}.

Real-world networks undoubtedly experience link or node failures. 
This has motivated the research community to develop 
graph algorithms and data structures that can efficiently deal with failures.
A notable example is the notion of dominators in digraphs with respect to a source vertex~$s$.  We say that a vertex $x$ dominates a vertex $v$ if all paths from $s$ to $v$ contain $x$. The dominance relation from $s$ is transitive and can be represented via a tree called the \emph{dominator tree} from $s$. The dominator tree from $s$ allows one to answer several reachability under failure queries, such as ``are there two edge- (or vertex-) disjoint paths from $s$ to $v$?" and ``is there a path from $s$ to $v$ avoiding a vertex $x$ (or an edge~$e$)?", in asymptotically optimal time.
The notion of dominators has been widely used in domains like circuit testing \cite{amyeen:01:vlsitest}, theoretical biology \cite{foodwebs:ab04}, memory profiling~\cite{memory-leaks:mgr2010}, constraint programming \cite{QVDR:PADL:2006}, connectivity \cite{2ECC:GILP:TALG}, just to state some.
Due to their numerous applications, dominators have been extensively studied for over four decades \cite{cf:ac, purdom1972immediate,domin:lt,domin:tarjan} and several linear-time algorithms for computing dominator trees are known \cite{domin:ahlt,dominators:bgkrtw,domin:bkrw,dom:gt04}.
While extremely useful, dominator trees are restricted to answering queries only from a single source $s$.

Oracles that answer queries in the presence of failures are often called \emph{$f$-sensitivity
oracles}\footnote{We adopt the use of the term from \cite{henzinger2017conditionalLowerBounds}. We note that other terms have also been used in the literature, such as ``fault-tolerant" oracles or oracles "for failure prone graphs".},
where $f$ refers to the upper bound on the number of failures allowed. 
If not explicitly mentioned, in this paper when we refer to sensitivity oracles we refer to $1$-sensitivity oracles that allow failures of either one edge or one vertex. In what follows we denote by $G-F$ the graph obtained from $G$ after deleting the set of vertices or edges $F$. If $F$ is a single vertex or edge $x$, we write $G-x$.

In this paper, we study 1-sensitivity oracles for the all-pairs reachability problem in planar digraphs.
Specifically, we wish to preprocess a planar graph efficiently and build a possibly small (in terms of space)
data structure  that can efficiently answer queries of the form
``is there a path from $u$ to $v$ avoiding $x$?", for query vertices $u,v$ and vertex (or edge) $x$.
Moreover, we study \emph{2-reachability problems}, where, given a directed graph $G$, we wish to determine if there are two vertex-disjoint (resp., edge-disjoint) paths from $u$ to $v$, for query vertices $u,v$. In particular, we consider \emph{2-reachability oracles}, which are 
 the existential variant of  1-sensitivity reachability oracles. Here, the desired data structure should,
for an arbitrary pair of query vertices $(u,v)$, efficiently find a vertex $x\notin\{u,v\}$ (resp., an edge $e$)
whose failure destroys all $u\to v$ paths in the graph,
or declare there is none. Note that in the latter case,  vertex $v$ is  2-reachable from vertex $u$\footnote{The name 2-reachability
comes from the fact that, by Menger's theorem, there exist two internally vertex-disjoint
$u\to v$ paths if and only if no single failing vertex can make $v$ unreachable from $u$}.

Our data structures support the aforementioned queries answered with dominator trees,
but we allow a source $s$ to be a query parameter as well, as opposed to a dominator tree which assumes a fixed source.
We focus on planar graphs not only because they are one of the most studied non-general classes of graphs, but also because dominator trees have been used in the past for solving problems on planar graphs (i.e., circuit testing \cite{amyeen:01:vlsitest}), and hence our result could potentially motivate further similar applications as an efficient tool that can answer all-pairs dominance queries.

Notice that a simple-minded solution to both 1-sensitivity reachability oracle and 2-reachability oracle problems with $O(n^2)$ space and preprocessing time and $O(1)$ query time is to compute the dominator tree from each source vertex $s$.
In general directed graphs, the all-pairs version of a dominator tree cannot be computed faster than matrix-multiplication or be stored in subquadratic space \cite{georgiadis2017AllPairsFTReachability}, which can be prohibitive in applications that require the processing of data of even moderate size.
In this paper we show how to achieve significantly better bounds for both these
problems when the input digraph is planar.

\paragraph{Related work.}  
There has been an extensive study of sensitivity oracles in directed graphs, with the initial studies dating several decades back.
Sensitivity oracles for single-source reachability have been studied widely under the name \emph{dominator trees} since the seventies (see e.g., \cite{domin:lt}).
Choudhary~\cite{choudhary2016DualFaultTolerant} considered the problem of computing $2$-sensitivity oracles for single-source reachability. 
In particular, she showed how to construct a data structure of size $O(n)$ that can answer in constant time 
reachability queries from a source $s$ to any vertex $v$ in $G-\{x,y\}$, for query vertices $v,x,y$. While the preprocessing time is not specified, a simple-minded initialization of her data structure requires $O(mn^2)$ time.
Baswana et al. \cite{Baswana2018FT}, considered the version of this problem with multiple failures. Specifically, they presented an $f$-sensitivity oracle, with size $O(2^fn)$ and preprocessing time $O(mn)$, that can compute the set or reachable vertices from the source vertex~$s$ in $G-F$ in $O(2^fn)$ time, where $F$ is the set of failed vertices or edges, with $|F|\leq f$.


King and Sagert \cite{king2002fully} were the first to study 1-sensitivity oracles for all-pairs reachability under single edge failures. In particular, they gave an algorithm that can answer queries in constant time in directed acyclic graphs (DAGs), after $O(n^3)$ time preprocessing. 
For general directed graphs Georgiadis et al.~\cite{georgiadis2017AllPairsFTReachability} showed a near-optimal 1-sensitivity oracle for all-pairs reachability, with $O(n^2)$ space and $O(\min\{mn, n^\omega \log n\})$ preprocessing time, that can answer in constant time queries of the form ``is there a path from $u$ to $v$ in $G-x$", for query vertices $u,v$ and failing vertex or edge $x$.
Their approach first produces dominator trees from all sources, which are then used to answer the queries.
Sensitivity oracles for reachability problems admit a trivial lower bound: they cannot be built faster than the time it takes to compute the corresponding (single-source or all-pairs) reachability problem in the static case (i.e., without failures).
Very recently, van den Brand and Saranurak \cite{brand2019sensitivity} presented an $f$-sensitivity oracle for all-pairs reachability with $O(n^2 \log n)$ size and $O(f^{\omega})$ query time, and $O(n^{\omega})$ time preprocessing. 
Their $f$-sensitivity oracle is nearly optimal for $f \in O(1)$ and is obtained by adopting an improved $f$-sensitivity oracle for the All-Pairs Shortest Paths problem. 
Therefore, the problem of constructing an $f$-sensitivity all-pairs reachability oracle in general digraphs is well understood for small values of $f$. 

As already mentioned, the $2$-reachability problem asks to build a data structure that can efficiently report for a pair of query vertices $(u,v)$ a single vertex (resp., edge) that appears in all paths from $u$ to $v$, or determine that there is no such vertex (resp., edge).  
Georgiadis et al.~\cite{georgiadis2017AllPairsFTReachability} show how to precompute the answers to all possible $2$-reachability queries in $O(\min \{n^\omega \log n, mn\})$ time.
The notion of $2$-reachability naturally generalizes to $k$-reachability where the query asks for a set of at most $(k-1)$ vertices (resp., edges) whose removal leaves $v$ unreachable from $u$, or determine that there is no such set of vertices (resp., edges).
For the case of $k$-reachability with respect to edge-disjoint paths, Abboud et al.~\cite{abboud2019faster} show how to precompute the answer for all pairs of vertices in $O(\min\{n^\omega, mn\})$ when $k=O(1)$, but only in the case of DAGs.
For the case of (non-necessarily acyclic) planar graphs, Łącki et al.~\cite{LackiNSW12} showed
an algorithm with $\Ot(n^{5/2}+n^2k)$ running time.
Hence, there are no non-trivial results on the $k$-reachability problem in general directed graphs.

Abboud et al. \cite{abboud2019faster} also considered a weaker version of $k$-reachability in which they only distinguish whether there are $k$ disjoint paths, or less (without reporting a set of at most $k-1$ vertices/edges that destroy all paths from $u$ to $v$, if such a set exists).
They show how to precompute all such answers, in the case of vertex-disjoint paths, in $O((nk)^\omega)$ time.
This weaker version of the problem can also be solved with respect to edge-disjoint paths by computing the value of all-pairs min-cut in  $O(m^\omega)$ time \cite{cheung2013graph} for general graphs and in $\Ot(n^2)$ time for planar graphs~\cite{LackiNSW12}.

The related problem of sensitivity oracles for strongly connected components (SCCs) was considered by Georgiadis et al. \cite{GeorgiadisIP17}. Specifically they presented a 1-sensitivity oracle with $O(m)$ size and preprocessing time, that can answer various SCC queries under the presence of single edge or vertex failures in asymptotically optimal time. For instance they can test whether two vertices $u,v$ are in the same SCC in $G-x$ in constant time, for query vertices $u,v$ and failed vertex or edge $x$, or they can report the SCCs of $G-x$ in $O(n)$ time.
Baswana et al. \cite{Baswana2019}, showed an $f$-sensitivity oracle with $O(2^fn^2)$ space, and $O(mn^2)$ preprocessing time, that can report the SCCs of $G-F$ in $O(2^fn \polylog n)$ time, where $F$ is a set of failed vertices or edges, for $|F|\leq f$. 

Reachability queries under edge or vertex failures can be also answered using more powerful sensitivity oracles for shortest paths or approximate shortest paths.
For planar directed graphs there is no known $o(n^2)$ space all-pairs distance sensitivity oracle with $O(\polylog n)$ query time. 
Baswana et al. \cite{BaswanaLM12} presented a single-source reachability oracle under single edge or vertex failures with $O(n \polylog n)$ space and construction time, that can report the length of the shortest path from $s$ to $v$ in $G-x$ in $O(\log n)$ time, for query vertex $v$ and a failed vertex or edge $x$.
They extend their construction to work for the all-pairs variant of the problem in $O(n^{3/2} \polylog n)$ preprocessing time and size of the oracle, and answer queries in $O(\sqrt{n} \polylog n)$ time.
Later on Charalampopoulos et al. \cite{Charalampopoulos19} presented improved sensitivity oracles for all-pairs shortest paths on planar graphs. Their sensitivity oracle also handles multiple failures at the expense of a worse trade-off between size and query time.
For the all-pairs version of the problem, their oracles have significantly worse bounds compared to the best known exact distance oracles (without failures) for planar graphs. 
The best known exact distance oracle for planar graphs with $O(\polylog n)$ query time was presented recently by Charalampopoulos et al. \cite{Charalampopoulos2019AlmostOptimalDistanceOraclesPlanar}; it uses  $O(n^{1+\epsilon})$ space and has $O(n^{(3+\epsilon)/2})$ construction time. 
However, we note that sensitivity distance oracles  cannot answer $2$-reachability queries.

Finally, note that we could in principle handle $f$-sensitivity queries with a fully dynamic reachability
or shortest-paths oracle with good worst-case update and query bounds~\cite{DiksS07, Sankowski04}.
However, not very surprisingly, this approach rarely yields better bounds than the $f$-sensitivity
solutions tailored to handle only batches of failures.

In summary, there exist efficient fault-tolerant reachability oracles for general digraphs with preprocessing time, size, and query time comparable to the fastest known static reachability oracles (without failures). This is the case, e.g., for the fault-tolerant reachability oracle and $2$-reachability oracle of \cite{georgiadis2017AllPairsFTReachability}, which almost matches the $O(\min\{mn,n^{\omega}\})$ bound for computing reachability without 
failures, and for the data structure for SCCs under failures of \cite{GeorgiadisIP17}, which has linear construction time and space, and it is capable of answering queries in asymptotically optimal time. 
However, and somehow surprisingly, such efficient fault-tolerant reachability oracles and $2$-reachability oracles, i.e., oracles with preprocessing time, size, and query time comparable to the fastest known static oracles (without failures) are not known in the case of any basic problem on \emph{planar directed graphs}.
Such oracles would implement some of the most important functionalities
of dominator trees, but for all possible sources at once.
Since dominator trees have several applications, including applications in planar digraphs,
 it seems quite natural  to ask whether such oracles exist.

An additional motivating factor for studying our problem is the large gap in known and possible bounds between undirected and directed graphs for related problems.
For the case of undirected graphs, there exist nearly optimal $f$-sensitivity oracles for answering connectivity queries under edge and vertex failures.
Duan and Pettie \cite{Duan2017ConnectivityOraclesVertexFailures} presented a near-optimal preprocessing $O(n\log{n})$-space $f$-sensitivity oracle that, for any set $F$ of up to $f$ edge-failures their oracle, spends $O(f \log f \log \log n)$ time to process the failed edges and then can answer connectivity queries in $G\setminus F$ in time $O(\log \log n)$ per query. This result is nearly optimal also for the case of planar undirected graphs. In the same paper, the authors also present near-optimal bounds for the case of vertex-failures. 
For general directed graphs, it is clear that  sensitivity oracles for all-pairs reachability  cannot achieve bounds anywhere close to the known bounds for sensitivity oracles for connectivity in undirected graphs.
While answering connectivity queries in undirected graphs is a much simpler task than answering reachability queries in digraphs, an intriguing question is whether there exists a general family of directed graphs that admits fault-tolerant reachability oracles with bounds close to the known results for sensitivity connectivity oracles in undirected graphs, even for the case of a single failure.

\paragraph{Our results.}
We answer 
the questions posed above
affirmatively by presenting the first  \emph{near-optimal} -- in terms of both time and space -- oracle handling \emph{all-pair-type} queries for \emph{directed} planar graphs
and supporting any \emph{single vertex or single edge failure}.
Specifically, we prove the following.
\begin{theorem}\label{t:mainresult}
  Let $G$ be a planar digraph. There exists
  an $O(n\log{n})$-space data structure answering queries of the
  form ``is there a $u\to v$ path in $G-x$'', where $u,v,x\in V$,
  in $O(\log{n})$ time.
  The data structure can be constructed in $O(n\log^2{n}/\log\log{n})$ time.
\end{theorem}

We remark that previous data structures handling failures in $O(\polylog n)$ time either
work~only for the single-source version of the problem (see dominator trees, or \cite{choudhary2016DualFaultTolerant} for two failures),
or work only on undirected graphs (see, e.g., \cite{Duan2010ConnectivityOracles, Duan2017ConnectivityOraclesVertexFailures} for oracles for general graphs, and~\cite{AbrahamCG12, BorradailePW12} for planar graphs),
or achieve nearly linear space only for dense graphs~\cite{brand2019sensitivity}.
%
It is worth noting that for planar digraphs vertex failures are generally more challenging than edge failures, 
since, whereas one can easily reduce edge failures to vertex failures,
the standard opposite reduction of splitting a vertex into an in- and an out-vertex
does not preserve planarity.

In order to achieve our 1-sensitivity oracle, we develop the following two new data structures that also work on general digraphs and can be of independent interest:
\begin{itemize}
\item Given a digraph $G$ and a directed path $P$ of $G$, we present a linear-space data structure that, after preprocessing $G$ in $O(n+m \log m/ \log \log m)$ time, can answer whether there exists a path from $u$ to $v$ in $G-x$ passing through a vertex of $P$, for any query vertices $u,v\in V$ and $x\in V\setminus V(P)$. In a sense, this result generalizes the dominator tree (a tree $T$ rooted at a source-vertex $s$ such that a vertex $t$ is reachable from $s$ in $G-x$ if and only if it is reachable in $T-x$) in
the following way.
Note that a pair of dominator trees from and to $s$ can be used to support queries of the form
``is $u$ reachable from $v$ through a ``hub'' $s$ in $G-x$?'', where $u,v,x\neq s$ are all query vertices and $s$ is fixed.
Our data structure allows to replace the single hub $s$
with any number of hubs that form a directed path, provided that these hubs cannot fail.
Since dominator trees have numerous applications, as discussed before,
we believe that our generalization can find other applications as well.

\item We show that given a digraph $G$ and an assignment of real-valued labels to the vertices of $G$,
in $O(m+n (\log n \log \log n)^{2/3})$ time one can construct a linear-space data structure
    that supports $O(1)$-time queries of the form
``what is the largest/smallest label in the strongly connected component of $u$ in $G-x$?", for any pair of query vertices $u,x\in V$.

\end{itemize}
%

By suitably extending our 1-sensitivity oracle, we obtain a nearly optimal $2$-reachability oracle for planar digraphs, summarized as follows.

\begin{restatable}{theorem}{thmtworeach}\label{t:2reach}
  In $O(n\log^{6+o(1)}{n})$ time one can construct an $O(n\log^{3+o(1)}{n})$-space
  data structure supporting the following queries in $O(\log^{2+o(1)}{n})$ time.
  For $u,v\in V(G)$, either~find some separating vertex $x\notin\{u,v\}$ lying on all $u\to v$ paths in $G$,
  or declare $v$ 2-reachable from~$u$.
\end{restatable}

Our 1-sensitivity and $2$-reachability oracles, combined, extend several supported operations of dominator trees to the all-pairs version of the problem. For example, for any $u,v$, our data structures can identify the set $X$ of all vertices (or edges) that appear in all paths from $u$ to $v$ in $O(|X|\log^{2+o(1)} n)$ time by executing $O(|X|)$ $2$-reachability queries. That is, a $2$-reachability~query returns a vertex $x$ that appears in all paths from $u$ to $v$, and all other vertices in $X\setminus x$ (if any) appear either in all paths from $u$ to $x$ or in all paths from $x$ to $v$, which we can identify by recursively executing $2$-reachability queries from $u$ to $x$ and from $x$ to $v$, and eventually reconstruct the set $X$. 

Whereas we achieve $\Ot(n)$ preprocessing time for both our oracles, the main conceptual challenge
lies in obtaining near-optimal space. However, efficient construction of the used data~structures, especially
our generalization of the dominator tree, proved to be a highly non-trivial task as well.

\paragraph{Overview of our 1-sensitivity oracle.}
As already discussed, it is sufficient to build a 1-sensitivity oracle for single vertex failures as the case of edge failures reduces to the case of vertex failures by applying edge-splitting on the edges of the graph (i.e., replacing each edge $zw$ by a new vertex $c$, and two edges $zc$ and $cw$). 
Since $m=O(n)$ in planar graphs, this process does not increase the number of vertices or edges significantly.
The initial step of our approach is the use of the basic\footnote{Thorup~\cite{Thorup04} also presented a more involved data structure that allowed him to reduce query time
to $O(1)$ while maintaining the preprocessing time $O(n\log{n})$. However, this data structure significantly differs from his basic $O(\log{n})$-query data structure
in the fact that one needs to represent reachability through separating directed paths~$P$ ``globally'' in the entire graph $G$, as opposed to
only representing reachability through $P$ ``locally'' in the subgraph $H\subseteq G$ we recurse on. It is not clear if this more sophisticated approach can be extended to handle vertex failures.}
hierarchical decomposition approach introduced by Thorup \cite{Thorup04}. 
For the problem of constructing a reachability oracle (with no failures) this initial phase allows one to focus on the following problem, at the expense of an increase by a factor $O(\log n)$ in the preprocessing time, the size, and the query time of the constructed oracle. 
Given a graph $G$ and a directed path $P$ of $G$, construct a data structure that answers efficiently whether there is a path from $u$ to $v$ 
containing any vertex of $P$, for any two vertices $u$ and $v$. 
It is rather easy to obtain such a data structure with linear space and preprocessing time that can answer the required queries in constant time.

Although we use the decomposition phase of~\cite{Thorup04} as an initial step in our approach, the main difficulty in our problem is to build a data structure that can efficiently answer whether there exists a path from $u$ to $v$ in $G-x$ that uses a vertex of a path $P$. 
In the presence of failed vertices this becomes much more challenging, compared to reachability queries with no failures, as the set of vertices of~$P$ that are reachable from (or can reached by) a vertex $w$, might be different under failures of different vertices.
Additionally, the case where the failed vertex appears on $P$ disconnects the path into two subpaths which we need to query. 
We cannot afford to simply preprocess all such subpaths, as there can be as many as $|P|=\Theta(n)$ 
of those
for all possible failures of vertices on the path~$P$.
We overcome these problems by using new insights, developing new supporting data structures
and further exploiting planarity.

We distinguish two cases depending on whether the failed vertex appears on $P$ or not. For each path $P$, we preprocess the graph to handle each case separately. 

To deal with the case where the failed vertex lies outside of $P$, we 
identify, for the query vertices $u,v$ and failed vertex $x$, the earliest (resp., latest) vertex on $P$ that $u$ can reach (resp., that can reach $v$) in $G-x$. 
Call this vertex $first^P_{G-x}(u)$ (resp., $last^P_{G-x}(v)$).
Given these vertices it suffices to test whether $first^P_{G-x}(u)$ appears no later than $last^P_{G-x}(v)$ on $P$.
Recall from our previous discussion that both these vertices depend on the failed vertex $x$.
A useful notion throughout the paper is the following. A path $Q$ between any two vertices $w,z$ is called a \emph{satellite} path (with respect to $P$) if no vertex of $Q$ other than $w,z$ is a vertex of $P$, i.e., if $V(Q)\cap V(P) \subseteq \{w,z\}$.
On a very high level, we first develop a near-optimal data structure that can identify
in constant time for each vertex $v$ the latest vertex $v'\in V(P)$ that has a satellite path to $v$ in $G-x$.
The performance of the data structure relies on the efficient constructions of
dominator trees~\cite{domin:ahlt,dominators:bgkrtw,dominators:Fraczak2013,dominators:poset} and their properties, as well as dynamic orthogonal range-searching data structures~\cite{ChanT17}.
Given $v'$, we then show that
$last^P_{G-x}(u)$ is the latest vertex on $P$ that is in the same SCC as $v'$ in $G-x$.
We generalize the problem of computing such a vertex to the mentioned problem
of efficiently finding a maximum-labeled vertex
in the SCC of $v$ in $G-x$, where $v,x\in V$ are query parameters.
For this problem we develop a near-optimal data structure with $O(1)$ query time.
Finally, we proceed with computing $first^P_{G-x}(u)$ analogously.
The query time in this case is $O(1)$.

In order to handle the case when the failed vertex $x$ is on $P$ we further exploit
planarity. We observe that by modifying 
the basic recursive decomposition
of Thorup to use fundamental cycle separators instead of root path separators (this modification was
previously used in e.g. \cite{AbrahamCG12, MozesS18}), we can
assume that the endpoints of $P$ in fact lie on a single face of $G$.
This additional assumption enables us to achieve two important things.
First we show that after linear preprocessing, in $O(\log{n})$ time we
can in fact compute the earliest (latest) vertex of \emph{any subpath of $P$} reachable from (that can reach)
a query vertex $v\in V$ by a satellite path. Here the subpath of interest
is also a query parameter.
Moreover, we introduce a concept of a detour of $x$ to be a path that starts earlier and ends later than $x$ on $P$.
A minimal detour of $x$ is a detour that does not simultaneously start earlier (on $P$) and end later (on $P$) that any other detour of $x$.
We use planarity to show that there can be at most two significantly different
minimal detours of any vertex $x\in V(P)$.
Consequently, we show a linear time algorithm for finding the two minimal detours
for each vertex $x\in V(P)$. 
Finally, we consider several possible scenarios of how the requested
$u\to v$ path in $G-x$ can interact with $P$ and $x$.
In all of these cases we show that there exists a certain canonical
path consisting of $O(1)$ subpaths that are either satellite paths,
minimal detours, or subpaths of $P$.
This allows us to test for existence of a $u\to v$ path in $G-x$
with only $O(1)$ queries to the obtained data structures.

Even though we are not able to reduce the query time in the case when $x$ lies on the path $P$
to constant, this turns out not to be a problem.
This is because in order to answer a ``global'' query, we need only one such query and $O(\log{n})$ constant-time queries
to the data structures when $x$ is out of the path.
It follows that the query time of the whole reachability data structure is $O(\log{n})$.

\paragraph{2-reachability oracle.} The 2-reachability oracle is obtained by both extending and reusing the 1-sensitivity oracle.
It is known that if the graph is strongly connected, then checking whether vertex $v$ is $2$-reachable
from $u$ can be reduced to testing whether $v$ is reachable from $u$
under only $O(1)$ single-vertex failures,
which are easily computable from the dominator tree from an arbitrary vertex of the graph~\cite{georgiadis2017AllPairsFTReachability}.
This observation alone would imply a 2-reachability oracle
for strongly connected planar digraphs within the time/space
bounds of our 1-sensitivity oracle.

However, for graphs with $k$ strongly connected components, a generalization of this
seems to require information from as many as $\Theta(k)$ dominator trees to cover
all possible $(u,v)$ query pairs.

Nevertheless, we manage to overcome this problem by using
the same recursive approach, and carefully
developing the ``existential'' analog of the 1-sensitivity
data structures handling failures either outside the separating path $P$,
or on the separating path $P$ when $P$'s endpoints 
lie on a single face of the graph.
There is a subtle difference though; in the 2-reachability
oracle the recursive call is made only when $P$ contains no
vertex that lies on a $u\to v$ path; once we find a path $P$ containing a vertex
that lies on any $u\to v$ path we make no further recursive calls.

In the case of failures outside the path $P$, we prove that $O(1)$ single-failure
queries are sufficient to decide whether there exists
$x\notin V(P)$ that destroys all $u\to v$ paths in $G$.
Even though we use $\Theta(n)$ dominator trees to encode the information about which single-failure queries we
should issue to the $1$-sensitivity oracle (for any pair $u,v$ of query vertices), our construction guarantees that the total
size of these dominator trees is $O(n)$.

When searching for vertices $x\in V(P)$ that lie on all $u\to v$ paths
in $G$, the reduction to asking few 1-sensitivity queries does not to work. 
Instead, we take a substantially different approach.
Roughly speaking, we simulate the single-failure query
procedure developed in the 1-sensitivity oracle
for all failing $x\in V(P)$ at once.
We prove that deciding if for any such $x$ the query to the 1-sensitivity oracle
would return false can be reduced to a
generalization of a 4-dimensional orthogonal range reporting problem,
where the topology of one of the dimensions is a tree as opposed to a line.
A simple application of heavy-path decomposition~\cite{SleatorT83}
allows us to reduce this problem to the standard 4-d
orthogonal range reporting problem~\cite{KarpinskiN09} at the cost
of $O(\log{n})$-factor slowdown in the query time compared to the standard case.
This turns out to be the decisive factor in the $O(\log^{2+o(1)} n)$ query time and
$O(n\log^{3+o(1)}{n})$ space usage of our 2-reachability oracle.
\paragraph{Organization of the paper.}
In Section~\ref{s:preliminaries} we fix the notation and recall some important properties of planar graphs.
In Section~\ref{s:thorup} we give a quite detailed overview of Thorup's construction
and explain how we modify it to suit our needs.
In Section~\ref{s:failures-overview} we show how to make the reachability data structure
from Section~\ref{s:thorup} to support vertex failures. Apart from that,
in Section~\ref{s:failures-overview} we also state and explain the usage of our main technical contributions -- Theorems~\ref{t:ds1}, \ref{t:ds2}~and~\ref{thm:find-maxID-in-SCC} --
and give a more detailed overview of how they are achieved.
The detailed proofs of these theorems can be found in Sections~\ref{s:ds1}, \ref{s:ds2}~and~\ref{s:max-scc-failure}, respectively.
In Section~\ref{s:dominators} we review some useful properties of dominator trees.
The 2-reachability data structure is covered in Section~\ref{s:2reach}.

\smallskip

\section{Preliminaries}\label{s:preliminaries}
In this paper we deal with \emph{directed} simple graphs (\emph{digraphs}).
We often deal with multiple different graphs at once.
For a graph $G$, we let $V(G)$ and $E(G)$ denote the vertex and edge set of $G$, respectively.
If $G_1,G_2$ are two graphs, then $G_1\cup G_2=(V(G_1)\cup V(G_2),E(G_1)\cup E(G_2))$
and $G_1\cap G_2=(V(G_1)\cap V(G_2),E(G_1)\cap E(G_2))$.
Even though we work with digraphs, some notions~that we
use, such as connected components, or spanning trees, are only
defined for undirected graphs.
Whenever we use these notions with respect to a digraph,
we ignore the directions of the edges.

Let $G=(V,E)$ be a digraph.
We denote by $uv\in E$ the edge from $u$ to $v$ in $G$.
A graph $G'$ is called a subgraph of $G$ if $V(G')\subseteq V(G)$ and $E(G')\subseteq E(G)$.
For $S\subseteq V(G)$, we denote by $G[S]$ the \emph{induced subgraph} $(S,\{uv: uv\in E(G), \{u,v\}\subseteq S\})$.
Given a digraph $G$, we denote by $G^R$ the digraph with the same set of vertices as $G$ and with all of the edges reversed compared to the orientation of the corresponding edges in $G$. That is, if $G$ contains an edge $uv$, $G^R$  contains an edge $vu$ and vice versa. We say that $G^R$ is the \emph{reverse graph} of $G$.
For any $X\subseteq V$ we define $G-X=G[V\setminus X]$.
For $x\in V$, we write $G-x$ instead of $G-\{x\}$.

A \emph{path} $P\subseteq G$ is a subgraph whose edges $E(P)$ can be
ordered $e_1,\ldots,e_k$ such that if $e_i=u_iv_i$, then
for $i=2,\ldots,k$ we have $u_i=v_{i-1}$.
Such $P$ is also called a $u_1\to v_k$ path.
We also sometimes write $P=u_1u_2\ldots u_k$.
A path $P$ is \emph{simple} if $u_i\neq u_j$ for $i\neq j$.
For a simple path $P=u_1\ldots u_k$ we define an order
$\prec_P$ on the vertices of $P$.
We write $v\prec_P w$ for $v=u_i$ and $w=u_j$ if $i<j$.
If $u\prec_P v$, we also say that $u$ is \emph{earlier} on $P$
than $v$, whereas $v$ is \emph{later} than $u$ on $P$.
We denote by $P[u_i,u_j]$ the unique subpath of $P$ from $u_i$ to $u_j$.
Similarly, let $P(u_i,u_j)$ be the subpath of $P$ from
the vertex following $u_i$ to the vertex preceding $u_j$ on $P$.
For $u,v\in V$, we say that $v$ is \emph{reachable} from $u$
if there exists a $u\to v$ path in $G$.
We call $v$ \emph{2-reachable} from $u$
if there exist two internally vertex-disjoint $u\to v$ paths in $G$.
$u$ and $v$ are \emph{strongly connected} if there exist
both paths $u\to v$ and $v\to u$ in $G$.
A \emph{non-oriented path} $P'\subseteq G$ is a subgraph that would
become a path if we changed the directions of some of its edges.
If $P_1=u\to v$ is a (potentially non-oriented) path and $P_2=v\to w$
is a (non-oriented) path, their concatenation $P_1P_2=u\to w$ is also a (non-oriented)~path.

We sometimes use trees, which can be rooted or unrooted.
If a tree $T$ is rooted, we denote by $T[v]$ the \emph{subtree} of $T$
rooted in one of its vertices $v$.
We denote by $T[u,v]$  the path between $u$ and $v$ on $T$, by $T(u,v]$ (resp., $T[u,v)$) the path between $u$ and $v$ on $T$, excluding~$u$ (resp., excluding~$v$).
Analogously, we use $T(u,v)$ to denote the path between $u$ and $v$ on $T$, excluding $u$ and $v$.
For any tree $T$, we use the notation $t(v)$ to refer to the parent of node $v$ in $T$. If $v$ is the root of the tree, then $t(v)=v$. 
To avoid cumbersome notation we sometimes write $w\in T[v]$ when
we formally mean $w\in V(T[v])$, $w\in T[u,v]$ when $w\in V(T[u,v])$ and so on.
If $T\subseteq G$ is a spanning tree of a connected graph $G$, then for any $uv\in E(G)\setminus E(T)$ the \emph{fundamental cycle} of $uv$
wrt. $T$
is a subgraph of $G$ that consists of the unique non-oriented simple $u\to v$
path in $T$ and the edge~$uv$.
\paragraph{Plane graphs.}
A \emph{plane embedding} of a graph is a mapping of its vertices to distinct points
and of its edges to non-crossing curves in the plane.
We say that $G$ is \emph{plane}
if some embedding of $G$ is assumed.
A \emph{face} of a connected plane $G$ is a maximal open connected set of points that are not
in the image of any vertex or edge in the embedding of $G$.
There is exactly one \emph{unbounded} face.
The \emph{bounding cycle} of a bounded (unbounded, respectively)
face $f$ is a sequence of edges bounding~$f$
in clockwise (counterclockwise, respectively) order. Here, we ignore the directions of edges.
An edge can appear in a bounding cycle at most twice.
An embedding of a planar graph (along with the bounding cycles
of all faces) can be found in linear time~\cite{HopcroftT74}.
A plane graph $G$ is \emph{triangulated} if all its faces' bounding cycles
consist of $3$ edges.
Given the bounding cycles of all faces, a plane graph can be triangulated by adding edges
inside its faces in linear time.

A graph $G'$ is called a \emph{minor} of $G$ if it can be obtained
from $G$ by performing 
a sequence of edge deletions, edge contractions, and vertex deletions.
If $G$ is planar then $G'$ is planar as well.
By the Jordan Curve Theorem, a simple closed curve $C$ partitions
$\mathbb{R}^2\setminus C$ into two connected
regions, a bounded one $B$ and an unbounded one $U$.
We say that a set~of points $P$ is \emph{strictly inside} (\emph{strictly outside}) $C$
if and only if $P\subseteq B$ ($P\subseteq U$, respectively).
$P$ is \emph{weakly inside} (\emph{weakly outside}) iff $P\subseteq B\cup C$
($P\subseteq U\cup C$, respectively).
If $G$ is plane then the fundamental cycle of $uv$ corresponds
to a simple closed curve in the plane.
We often identify the fundamental cycle with this curve.

\begin{lemma}[e.g., \cite{Kleinbook}]\label{l:cycle-separator}
  Let $G=(V,E)$ be a connected triangulated plane graph with $n$ vertices. Let $T$ be a spanning tree of $G$.
  Let $w:V\to \mathbb{R}_{\geq 0}$ be some assignment of weights to
  the vertices of~$G$. Set $W:=\sum_{v\in V} w(v)$.
  Suppose that for each $v\in V$ we have $w(v)\leq \frac{1}{4}W$.

  There exists such $uv\in E\setminus E(T)$ that the total weight
  of vertices of $G$ lying strictly on one side of the fundamental cycle of $uv$ wrt. $T$
  is at most $\frac{3}{4}W$.
  The edge $uv$ can be found in linear time.
\end{lemma}

\section{The Reachability Oracle by Thorup}\label{s:thorup}
 In this section we describe the basic reachability oracle of Thorup~\cite{Thorup04}
 that we will subsequently extend to support single vertex failures.
 His result can be summarized as follows.
\begin{theorem}[\cite{Thorup04}]\label{t:reach}
   Let $G$ be a directed planar graph. One can preprocess $G$ in $O(n\log{n})$
   time so that arbitrary reachability queries are supported in $O(\log{n})$ time.
\end{theorem}

\begin{definition}[\cite{Thorup04}]
  A $2$-layered spanning tree $T$ of a digraph $H$
  is a rooted spanning tree such that
  any non-oriented path in $T$ from the root is a concatenation
  of at most two directed paths in~$H$.
\end{definition}

We will operate on graphs with some \emph{suppressed vertices}. Those suppressed vertices
will guarantee certain useful topological properties of our
plane graphs, but will otherwise be forbidden to be used from the point
of view of reachability.
In other words, when answering reachability queries we will only care about directed paths that do not go
through suppressed vertices.

\begin{remark}
Our description differs from that of Thorup in the fact that we allow
$O(1)$ suppressed vertices (where Thorup needed only one that was additionally
always the root of the spanning tree, and no balancing of suppressed
  vertices was needed),
and we insist on using simple cycle separators (whereas Thorup's
separators consisted of two root paths).
Whereas using cycle separators does not make any difference
for reachability, we rely on them later when we handle single vertex failures.
\end{remark}

\newcommand{\idx}{\imath}
\newcommand{\supr}{A}

\begin{lemma}[\cite{Thorup04}]\label{l:layered}
  Let $G=(V,E)$ be a connected digraph.
  In linear time we can construct digraphs
  $G_0,\ldots,G_{k-1}$, where $G_i=(V_i,E_i)$, and $\supr_i\subseteq V_i$ is a set of suppressed vertices, $|\supr_i|\leq 1$, and their respective
  spanning trees $T_i$, and a function $\idx:V\to \{0,\ldots,k-1\}$ such that:
  \begin{enumerate}[itemsep=-1pt]
    \item The total number of edges and vertices in all $G_i$ is linear
      in $|V|+|E|$.
    \item Each $G_i$ is a minor of $G$, and $G_i-\supr_i$ is a subgraph of $G$.
    \item For any $u,v\in V$ and any directed path $P=u\to v$, $P\subseteq G$ if and only
      if $P\subseteq G_{\idx(u)}-\supr_{\idx(u)}$ or $P\subseteq G_{\idx(u)-1}-\supr_{\idx(u)-1}$.
    \item Each spanning tree $T_i$ is $2$-layered.
  \end{enumerate}
\end{lemma}

The basic reachability data structure of Thorup~\cite{Thorup04} can be described
as follows.
Observe that by Lemma~\ref{l:layered} 
$v$ is reachable from $u$ in $G$ (where $u,v\in V$)
if and only if it is reachable from $u$ in either $G_{\idx(u)}-\supr_{\idx(u)}$ or $G_{\idx(u)-1}-\supr_{\idx(u)-1}$.
Moreover, all the graphs $G_i$ in Lemma~\ref{l:layered}, being minors of $G$, are planar as well.
Therefore, by applying Lemma~\ref{l:layered}, the problem
is reduced to the case when (1) a planar graph $G=(V,E)$ has a 2-layered spanning tree $T$,
(2) we are only interested in reachability without going through
some set of $O(1)$ (in fact, at most $5$, as we will see) suppressed vertices $\supr$ of $G$.
Under these assumptions, the problem is solved recursively as follows.

If $G$ has constant size, or has only suppressed vertices,
we compute its transitive closure 
so that queries are answered in $O(1)$ time.
Otherwise, we apply recursion.
We first temporarily triangulate $G$ by adding edges and obtain graph $G^\Delta$.
Note that $T$ is a 2-layered spanning tree of $G^\Delta$ as well.

\begin{fact}\label{f:contract}
  Let $G=(V,E)$ be a connected plane digraph. Let $T$ be a 2-layered spanning tree of $G$
  rooted at $r$. Let $uv \in E\setminus E(T)$. 
  Let $V_1$ (resp., $V_2$) be the subset of $V$
  strictly inside (resp., strictly outside) the fundamental cycle $C_{uv}$ of $uv$ wrt. $T$.
  Let $S_{uv}$ be the non-oriented path $C_{uv}-uv$.

  Let $G'$ ($T'$) be obtained from $G$ ($T$) by contracting (the edges of) $S_{uv}$
  into a single vertex $r'$.
  Then for $i=1,2$, $T'[V_i\cup \{r'\}]$ is a 2-layered spanning tree of $G'[V_i\cup \{r'\}]$.

\end{fact}

Using Lemma~\ref{l:cycle-separator}, we compute in linear time a balanced fundamental cycle
separator $C_{ab}$, where $ab\in E(G^\Delta)\setminus E(T)$.
Let $A$ be the set of suppressed vertices of $G$.
The weights assigned to vertices depend on the size of $\supr$:
if $|\supr|\leq 4$ then we assign weights $1$ uniformly to all vertices of $G$,
and otherwise we assign unit weights to the vertices of
$\supr$ only (the remaining vertices $V$ get weight $0$).
Let the non-oriented path $S_{ab}=C_{ab}-ab$ be called \emph{the separator}.
Let $V_1,V_2,G',T'$ and $r'$ be defined as in Fact~\ref{f:contract}.

Note that for any $u,v\in V$, a $u\to v$ path in $G-\supr$
can either go through a vertex of $S_{ab}$, or is entirely contained
in exactly one $G[V_i]-\supr$, for which $\{u,v\}\subseteq V_i$ holds.
We deal with these two cases separately.
Since $G[V_i]\subseteq G'[V_i\cup\{r'\}]$, queries about a $u\to v$ path not going through $S_{ab}$ can be delegated to 
the data structures built recursively 
on each $G'[V_i\cup\{r'\}]$ with suppressed set $\supr_i=(\supr\cap V_i)\cup \{r'\}$.
By Fact~\ref{f:contract}, $G'[V_i\cup\{r'\}]$ has a 2-layered
spanning tree $T'[V_i\cup \{r'\}]$ rooted in $r'$.
Moreover, a path $u\to v$ exists in $G[V_i]-\supr$
if and only if a path $u\to v$ not going through a suppressed
set $\supr_i$ exists in $G'[V_i\cup \{r'\}]$.
Note that since $S_{ab}$ does not necessarily go through any
of the vertices of~$\supr$, $\supr_i$ might be larger than $\supr$
by a single element -- this is why balancing of suppressed vertices is needed.
Hence, indeed recursion can be applied in this case.

Paths $u\to v$ going through $V(S_{ab})$ in $G-\supr$ are in turn handled as follows.
Since $T$ is 2-layered, $S_{ab}$ can be decomposed
into at most 4 edge-disjoint directed paths in $G$.
Consequently, $S_{ab}-\supr$ can be split
into at most $4+|\supr|$ edge-disjoint directed
paths in $G-\supr$.
Next, we take advantage of the following lemma.

\begin{lemma}[\cite{Subramanian93, Thorup04}]\label{l:pathreach}
  Let $H$ be a directed graph and let $P\subseteq H$ be a simple directed path.
  Then, in linear time we can build a data structure that supports constant-time queries
  about the existence of a $u\to v$ path that necessarily goes through $V(P)$.
\end{lemma}
The above lemma is based on the following simple fact that will prove useful later on.

\newcommand{\vf}{\ensuremath{first}}
\newcommand{\vl}{\ensuremath{last}}
\newcommand{\vfs}{\ensuremath{first^*}}
\newcommand{\vls}{\ensuremath{last^*}}

\begin{fact}[\cite{Subramanian93, Thorup04}]\label{f:pathreach}
  Let $H$ be a digraph. Let $P\subseteq H$ be a simple directed path.
  Denote by $\vf_H^P(u)$ the earliest vertex of $P$ \emph{reachable from}
  $u$ in $H$. Denote by $\vl_H^P(v)$ the latest
  vertex of $P$ that \emph{can reach} $v$ in $H$.
  Then there exists a $u\to v$ path in $H$ going through $V(P)$
  iff $\vf_H^P(u)\preceq_P \vl_H^P(v)$.
\end{fact}
Consequently, by building at most $4+|\supr|$ data structures of Lemma~\ref{l:pathreach},
we can check whether $v$ is reachable from $u$ through $V(S_{ab})$ in $G-\supr$
in $O(1+|\supr|)$ time.

\begin{restatable}{lemma}{suppressedset}
At each recursive call, the size of the suppressed set $\supr$ is at most $5$.
\end{restatable}
\begin{proof}
  We apply induction on the level of the recursive call
  that is passed graph $G$ with suppressed set $\supr$.
  In the first recursive call (level $0$) after applying Lemma~\ref{l:layered},
  $\supr$ has size at most $1$.
  
  Let the level of the recursive call be at least $1$.
  In the parent call the suppressed set $\supr'$ had size at most $5$.
  If $|\supr'|=5$, then the balanced cycle separator $C_{ab}$ put at most $3$
  vertices of $\supr'$ strictly on one side of $C_{ab}$.
  Since the entire set $V(C_{ab})$ was subsequently contracted
  into a single vertex that got suppressed, $\supr$ has size at most $3+1=4$.
  
  On the other hand, if $|\supr'|\leq 4$, then clearly $|\supr|\leq |\supr'|+1\leq 5$.
\end{proof}

\begin{restatable}{lemma}{recursionlevels}\label{l:levels}
  The depth of the recursion is $O(\log{n})$.
\end{restatable}
\begin{proof}
  If the cycle separator computed at the parent call balanced all the vertices
  of the graph $G'$ with suppressed set $A'$ and $n$ vertices, then $G$ has at most $\frac{3}{4}n+1$
  vertices.

  However, if $|A'|=5$, then $G$ might have nearly $n$ vertices.
  But in this case $|A|\leq 4$, so the in the children recursive
  calls the graphs will certainly have no more than $\frac{3}{4}n$ vertices.

  This proves that at each recursive call the number of vertices is
  reduced by a constant factor compared to the grandparent call,
  i.e., the depth of the recursion tree is $O(\log{n})$.
\end{proof}

At each recursive call of the data structure's construction procedure we use
only linear preprocessing time.
Observe that at each recursive level the total number of vertices
in all graphs of that level is linear: there are at most $n$ vertices
that are not suppressed and each of these vertices resides in a unique
graph of that level.
Since each graph has to have at least one non-suppressed vertex,
the number of graphs on that level is also at most $n$. 
Therefore, given that $|A|=O(1)$ in every recursive call,
the sum of sizes of the graphs on that level is $O(n)$.
By Lemma~\ref{l:levels}, we conclude that the total time spent
in preprocessing is $O(n\log{n})$.

In order to answer a query whether an $u\to v$ path exists, we
only need to query $O(\log{n})$ data structures of Lemma~\ref{l:pathreach} handling queries
about reachability through a directed path.

\section{Reachability Under Failures}\label{s:failures-overview}

In this section we explain how to modify
the data structure
discussed in Section~\ref{s:thorup} to support~queries
of the form ``is $v$ reachable from $u$ in $G-x$?'', where
$u,v,x\in V$ are distinct query parameters.

First recall that, by Lemma~\ref{l:layered}, each path $P=u\to v$
exists in $G$ if and only if it exists in either $G_{\idx(u)}-A_{\idx(u)}$
or $G_{\idx(u)-1}-A_{\idx(u)-1}$.
This, in particular, applies to paths $P$ avoiding $x$.
As a result, $v$ is reachable from $u$ in $G-x$ if and
only if $v$ is reachable from $u$ in either $G_{\idx(u)}-A_{\idx(u)}-x$
or $G_{\idx(u)-1}-A_{\idx(u)-1}-x$.
That being said, we can again concentrate on the case
when $G$ has a 2-layered spanning tree
and we only care about reachability in $G-A$, where $A\subseteq V$ has size $O(1)$.

We follow the recursive approach 
of Section~\ref{s:thorup}.
The only difference lies in handling paths in $G-A$ going
through the separator $S_{ab}$.
Ideally, we would like to generalize the data structure of
Lemma~\ref{l:pathreach} so that
single-vertex failures are supported.
However, it is not clear how to do it in full generality.
Instead, we show two separate data structures, which, when combined,
are powerful enough to handle paths going through a \emph{cycle} separator.

Recall that $S_{ab}-A$ can be decomposed into $O(1)$ simple directed
paths $P_1,\ldots,P_k$ in $G-A$ that can only share endpoints.
Denote by $D\subseteq V(S_{ab})$ the set of endpoints of these paths.
Suppose we want to compute whether
there exists a $u\to v$ path~$Q$ in $G-A-x$ that additionally
goes through some vertex of $S_{ab}$.
We distinguish several cases.
  
  \textbf{1.} Suppose that $x\in D$. Recall that there exist only
    $O(1)$ such vertices $x$.
    Moreover, $S_{ab}-A-x$ can be decomposed into $O(1)$ simple paths
    in $G-A-x$.
    Hence, for each such $x$, we build $O(1)$ data structures
    of Lemma~\ref{l:pathreach} to handle reachability queries
    through $S_{ab}$ in $G-A-x$ exactly as was done in Section~\ref{s:thorup}.
    The preprocessing is clearly linear and the query time
    is $O(1)$ in this case.

\textbf{2.} Now suppose $x\notin D$ and there exists such $P_i$
that $V(Q)\cap V(P_i)\neq \emptyset$ and $x\notin V(P_i)$.
Paths of this kind are handled using the following theorem (for $G:=G-A$, $P:=P_i$) proved in Section~\ref{s:ds1}.
\begin{restatable}{theorem}{thmdsone}\label{t:ds1}
  Let $G$ be a digraph and let $P\subseteq G$ be a simple directed path.
  In $O\left(m\frac{\log{n}}{\log\log{n}}\right)$ time one can build a linear-space data structure that can decide, for any $u,v\in V$, and $x\notin V(P)$, if there exists a $u\to v$ path going through
  $V(P)$ in $G-x$.
  Such queries are answered in $O(1)$ time.
\end{restatable}

\textbf{3.} The last remaining case is when none of the above cases apply.
This means that $x\notin D$ and for all $P_i$ either
$Q$ does not go through $V(P_i)$ or $x\in V(P_i)$.
Since $V(P_i)\cap V(P_j)\subseteq D$ for all $j\neq i$, there
can be at most one such $i$ that $x\in V(P_i)$.
For all $j\neq i$, $Q$ does not go through $V(P_j)$.
On the other hand, since $Q$ goes through $S_{ab}$, it in fact has to
go through $V(P_i)$.

\begin{restatable}{lemma}{singleface}\label{l:singleface}
  Let $\bar{V_i}=V\setminus V(S_{ab})\cup V(P_i)$.
  Then the endpoints of $P_i$ lie on a single face of the component
  of $G[\bar{V_i}]-A$ that contains $P_i$.
\end{restatable}
By Lemma~\ref{l:singleface}, the last case
can be handled using the following theorem proved in Section~\ref{s:ds2}.
\begin{proof}
Since $C_{ab}$ is a cycle separator, 
we can also view $C_{ab}$ as a closed curve
that does not cross any edges (but can follow the edges)
  of $G-A$ and has (the embedding of) each $P_i$ as a contiguous part.
  After removing vertices $V(S_{ab})\setminus V(P_i)$
  from $G-A$, the part of the curve that does
  not correspond to $P_i$ connects the endpoints
  of $P_i$ and does not intersect
  (the embedding of) any other vertices or edges of $G[\bar{V_i}]-A$.
  Hence, this part lies inside
  a single face of the component of $G[\bar{V_i}]-A$
  that contains $P_i$.
  We conclude that the endpoints of $P_i$
  lie on that single face.
\end{proof}

\begin{restatable}{theorem}{thmdstwo}\label{t:ds2}
  Let $G$ be a plane digraph and let $P\subseteq G$ be a simple path whose endpoints lie on a single face of $G$.
  In linear time one can build a data structure that can compute in $O(\log{n})$ time,
  whether there exists a $u\to v$
  path going through $V(P)$ in $G-x$, where $u,v\in V$ and $x\in V(P)$.
\end{restatable}

In order to prove Theorem~\ref{t:mainresult}, first note that
the preprocessing time of our data structure is 
$O(n\log^2{n}/\log\log{n})$, since all the data structures of Lemma~\ref{l:pathreach}, and Theorems~\ref{t:ds1}~and~\ref{t:ds2}
can be constructed in $O(n\log{n}/\log\log{n})$ time for each
input graph of the recursion.
Since all these data structures use only linear space, the total space used is $O(n\log{n})$.
To analyze the query time, observe that similarly
to Section~\ref{s:thorup}, we only query $O(\log{n})$
``reachability through a path'' data structures.
All of them, except of the data structure of Theorem~\ref{t:ds2},
have constant query time.
However, the data structure of Theorem~\ref{t:ds2}
is only used to handle the case when $x\in V(S_{ab})$.
Observe that for any $w\in V$ there is no more than one
node in the recursive data structure such that
$w$ is a vertex of the respective separator $S_{ab}$.
As a result, we only need to perform a single
query to a data structure of Theorem~\ref{t:ds2}.

\newcommand{\vlsp}{\ensuremath{last^{*P}}}
\newcommand{\vfsp}{\ensuremath{first^{*P}}}
\paragraph{Overview of the ``failed vertex out of $P$'' data structure of Theorem~\ref{t:ds1}.}
To deal with this case, we identify, for the query vertices $u,v$ and failed vertex $x$,
the earliest (latest) vertex $\vf^P_{G-x}(u)$ ($\vl^P_{G-x}(v)$, resp.) on $P$ that $u$ can reach (that can reach $v$, resp.) in $G-x$. 
Given those vertices, by Fact~\ref{f:pathreach}, it suffices to test whether $\vf^P_{G-x}(u)\preceq_P \vl^P_{G-x}(v)$.
However, it comes at no surprise that $\vf^P_{G-x}(u)$ and $\vl^P_{G-x}(v)$ depend on the failed vertex $x$.
In the remaining part of this overview we only consider computing $\vl^P_{G-x}(v)$ (computing $\vf^P_{G-x}(u)$ is symmetrical).

A path $Q$ between any two $w,z\in V$ is called a \emph{satellite} path (with respect to $P$) if no vertex of $Q$ other than $w,z$ is a vertex of $P$. 
Let $\vlsp_{G-x}(v)$ denote the latest vertex on $P$ that can reach $v$
in $G-x$ using a satellite path.
The first step to find $\vl^P_{G-x}(v)$ is to efficiently compute $\vlsp_{G-x}(v)$.

Let $P=p_1\ldots p_\ell$.
For each $p_i$ we first identify all the vertices $w$ for which $p_i$ is the latest vertex on $P$ that can reach $w$ with a satellite path, 
which we refer to as the $i$-th layer $L_i$.
We also compute all the possible failed vertices $x$ that can fail all paths from
$p_i$ to $w$ in $G[L_i\cup \{p_i\}]$; let us call those vertices the $(p_i,w)$-cut-points. 
We use dominator trees for a compact tree representation of the $(p_i,w)$-cut-points for all vertices $w\in L_i$.
While the division of $V\setminus V(P)$ into layers alone
allows us to compute $\vlsp_{G-x}(v)$ in several possible query cases, the difficult case is when the failed vertex~$x$
is one of the $(p_i,v)$-cut-points. 
We prove that this case can be solved by (1) identifying, during the preprocessing,
for each possible $(p_i,w)$-cut-point $c$, the latest vertex $p_j$ of $P$, $p_j\prec p_i$,
that has a satellite path to $c$ avoiding all $(p_i,c)$-cut-points 
(we call $p_j$ the rescuer of $c$),
and (2) computing (during the query) the latest rescuer of all $(p_i,v)$-cut-points $c$
that are not $(p_i,x)$-cut-points.

The main challenge in (1) is efficient preprocessing. We cannot afford
to execute traversals for identifying the existence of such paths.
To overcome this difficulty, we leverage the properties of the dominator tree used for representing all $(p_i,w)$-cut-points,
combined with a use of a dynamic two-dimensional point reporting data structure~\cite{ChanT17}.
The combination of those data structures allows us to compute the rescuers for all vertices in $O(n \log n / \log \log n)$ total time.
To efficiently perform step (2) during the query, it is enough to
preprocess the dominator tree of $G[L_i\cup\{p_i\}]$ with vertices
labeled with their rescuers and find the maximum labeled vertices
on a $x$-to-$v$ path in that tree.
This can be done in $O(1)$-time after linear preprocessing~\cite{DemaineLW14}.

Finally, we show that $\vl^P_{G-x}(v)$ is the latest vertex of $P$ that is strongly connected to $\vlsp_{G-x}(v)$
in $G-x$.
Given $\vlsp_{G-x}(v)$, one could identify $\vl^P_{G-x}(v)$ in $O(\log |V(P)|)$ time with a binary search on $p_k\in P$ by answering efficiently queries
``are $p_k$ and $\vlsp_{G-x}(v)$ strongly connected in $G-x$?" using the data structure of \cite{GeorgiadisIP17}.
However, this results in $O(\log n)$ query time for each path $P$, while we are aiming at constant query time per path.
We improve this bound by showing in Section~\ref{s:max-scc-failure} that such queries can be answered in $O(1)$ time,
by means of the following non-trivial extension of the framework in \cite{GeorgiadisIP17},
which we believe might be of independent interest.

\begin{restatable}{theorem}{thmfindmax}\label{thm:find-maxID-in-SCC}
  Given a digraph $G$ and an assignment $f:V\to \mathbb{R}$ of labels to the vertices of $G$, we can preprocess $G$ in $O(m+n(\log n \log \log n)^{2/3})$ time, so that the following queries are supported in $O(1)$ time. 
  Given $x,v\in V(G)$, find a vertex with the maximum label in the SCC of $v$ in $G-x$.
\end{restatable}

\paragraph{Overview of the ``failed vertex on $P$'' data structure of Theorem~\ref{t:ds2}.}
In order to handle this case we exploit planarity extensively.
Let $P=p_1\ldots p_\ell$.
We define for each vertex $p_k=x\in V(P)$ a \emph{detour} of $x$ to be a $p_i\to p_j$ satellite
path, where $i<k<j$.
A \emph{minimal detour} of $x$ is such a detour $p_i\to p_j$ of $x$
that no $p_{i'}\to p_{j'}$ detour of $x$
such that $i\leq i'\leq j'\leq j$ and $j'-i'<j-i$ exists.
Since $p_1$ and $p_\ell$ lie on a single face of $G$
we show that there are at most two non-equivalent (we identify a detour with a pair of its endpoints)
minimal detours of any vertex $x\in V(P)$.
Moreover, we show a linear time algorithm for finding these minimal detours
for each vertex $x\in V(P)$.

Similarly as in the data structure of Theorem~\ref{t:ds1}, we would like
to construct the requested $u\to v$ path $Q$ in $G-x$ of ``pieces''
that are either satellite paths, subpaths of $P$, or paths between
strongly connected vertices.
Since $x$ necessarily lies on $P$, the sets of satellite paths in $G$ and $G-x$
are the same, i.e., the satellite paths do not depend on the failed $x$.
However,
removing any $x\in V(P)$ from the graph breaks $P$ into two parts $P[p_1,p_{k-1}],P[p_{k+1},p_\ell]$.
Thus, to be able to take advantage of Fact~\ref{f:pathreach}, it seems necessary that we efficiently
compute $\vfsp_G(u,p_1,p_{k-1})$ ($\vlsp_G(v,p_1,p_{k-1})$), defined as the
earliest (latest, resp.) vertex on $P[p_1,p_{k-1}]$ that
$u$ can reach (that can reach $v$, resp.) by a satellite path in $G$.
Analogously we need earliest (latest) vertices on the suffix subpaths of the form $P[p_{k+1},p_\ell]$
that can be reached (can reach, resp.) by query vertices.
Unfortunately, it is not clear how to preprocess a general digraph $G$ in nearly linear time
so that such queries are supported in polylogarithmic time.
To circumvent this problem, we again use the fact that $p_1$
and $p_\ell$ lie on a single face of $G$. Then cutting $G$ along
$P$ yields two plane graphs $G_1,G_2$ such that the path $P$ is a contiguous part
of a single face of both $G_1$ and $G_2$.
Since a satellite path is entirely in $G_1$ or $G_2$, we can deal with
them separately.
We show that by extending $G_i$~with~a~tree of auxiliary vertices 
and building the optimal reachability oracle~\cite{HolmRT15} on it,
we can find the vertices $\vfsp_G(u,p_a,p_b)$, $\vlsp_G(v,p_a,p_b)$ for any
$p_a,p_b$ in $O(\log{n})$ time using a variation of binary search.

Finally, given all the auxiliary data structures and tools, we show
how to answer ``is there a $u\to v$ path $Q$ in $G-x$'' queries in $O(\log{n})$ time.
To this end, we distinguish several scenarios depending on the possible interactions of $Q$ with $P$.
Let $\enter{Q}$ (resp., $\leave{Q}$) be the first (resp., last) vertex on $Q$ that appears on $P$.
We consider four cases depending on whether $\enter{Q}$ and $\leave{Q}$ appear earlier or later than $x$ on $P$.
In each case we concentrate on looking for~a~certain ``canonical'' path
that visits earliest/latest possible vertices on both sides of $x$ on $P$
and jumps between the sides of $x$ using minimal detours.
All in all, we can limit ourselves to paths consisting of $O(1)$
subpaths of a very special form that can be handled by our auxiliary data structures.

\paragraph{Finding the maximum label in an SCC under failures.}
In order to obtain the general data structure of Theorem~\ref{thm:find-maxID-in-SCC}, we exploit the framework developed in \cite{GeorgiadisIP17}. 
We extend this framework to identify the minimum/maximum label in any single SCC in $G-x$, except of the one SCC that contains an arbitrary but fixed vertex $s$. 
To deal with this case, we show how we can precompute the maximum value in the SCC of $s$ in $G-x$, for all possible failures $x$, via a reduction to a special case of offline two-dimensional range minimum queries, which can be further reduced to decremental one-dimensional range minimum queries~\cite{Wilkinson2014Amortized}.


\section{Dominators in Directed Graphs}\label{s:dominators}
In this section we review the dominator trees and their properties.

A \emph{flow graph} is a directed graph with a start vertex $s$, where all vertices are reachable from $s$.
Let $G_s$ be a flow graph 
with start vertex $s$.
A vertex $u$ is a \emph{dominator} of a vertex $v$ ($u$ \emph{dominates} $v$) if every path from $s$ to $v$ in $G_s$ contains $u$; $u$ is a \emph{proper dominator} of $v$ if $u$ dominates $v$ and $u \not= v$.
Let \emph{dom$(v)$} be the set of dominators of $v$. Clearly, \emph{dom}$(s)$ $=\{s\}$ and for any $v\neq s$ we have that $\{s,v\}\subseteq$ \emph{dom$(v)$}: we say that $s$ and $v$ are the \emph{trivial dominators} of $v$ in the flow graph $G_s$.
The dominator relation is reflexive and transitive. Its transitive reduction is a rooted tree~\cite{AU72,LM69}, known as the \emph{dominator tree} $D$: $u$ dominates $v$ if and only if $u$ is an ancestor of $v$ in $D$.
If $v \not= s$,  the parent of $v$ in $D$, denoted by $d(v)$, is the \emph{immediate dominator} of $v$: it is the unique proper dominator of $v$ that is dominated by all proper dominators of $v$. Similarly, we can define the dominator relation in the flow graph $G_s^R$, and let $D^R$ denote the dominator tree of $G_s^R$. We also denote the immediate dominator of $v$ in $G_s^R$ by $d^R(v)$.
Lengauer and Tarjan~\cite{domin:lt} presented an algorithm for computing dominators in  $O(m \alpha(m,n))$ time for a flow graph with $n$ vertices and $m$ edges, where $\alpha$ is a functional inverse of Ackermann's function~\cite{dsu:tarjan}.
Subsequently, several linear-time algorithms
were discovered~\cite{domin:ahlt,dominators:bgkrtw,dominators:Fraczak2013,dominators:poset}.
We apply the tree notation introduced earlier on for referring to
subtrees and paths to dominator trees as well.

\begin{lemma}[\cite{2VCB}]
	\label{lemma:paths-through-SAP}
	Let $G$ be a flow graph with start vertex $s$ and let $v\not=s$. 
 	Let $w$ be any vertex that is not a descendant of $v$ in $D$.
	All simple paths in $G$ from $w$ to any descendant of $v$ in $D$ must contain $v$.
\end{lemma}

\begin{lemma}[\cite{Georgiadis2015dom}]\label{lemma:parent-property}
  For each $zw\in E(G_s)$, where $z\neq s$, $z$ is a descendant of $d(w)$ in $D$.
\end{lemma}

\begin{lemma}
	\label{lemma:paths-ancestor-descendant-dom}
	Let $G$ be a flow graph with start vertex $s$ and let $v\not=s$ be a non-leaf in the dominator tree $D$ of $G$. 
	There exists a simple path in $G[D[v]]$ from $v$ to any $w \in D[v]$.
\end{lemma}
\begin{proof}
Assume there is no path from $v$ to some vertex $w\in D[v]$ in $G[D[v]]$.
First, notice that there is a path from $v$ to $w$ in $G$, since $w$ is reachable from $s$ in $G$, but not $G-v$.
Hence, all simple paths from $v$ to $w$ in $G$ contain vertices in $V\setminus D[v]$. 
Take such a path $Q$, and let $z$ be any vertex from $V\setminus D[v]$, and let $Q'$ be the subpath from $z$ to $w$.
Since $Q$ is simple and $v$ appears as the first vertex on $Q$, $v\notin V(Q')$.
By Lemma \ref{lemma:paths-through-SAP}, $Q'$ contains $v$ as $z \notin D[v]$. A contradiction. The lemma follows.
\end{proof}

\section{Proof of Theorem \ref{t:ds1}}\label{s:ds1}

In this section we provide the proof of the following theorem.
\thmdsone*
Here, we do not assume that the underlying graph $G$ is planar.
Let $n=|V(G)|$ and $m=|E(G)|$.
Wlog. assume $n\leq m$.
Let $P =  p_1 p_2\dots p_\ell$ be the directed path we consider.
To answer our queries at hand for a failing vertex $x\notin V(P)$, we 
use the same approach as in Lemma~\ref{l:pathreach}:
observe that since $x\notin V(P)$, by Fact~\ref{f:pathreach},
$u$ can reach $v$ through $P$ in $G-x$ if and only if
$\vf_{G- x}^P(u)\preceq_P \vl_{G - x}^P(v)$.
In what follows we only show how to compute $\vl_{G - x}^P(v)$ efficiently,
since $\vf_{G - x}^P(u)=\vl_{\rev{G}-x}^{\rev{P}}(u)$ and thus it
can be computed by proceeding identically on the reverse graph $G^R$.
For brevity, in the remaining part of this section we omit the superscript/subscript $P$
and write $\vl_{G-x}$ instead $\vl_{G-x}^P$, $\vls_{G}$ instead $\vlsp_G$, $\prec$ instead of $\prec_P$, etc.

\begin{definition}
	We call a simple directed path $Q=e\to f$ of $G$ \emph{satellite} 
	if it does
	not go through $V(P)$ as intermediate vertices, i.e.,
	$V(Q)\cap V(P)\subseteq \{e,f\}$.
\end{definition}

For any $v\in V\setminus V(P)$, we also denote by $\vls_{G-x}(v)$ the latest vertex 
of $P$ that can reach $v$ in $G-x$ by a satellite path, if such a vertex exists.
For $v\in V(P)$, set $\vls_{G-x}(v)=v$.
Having computed $\vls_{G-x}(v)$, we then compute $\vl_{G - x}^P(v)$ using the following lemma.

\begin{lemma}\label{lem:from-exit-to-latest}
$\vl_{G - x}(v)$ is the latest vertex of $P$ in the SCC of $\vls_{G-x}(v)$ in $G-x$. 
\end{lemma}
\begin{proof}
Since  $\vls_{G-x}(v)$ can reach $v$ in $G-x$, we have that $\vls_{G-x}(v)\preceq \vl_{G - x}(v)$.
  As $x\notin V(P)$, there exists a path $\vls_{G-x}(v)\to \vl_{G - x}(v)$ in $G-x$ following $P$.

  If $v\in V(P)$, then, by the definition, $\vl_{G-x}(v)$ can reach $v=\vls_{G-x}(v)$.
  Suppose $v\notin V(P)$ in $G-x$.
  Then take any simple path $Q=\vl_{G-x}(v)\to v$ in $G-x$. Let $r$ be the last
  vertex on $Q$ such that $r\in V(P)$.
  The subpath $r\to v$ of $Q$ is a satellite path, so $r\preceq \vls_{G-x}(v)$
  and thus there exists a $r\to \vls_{G-x}(v)$ path in $G-x$.
  Since $r$ is reachable from $\vl_{G-x}(v)$ in $G-x$, we
  conclude that there exists a $\vl_{G-x}(v)\to \vls_{G-x}(v)$ path in $G-x$.
Hence, $\vl_{G - x}(v)$ and $\vls_{G-x}(v)$ are strongly connected in $G-x$.

  Clearly, all vertices that are strongly connected with $\vls_{G-x}(v)$ can reach $v$ in $G-x$.
Therefore, $\vl_{G - x}(v)$ is the latest vertex of $P$ in the strongly connected component of $\vls_{G-x}(v)$ in $G-x$.
\end{proof}

Let us now note the following simple lemma on strong
connectivity between vertices of a path.

\begin{lemma}\label{l:path-scc}
  Let $H$ be a digraph and let $Q=q_1,\ldots,q_\ell$ be a directed path in $H$.
  Then for any $i=1,\ldots,\ell$, there exists
  two indices $a\in \{1,\ldots,\ell\}$ and $b\in \{i,\ldots,\ell\}$
such  that the only vertices of $P$ that are strongly connected
  to $q_i$ are $q_a,\ldots,q_b$.
\end{lemma}
\begin{proof}
Assume by contradiction that there exist $q_j$ and $q_k$ such
that $q_k$ lies between $q_i$ and $q_j$ on $P$,
and $q_j$ is strongly connected to $q_i$, whereas
$q_k$ is not strongly connected to $q_i$.
Without loss of generality suppose $i<j$ (the case $i>j$ is symmetric).
We have $i<k<j$.
Since $q_i$ and $q_j$ are strongly connected,
there exists a path $q_j\to q_i$ in $G$.
However, since $q_k$ lies before $q_j$ on $P$,
there exists a path $q_k\to q_j\to q_i$ in $G$.
But $q_i$ lies before $q_k$ on $P$,
so there exists a $q_i\to q_k$ path in $G$.
We conclude that $q_i$ and $q_k$ are strongly connected, a contradiction.
\end{proof}

Georgiadis et al.~\cite{GeorgiadisIP17} showed the following theorem.
\begin{theorem}[\cite{GeorgiadisIP17}]\label{t:scc-fail-query}
  Let $G$ be a digraph. In linear time one can construct a data
  structure~supporting $O(1)$-time queries of the form ``are $u$ and $v$ strongly
  connected in $G-x$?'', where $u,v,x\in V(G)$.
\end{theorem}

Hence, after linear preprocessing, by Lemma~\ref{l:path-scc} applied to $H=G$ and $Q=P$,
we could compute the vertex $\vl_{G-x}(v)$ (which, by Lemma~\ref{lem:from-exit-to-latest},
is the latest vertex of $P$ in the SCC of $\vls_{G-x}(v)\in V(P)$)
by using binary search.
Each step of binary search would take a single query to the data structure of~\cite{GeorgiadisIP17},
so computing $\vl_{G-x}(v)$ out of $\vls_{G-x}(v)$ would take $O(\log{n})$ time.

However, we can do better using Theorem~\ref{thm:find-maxID-in-SCC}.
In order to compute $\vl_{G-x}(v)$ out of $\vls_{G-x}(v)$
in constant time
using Theorem~\ref{thm:find-maxID-in-SCC}, we assign a label $f(v)$
to each vertex $v$ as follows: for each vertex $p_i\in V(P)$ we set $f(p_i)=i$.
For each vertex $w\notin V(P)$, we set $f(w)=0$.
The maximum labeled vertex in the SCC of $\vls_{G-x}(v)$
in $G-x$ is precisely $\vl_{G-x}(v)$.
Our final task is to show how to compute $\vls_{G-x}(v)$ efficiently.

\subsection{Computing $\vls_{G-x}(v)$.}\label{s:satellite}
For $i=\ell,\ldots,1$, define
the \emph{layer} $L_i$ to be the vertices of $V\setminus V(P)$
reachable from $p_i$ \emph{by a satellite path}, minus $\bigcup_{j=i+1}^{\ell} L_j$.
In other words, $L_\ell$ contains vertices not on the path $P$
that are reachable from $p_\ell$ by a satellite path,
and each subsequent layer $L_i$ contains vertices reachable from $p_i$
by a satellite path
that are not reachable from $p_{i+1},\ldots,p_\ell$ by a satellite path.
The layers $L_\ell,\ldots,L_1$ can be computed by performing
$\ell$ graph searches with starting points $p_\ell, \dots, p_1$.
The graph search never enters the vertices of $P$ (except the starting vertex)
or the vertices of previous layers.
This way, the total time needed to perform all $\ell$ graph searches is linear.

For each layer $L_i$ we also compute the dominator tree $D_i$ of $G[L_i\cup \{p_i\}]$
rooted at $p_i$. 
Denote by $d_i(w)$ the parent of $w\in L_i$ in $D_i$.
Since we have
$E(G[L_i\cup \{p_i\}]) \cap E(G[L_j\cup \{p_j\}]) = \emptyset$
for $i\not=j$, the dominator trees for all $i=1,\dots,\ell$ can be computed in linear time overall.

\newcommand{\ly}{layer}

Before proceeding with the computation, we need a few more definitions. 
Let $v\in V\setminus V(P)$.
We denote by $\ly(w)$ the index in $j\in [\ell]$ for which $w\in L_{j}$,
if it exists.
Recall that $\vls_{G}(v)$ is the latest vertex on $P$ with a satellite path to $v$ in $G$. 

Suppose $\vls_G(v)$ exists, since otherwise
$\vls_{G-x}(v)$ does not exist either.
We~distinguish two cases: (i) when $\ly(x)$ does not exist or $\ly(x)\not=\ly(v)$, and the more involved case (ii) when $\ly(x)=\ly(v)$.
We first show that case (i) is actually very easy to handle.

\begin{lemma}
  If $\ly(x)$ does not exist or $\ly(x) \not= \ly(v)$,  then $\vls_{G-x}(v)=\vls_{G}(v)$. 
\end{lemma}
\begin{proof}
  Let $l_v=\ly(v)$. By the definition of a layer, $\vls_G(v)=p_{l_v}$.
  There is a satellite path from $p_{l_v}$ to $v$ in $G[L_{l_v}\cup \{p_{l_v}\}]$, avoiding $x$, as $x\notin L_{l_v}$. 
  Consequently, $\vls_{G}(v)\preceq \vls_{G-x}(v)$.

  However, since $G-x$ is a strict subgraph of $G$, we also have $\vls_{G-x}(v)\preceq \vls_{G}(v)$.
  We conclude that indeed $\vls_{G-x}(v)=\vls_{G}(v)$.
\end{proof}

\newcommand{\und}{undom}

The most interesting case is when $\ly(x)=\ly(v)$, which we deal with as follows.
Set $l=\ly(x)=\ly(v)$.
We will exploit the dominator tree $D_l$. If $x$ does not dominate $v$ in $D_{l}$, then again one can
show that $\vls_{G-x}(v)=\vls_{G}(v)$.
Otherwise, we will show that it is enough to compute the latest vertex $p_q\in V(P)$,
for which there is a path from $p_q$ to some vertex of $D_{l}[x,v]$ in $G-x$.
In order to efficiently compute the appropriate vertices $p_q$,
we in turn show that it is sufficient to execute a precomputation phase during
which we store for each vertex $w\in V\setminus V(P)$ only the latest vertex
$\und(w)\in V(P)$ such that there
is a satellite path from $\und(w)$ to $w$ avoiding the parent
of $w$ in $D_{\ly(w)}$.
Finally, we explain
how to compute the values $\und(w)$ for all $w$ in $O(m \log m/ \log \log m)$ time.
Below we study the case $\ly(x)=\ly(v)$ in more detail.

\begin{lemma}
  If $\ly(x) = \ly(v)$ and $x$ is not an ancestor of $v$ in $D_{\ly(v)}$, then $\vls_{G-x}(v)=\vls_{G}(v)$.
\end{lemma}
\begin{proof}
We already argued that $\vls_{G-x}(v)\preceq\vls_{G}(v)$.
  Since $x$ is not an ancestor of $v$ in $D_{l}$, the dominator tree of $G[L_{l}\cup\{p_l\}]$,
  there is a path from $\vls_{G}(v)$ to $v$ in $G[L_l \cup \{p_l\}]$, avoiding $x$.
Hence, $\vls_{G-x}(v) = \vls_{G}(v)$. 
\end{proof}

Recall that we denote by $\und(w)$ the latest vertex on $P$ that has a satellite path to $w$ in $G-d_{\ly(w)}(w)$
(that is, avoiding the parent of $w$ in the dominator tree $D_{\ly(w)}$).
In other words,
$\und(w)=\vls_{G - d_{\ly(w)}(w)}(w)$. If no such vertex exists, $\und(w) = \perp$.

\begin{lemma}\label{l:undom-express}
  Let $Q$ be a satellite path from $p_k=\und(w)$ to $w$ in $G$ avoiding $d_{l}(w)$, where $l=\ly(w)$.
  Then $k<l$ and $Q$ can be expressed as $Q=Q_1 Q_2 Q_3$, such that $\und(w)\to y=Q_1 \subseteq G[L_{k}\cup \{p_k\}]$,
  $Q_2 =yz\in V(L_{k})\times V(L_{l})$ (i.e., $Q_2$ is a single-edge path), and $z\to w=Q_3 \subseteq G[D_l[d_l(w)]\setminus\{d_l(w)\}]$.
\end{lemma}
\begin{proof}
  First observe that $k<l$ follows from the fact that $w$ is not reachable
  from $p_{l+1},\ldots,p_{\ell}$ at all in the corresponding layers (by the definition of $L_l$),
  and it is not reachable from $p_l$ in $G[L_l \cup p_l]-d_l(w)$ by the definition of a dominator tree.

  Moreover, by the definition of layers, there is no edge in $G$ 
  from a vertex of $L_i$ to a vertex of $L_j$, where $i>j$.
  In particular, there is no edge from $L_l$ to $L_k$.

  Since the edges in $G$ can only go from lower-numbered
  layers to the higher-numbered ones, and $Q$ is a satellite
  path, $Q$ cannot visit a lower-numbered layer after
  visiting a higher-numbered layer.
  Suppose that $Q$ goes through a vertex of $a\in L_j$, where $k<j<l$.
  Let $R=a\to w$ be a subpath of $Q$.
  By the definition of $Q$, $R$ does not go through $d_l(w)$.
  Moreover, $a$ is reachable from $p_j$ in $G-d_l(w)$ (since $a\in L_j$ and $d_l(w)\in L_l$)
  by a satellite path.
  Consequently, $w$ is reachable from $p_j$ in $G-d_l(w)$
  by a satellite path, which contradicts the fact that $p_k=\vls_{G-d_l(w)}(w)$.
  We conclude that indeed $V(Q)\subseteq L_k\cup \{p_k\}\cup L_l$,
  and the vertices of $L_l$ appear on $Q$ only after the vertices of $L_k$.
  Hence $Q$ can be expressed as $Q_1(yz)Q_3$, where $Q_1\subseteq G[L_k\cup\{p_k\}]$
  and $Q_3\subseteq G[L_l]$.
  
  It remains to prove that in fact we have $Q_3\subseteq  G[D_l[d_l(w)]\setminus\{d_l(w)\}]$.
  Clearly, $d_l(w)\notin V(Q_3)$ since $Q$ avoids $d_l(w)$.
  Suppose a vertex $t\in V(Q_3)\cap (L_l\setminus D_l[d_l(w)])$ exists.
  Then, by Lemma~\ref{lemma:paths-through-SAP} 
  the subpath $t\to w$ of $Q_3$
  has to go through $d_l(w)$, a contradiction.
\end{proof}

For convenience we identify $\perp$ with $p_0$ and extend the order $\prec$ so
that $\perp\prec p_i$ for all $i=1,\ldots,\ell$.

\begin{lemma}\label{l:pathmax}
  If $\ly(x)= \ly(v)$ and $x$ is an ancestor of $v$ in $D_{\ly(v)}$, then $\vls_{G-x}(v)=p_q$, where $ q=\max\{t: p_t =\und(w), w \in D_{\ly(v)}(x,v]\}$. 
\end{lemma}
\begin{proof}
  Let $l=\ly(v)=\ly(x)$.
  We first show that for all $w\in D_{l}(x,v]$, if $\und(w)\neq\perp$ then the vertex $\und(w)$ has a satellite path to $v$ avoiding $x$ in $G$.
  By this, $p_q\preceq \vls_{G-x}(v)$  follows.
  Consider a satellite path $Q$ from $\und(w)$ to $w$, avoiding $d_{l}(w)$. 
  By Lemma~\ref{l:undom-express}, $V(Q)\cap L_l\subseteq D_l[d_l(w)]\setminus\{d_l(w)\}$.
  Since $D_l[d_l(w)]\setminus\{d_l(w)\}\subseteq D_l[x]\setminus \{x\}$, $Q$ avoids $x$.

  Now suppose that $p_q\prec \vls_{G-x}(v)$.
  Then, $\und(w)\prec \vls_{G-x}(v)$ for all $w\in D_{l}(x,v]$.
Take any simple satellite path $Q$ from $\vls_{G-x}(v)$ to $v$ in $G-x$, and consider the earliest vertex $w\in D_{l}(x,v]$ that it contains.
  Then the $\vls_{G-x}(v)\to w$ subpath of $Q$ avoids all vertices on $D_{l}(x,w)$ (including $d_{l}(w)$).
This contradicts that $\und(w)$ is the latest vertex on $P$ that has a satellite path to $w$ in $G-d_{l}(w)$.
  Hence, $\vls_{G-x}(v)\preceq p_q$ and we obtain $\vls_{G-x}(v)=p_q$.
\end{proof}

Suppose we have the vertices $\und(w)$ computed for all $w\in V\setminus V(P)$.
For any $i=1,\ldots,\ell$, let $D_i'$ be the tree $D_i$ with labels on its
edges added, such that the label of an edge $d_i(w)w$, where $w\in L_i$, is equal to 
$\und(w)$.
Then, by Lemma~\ref{l:pathmax}, computing $\vls_{G-x}(v)$ when $l=\ly(x)=\ly(v)$
and $x$ is an ancestor of $v$ in $D_l$
can be reduced to finding a maximum label on the $x'$ to $v$ path of $D_l'$,
where $x'$ is the child of $x$ in $D_{l}$ that $v\in D_l[x']$.
Such queries can be answered in constant time after linear preprocessing of $D_l'$:
we can use the data structure of~\cite{BenderF04} for finding $x'$,
and the data structure of~\cite{DemaineLW14} for finding the maximum label on a path.

\subsection{Computing $undom(w)$, for all $w$, in $O(n \log n/\log \log n)$ time.}

\newcommand{\undp}{ud}

Our final task is to compute the vertices $\und(w)$ for all $w\in V\setminus V(P)$.
Let initially $\undp(w):=\perp$ for all $w\in V\setminus V(P)$.
The goal is to eventually obtain $\undp(w)=\und(w)$ for each $w$.
The computation will be divided into $\ell$ phases numbered $\ell$ down to $1$.
During the phase $i$, we process all edges
$xy\in (V\times L_i)\cap E(G)$.
Recall that by the definition of the levels, there is no edge $e\in V(L_j) \times V(L_i)$ for $j>i$.

Each phase is subdivided into rounds, where each round processes a single edge
from $E(G)\cap ((V\setminus L_i)\times L_i)$ and some edges from $E(G)\cap (L_i\times L_i)$.
The edges originating in ``later'' layers are processed before the edges originating in ``earlier'' layers.
Consider the round processing an edge $xy\in L_j\times L_i$, where $j<i$.
We start by initializing a queue of edges $Q$, initially containing only $xy$.
While $Q$ is not empty, we extract an edge $zw$ from $Q$.
We test whether $\undp(w) \not= \perp$, and if so, we set $\undp(w) := p_j$, we remove from $G$ all edges from $E(G)\cap (D_i[w]\times (L_i\setminus D_i[w]))$
and push them to $Q$. 
The round ends when $Q$ becomes empty.
Afterwards, we proceed with the next round.
A phase ends when there are no unprocessed edges incident to a vertex of $L_i$.

We now prove the correctness of the above procedure.
\begin{lemma}\label{l:correct}
  Fix some phase $i$.
  After all rounds processing edges from $$E(G)\cap ((L_j\cup L_{j+1}\cup\ldots L_{i-1})\times L_i),$$
  we have $\undp(w)=\und(w)$ for all $w\in L_i$ such that
  $p_j\preceq \und(w)$ and $\undp(w)=\perp$ for all $w\in L_i$ such that
  $\und(w)\prec p_j$.
\end{lemma}
\begin{proof}
We proceed by induction on $j$. The base case is $j=i$ and then the claim
is clearly true.

Suppose $j< i$ and that the claim holds for all $j'>j$.
  Since edges $xy\in E(G)\cap (L_j\times L_i)$ are processed
  always after $E(G)\cap ((L_{j+1}\cup\ldots L_{i-1})\times L_i)$,
  and the procedure can only set $\undp(w):=j$ if $\undp(w)=\perp$
  when processing such an $xy$,
  we need to prove that for all $w\in L_i$:
  \begin{enumerate}[label={(\arabic*)}]
    \item if the procedure
  sets $\undp(w)=p_j$ then there exists a satellite $p_j\to w$ path
  avoiding $d_i(w)$, and
    \item if $\undp(w)=\perp$ before
  processing edges originating in $L_j$ and there is a satellite $p_j\to w$ path
  avoiding $d_i(w)$, then the procedure will set $\undp(w)=p_j$.
  \end{enumerate}

  Let us start with proving (1).
We prove this by induction on the edges extracted from the queue~$Q$, during the round for which we set $\undp(w) = p_j$. 
The base case is when the first edge gets extracted from $Q$, that is,
  the edge $xy\in E(G)\cap (L_j\times L_i)$ that initiates the round. 
In this case, there is a path from $p_j$ to $x$ in $L_j$ (which avoids $d_i(w)$ as $d_i(w)\notin L_j$) followed by the edge $xy$. 
Hence the claim holds for the base case.
Now suppose the claim holds for all edges that
have been extracted from $Q$ prior to some edge $zw$ in this round.
Consider the moment when $zw$ was inserted into $Q$; it was while another edge $e'=z'w'$
  was being processed, we correctly set $\undp(w') = p_j$ (by induction), and $zw$ was an edge from $E(G)\cap (D_i[w']\times (L_i\setminus D_i[w'])$. 
  Moreover, by Lemma~\ref{lemma:parent-property}, $d_i(w)$ is an ancestor of $z$ in $D_i$. 
  Since $w \notin D_i[w']$, $d_i(w)$ is also ancestor of $w'$ in $D_i$, as $w'$ is an ancestor of $z$. 
  Overall, we have a path from $p_j$ to $w'$ avoiding $d_i(w')$, and all ancestors of $d_i(w')$ in $D_i$ (as otherwise, by Lemma~\ref{lemma:paths-through-SAP},
  this $p_j\to w'$ path would go through $d_i(w')$, which would be a contradiction),
  in particular $d_i(w)$.
  By Lemma \ref{lemma:paths-ancestor-descendant-dom}, $w'$ has a path to $z$ avoiding all vertices from $V\setminus D_i[w']$ (including $d_i(w)$).
Hence, combined with edge $zw$, there is a path from $p_j$ to $w$ avoiding $d_i(w)$.
This concludes the first part of the proof.

  Now let us consider item (2). 
  For a contradiction, suppose there exists such $w\in L_i$
  that $\undp(w)=\perp$ before processing the edges originating in $L_j$,
  and there exists a simple satellite path $Q=p_j\to w$ avoiding
  $d_i(w)$, but the procedure does not set $\undp(w):=p_j$.
  By the inductive hypothesis, $\und(w)=p_j$.
  Hence, by Lemma~\ref{l:undom-express}, $Q$ first goes through the layer
  $L_j$, then through a single edge in $xy\in L_j\times L_i$,
  and then through the layer $L_i$.
  Out of such vertices $w$ and paths $Q$, choose such $w$ and $Q$
  so that the round when $xy$ is processed is earliest possible
  and, in case of ties, $Q$ has the minimal number of edges in $L_i\times L_i$
  that are not of the form $d_i(a)a$ (i.e., do not go
  from a parent to a child in $D_i$).
  Let $Q=Q_1(xy)Q_3$.
  Let all the edges not of the form $d_i(a)a$ in $Q_3$
  be $x_1y_1,\ldots,x_gy_g$ in order of their appearance on $Q_3$.
  We possibly have $g=0$.
  Set $y_0=y$.
  
  We first show that for all $k=1,\ldots,g$, $d_i(y_k)$ is a proper ancestor of $y_{k-1}$ in $D_i$.
  Suppose $d_i(y_k)$ is not a proper ancestor of $y_{k-1}$.
  Since by Lemma~\ref{lemma:parent-property} $d_i(y_k)$ is an ancestor of $x_k$,
  and all the edges in the subpath $y_{k-1}\to x_k$ go to children in
  the dominator tree, $d_i(y_k)\in D_i[y_{k-1}]$.
  Hence, the subpath $y_{k-1}\to y_k$ of $Q$ could be replaced
  with a $y_{k-1}\to y_k$ path that goes only to the children
  in the dominator tree, contradicting the minimality of $g$.
  Similarly we can show that $y_{k}\neq y_{k-1}$
  for all $k=1,\ldots,g$.

  It follows that $d_i(y_k)$ is an ancestor of $d_i(y_{k-1})$ and also
  that each subpath $y\to y_{k}$ goes only through vertices
  of $D_i[d_i(y_k)]\setminus \{d_i(y_k)\}$.
  In particular, for each $y_k$, the subpath $p_j\to y_k$
  avoids $d_i(y_k)$.
  Observe that we also necessarily have $y_g=w$,
  as otherwise, $w$ would be a proper descendant of $y_g$,
  and thus the subpath $y_g\to w$ would go through $d_i(w)$.

  Suppose that for some $k\in\{0,\ldots,g-1\}$ we have $\und(y_k)\neq p_j$.
  Then, since a satellite path $p_j\to y_k$ avoiding $d_i(y_k)$ exists,
  $p_j\prec \und(y_k)=p_{j'}$.
  Let $R$ be a $p_{j'}\to y_k$ satellite path avoiding $d_i(y_k)$.
  By Lemma~\ref{l:undom-express}, all the vertices of $L_i$
  that $R$ goes through lie in $D_i[d_i(y_k)]\setminus \{d_i(y_k)\}$.
  Consequently, $R$ avoids $d_i(w)$ since $d_i(w)=d_i(y_g)$ is an
  ancestor of $d_i(y_k)$.
  So, by replacing in $Q$ the $p_i\to y_k$ subpath with $R$
  we would obtain a $p_j'\to w$ satellite path avoiding $d_i(w)$,
  which would contradict $\und(w)=p_j$.
  
  By the fact that $xy$ is the earliest possible processed edge
  and the minimality of $g$, it follows that for all $y_k=y_0,\ldots,y_{g-1}$
  the algorithm sets $\undp(y_k)=p_j$.
  In particular, after the algorithm sets $\undp(y_{g-1}):=p_j$,
  for all remaining edges $x'y'\in L_i\times (L_i\setminus D_i[y_{g-1}])$,
  such that $\undp(y')=\perp$, $\undp(y'):=p_j$ is set.
  Note that by the fact that $d_i(y_g)$ is a proper ancestor
  of $y_{g-1}$ and $y_g\neq y_{g-1}$, it follows that
  $x_gy_g\in L_i\times (L_i\setminus D_i[y_{g-1}])$.
  If $x_gy_g\notin E(G)$ at this point, then $\undp(w=y_g)$ has
  already been set to something else than $\perp$ -- a contradiction.
  Otherwise, $x_gy_g$ will be processed at this point
  and therefore $\undp(w=y_g)$ will be set to $p_j$, a contradiction.
\end{proof}

\begin{corollary}
  After the algorithm finishes, we have $\undp(w)=\und(w)$ for all $w\in V\setminus V(P)$.
\end{corollary}
\begin{proof}
  We apply Lemma~\ref{l:correct} for all $i$ and $j=1$.
\end{proof}

Finally, we now analyze the time complexity of the procedure
that computes the vertices $\und(w)$ for all $w\in V\setminus V(P)$.
\begin{lemma}
  The algorithm from this section can be implemented to run in time $O\left(m \frac{\log m}{\log \log m}\right)$.
\end{lemma}
\begin{proof}
The time for computing the dominator trees of all layers $L_i$ is  $O(m)$ in total. 
Moreover, we can compute the correct order with which to process the incoming edges of each layer $L_i$, in $O(n+m)$ in total
using radix-sort.
Observe that the insertions and extraction to/from the maintained queue $Q$, through all layers $L_i$, take $O(m)$ time overall
since each edge is inserted into $Q$ at most once.

Now we bound the total time spent on reporting and deleting the edges from $G$.  
  This can be done with the help of a dynamic two-dimensional range reporting data structure, as we next explain. We build such a data structure for each graph $G[L_i]$ separately -- recall that in the $i$-th phase we only report/remove the edges of $E(G)\cap (L_i\times L_i)=E(G[L_i])$.

  Let $n'=|L_i|$ and $m'=|E(G[L_i])|$.
  Let $\pord:L_i\cup \{p_i\}\to [1..n'+1]$ be some preorder of the dominator tree $D_i$.
  Let $\tsz(v)=|V(D_i[v])|$.
Clearly, both $\pord$ and $\tsz$ can be computed in linear time.
  Note that we have $u\in D_i[v]$ if and only if $\pord(u)\in [\pord(v),\pord(v)+\tsz(v)-1]$.

  We map each edge $xy\in E(G[L_i])$ to a point $(\pord(x),\pord(y))$ on the plane.
Let $A$ be the set of obtained points.
We store the points $A$ in a two-dimensional dynamic range reporting data
  structure of Chan and Tsakalidis~\cite{ChanT17}.\footnote{In our application, any \emph{decremental} two-dimensional range reporting data
  structure would be enough to obtain a nearly linear time algorithm. 
  In fact, what we need is a range \emph{extraction} data structure that removes all the reported points
  that remain in the query rectangle so that no point is reported twice.
  To the best of our knowledge, such a problem is not well-studied and
  this is why we use a seemingly much more general data structure of~\cite{ChanT17}.}
This data structure supports insertions and deletions
in $O(\log^{2/3+o(1)}{b})$ time, and allows reporting
all points in a query rectangle in $O\left(\frac{\log{b}}{\log\log{b}}+k\right)$ time,
where $k$ is the number of points reported and $b$ is the maximum number of points present in the point set at any time.
Hence, we can build the range reporting
data structure in $O(m'\log^{2/3+o(1)}{m'})=O\left(m'\frac{\log{m'}}{\log\log{m'}}\right)$ time, by inserting all points from $A$.

  Suppose we are required to remove and report the edges of $E(G[L_i])\cap (D_i[v]\times (L_i\setminus D_i[v]))$.
Then, it is easy to see that this can be done by querying the range reporting data
structure for the edges corresponding to points in the set
$$[\pord(u),\pord(u)+\tsz(u)-1]\times ([1,\pord(u)-1]\cup [\pord(u)+\tsz(u),n']),$$
which, clearly, consists of two orthogonal rectangles in the plane.
  This way, we get all the remaining edges $xy$ of $G[L_i]$ such that $x\in D[v]$
  and $y\notin D[v]$.
Subsequently, (the corresponding points of) all the found edges are
removed from the range reporting data structure.
  Since each edge of $G[L_i]$ is reported and removed
only once, and we issue $O(m')$ queries to the range reporting
data structure, the total time used to process
any sequence of deletions is $O\left(m'\frac{\log{m'}}{\log\log{m'}}\right)$.
Since the sum of $m'$ over all layers is $m$,
the total time to report all edges and removing them at all levels $L_i$ is $O\left(m\frac{\log{m}}{\log\log{m}}\right)$.
\end{proof}

\section{Proof of Theorem~\ref{t:ds2}}\label{s:ds2}

In this section we prove the following theorem.
\thmdstwo*
Let the subsequent vertices of $P$ be $p_1,\ldots,p_\ell$,
i.e., $P=p_1\ldots p_\ell$.
For simplicity, let us set $\prec:=\prec_P$.
Recall that the failed vertex $x$ satisfies $x\in V(P)$.

Define $G_1$ and $G_2$ to be the two subgraphs of $G$
lying weakly on the two sides of $P$.
More formally, the bounding cycle of the infinite face
(which, by our assumption, contains $p_1$ and $p_\ell$) of $G$ can be
expressed as $C_1C_2$, where $C_1$ is a
non-oriented $p_1\to p_\ell$ path, and $C_2$ is a non-oriented
$p_\ell\to p_1$ path.
Then we define $G_1$ to be the subgraph of $G$ weakly inside
the cycle $C_1P^R$ (where $P^R$ is the non-oriented path $P$ reversed),
and $G_2$ to be the subgraph of $G$ weakly inside the cycle
$C_2P$.
We have $G_1\cap G_2=P$.
Note that $P$ is a part of the infinite face's bounding cycle of each $G_i$.

Our strategy for finding a $u\to v$ path $Q$ in $G-x$, $x\in V(P)$ that goes
through $V(P)$ will be to consider a few cases based on how
$Q$ crosses $P$. In order to minimize the number of cases we will
define a number of auxiliary notions in Section~\ref{s:ds2aux}.
Next, in Section~\ref{s:ds2prep} we will show how these notions can be efficiently computed.
Subsequently in Section~\ref{s:ds2simplify}
we will show a few technical lemmas stating that we only
need to look for paths $Q$ of a very special structure.
Finally, in Section~\ref{s:ds2query} we describe the query procedure.

\subsection{Auxiliary notions}\label{s:ds2aux}

We first extend the notation $\vf,\vl,\vfs,\vls$ from Section~\ref{s:ds1}
to subpaths of $P$.
For any $H\subseteq G$ denote by $\vf_H(v,a,b)$ ($\vl_H(v,a,b)$) the earliest
(latest resp.) vertex of $P[a,b]$ that $v$ can reach (that can reach $v$, resp.) in $H$.
In fact, as we will see, we will only need to compute $\vf_H(v,a,b)$ or $\vl_H(v,a,b)$
for $v\in V(P[a,b])$.

Similarly, for $v\in V\setminus V(P)$ denote by $\vfs_H(v,a,b)$ ($\vls_H(v,a,b)$) the earliest
(latest resp.) vertex of $P[a,b]$ that $v$ can reach (that can reach $v$, resp.) \emph{by a satellite path} in $H$.
For $v\in V(P)$ we set $\vfs_H(v,a,b)=\vls_H(v,a,b)=v$ if $a\preceq v\preceq b$.
Note that in some cases we leave the values $\vfs_H(v,a,b)$ and $\vls_H(v,a,b)$ undefined.

For brevity we sometimes omit the endpoints of the subpath
if we care about some earliest/latest vertices of the whole path $P=P[p_1,p_\ell]$,
and write (as in Section~\ref{s:ds1}) e.g., $\vfs_H(v)$ instead of $\vfs_H(v,p_1,p_\ell)$
or $\vl_H(v)$ instead of $\vl_H(v,p_1,p_\ell)$.

We will also need the notions of earliest/latest jumps and detours that are new to this section.

\newcommand{\jump}{\alpha}
\begin{definition}
  Let $w\in V(P)$. Let $H$ be a digraph such that $V(P)\subseteq V(H)$.
  We call the earliest vertex of $P$ that $w$ can reach using a satellite path
  in $H$ the \emph{earliest jump} of $w$ in $H$ and denote
  this vertex by $\jump^-_H(w)$.
  Similarly, we call the latest vertex of $P$ that $w$ can reach using
  a satellite path in $H$ the \emph{latest jump} of $w$ in $H$
  and denote this vertex by $\jump^+_H(w)$.
\end{definition}

\begin{definition}
  Let $x=p_k$. We define a \emph{detour} of $x$ to be any directed satellite path
  $D=p_i\to p_j$, where $i<k<j$.
  Detour $D$ is called \emph{minimal} if there is no detour $D'=p_{i'}\to p_{j'}$
  of $x$ such that $i\leq i'\leq j'\leq j$ and $j'-i'<j-i$.
  A pair $(p_i,p_j)$ is called a \emph{minimal detour pair} of $x$ if
  there exists a minimal detour of $x$ of the form $p_i\to p_j$.
\end{definition}

The fact that $p_1$ and $p_\ell$ lie on a single face of $G$ implies
the following important property of minimal detours.

\begin{lemma}\label{l:two-detours}
  For any $x\in V(P)$, there exist at most two minimal detour pairs of $x$.
\end{lemma}
\begin{proof}
  Let $x=p_k$.
  Assume the contrary and
  suppose there are at least three minimal detour pairs and let us take
  some three minimal detours $D_1,D_2,D_3$ corresponding to these pairs.
  Some two of $D_1,D_2,D_3$, say $D_1,D_2$ are both contained
  in some of $G_1$ and $G_2$ -- say $G_1$.

  Suppose $D_1=p_i\to p_j$ and $D_2=p_{i'}\to p_{j'}$.
  Assume w.l.o.g. that $i\leq i'$.
  Then since $(p_i,p_j)$ and $(p_{i'},p_{j'})$ are distinct minimal detour
  pairs of $x$, we have $i<i'<k<j<j'$.

  Recall that all vertices of $P$ lie on the infinite face of $G_1$.
  Since $i<i'<j<j'$, any $p_i\to p_j$ path in $G_1$ crosses (i.e., has a common internal vertex with)
  any $p_{i'}\to p_{j'}$ path in $G_1$.
  Let $z\in (V(D_1)\cap V(D_2))\setminus V(P)$.
  So $D_1$ can be represented as $D_1'D_1''$, where $D_1''=z\to p_j$.
  Similarly, $D_2=D_2'D_2''$, where $D_2'=p_{i'}\to z$.
  Observe that $D_2'D_1''=p_{i'}\to p_j$ is a detour of $x$.
  But since $j-i'<j-i$, this contradicts the minimality of $D_1$.
\end{proof}

\subsection{Computing $\vf,\vl,\vfs,\vls$, jumps and detours}\label{s:ds2prep}


\begin{lemma}\label{l:reach-interval}
  After linear preprocessing of $G$, for any $v\in V$ and $a,b\in V(P)$, where $a\preceq b$,
  the vertices $\vfs_G(v,a,b)$ and $\vls_G(v,a,b)$ (if they exist) can be computed
  in $O(\log{n})$ time.
\end{lemma}
\begin{proof}
  The case $v\in V(P)$ is trivial by the definition of $\vfs_G(v,a,b)$ and $\vls_G(v,a,b)$:
  these vertices are defined and equal to $v$ only for $a\preceq v\preceq b$.
  Hence, below we assume $v\in V\setminus V(P)$.
  
  First observe that any satellite path is entirely
  contained in either $G_1$ or $G_2$.
  Hence, $\vfs_G(v,a,b)$ is the earlier of $\vfs_{G_1}(v,a,b)$
  and $\vfs_{G_2}(v,a,b)$.
  In the following we focus on computing $\vfs_{G_1}(v,a,b)$.
  Satellite paths in $G_2$ can be handled identically.
  Computing latest vertices can be done symmetrically.

  Note that since $v\notin V(P)$, a satellite path $Q=v\to V(P[a,b])$ in $G_1$ does not
  use the outgoing edges of the vertices of $P$.
  Let us thus remove the outgoing edges of all $w\in V(P)$ from $G_1$
  and this way obtain~$G_1'$.
  Note that all the remaining paths in $G_1'$ are satellite now
  and all the satellite paths from $v$ in~$G_1$ are preserved in $G_1'$.

  For simplicity let us assume that $\ell$ is a power of 2. This is without loss of generality,
  since otherwise we could extend $P$ by adding less than $\ell$ vertices
  ``inside'' the last edge of $P$ so that the length
  of $P$ becomes a power of two.
  Recall that all the outgoing edges of $V(P)$ are removed in $G_1'$ anyway
  so this extension does not influence reachability
  or the answers to the considered queries.
  
  \newcommand{\iset}{\mathcal{I}}
  Let $\iset$ be the set of \emph{elementary} intervals defined as follows:
  for each $f\in \{0,\ldots,\log \ell\}$ the intervals $[(k-1) \cdot 2^{f}+1,k \cdot 2^f]$, for all $k\in[\frac{\ell}{2^f}]$, is included in $\iset$.
  Observe that $\iset$ has $\ell-1$ intervals
  and these intervals can be conveniently arranged into a full binary tree.

  We further extend $G_1'$ by adding an auxiliary vertex $p_{[e,f]}$
  for each $[e,f]\in \iset$ such that $e<f$.
  We also identify each vertex $p_i\in V(P)$ with $p_{[i,i]}$.
  For each $[e,f]\in \iset$ with $e<f$ we add to $G_1'$
  two edges $p_{[e,m]}p_{[e,f]}$ and $p_{[m+1,f]}p_{[e,f]}$, where
   $m=\left\lfloor\frac{e+f}{2}\right\rfloor$.
  
  Observe that for any $[e,f]\in \iset$, and $v\in V\setminus V(P)$
  a path from $v$ to any of $p_e,\ldots,p_f$ exists in $G_1'$
  if and only if a $v\to p_{[e,f}]$ path exists in $G_1'$.
  
  Since the vertices $p_1,\ldots,p_\ell$ lie
  on a single face of $G_1$ in this order,
  after the extension the graph~$G_1'$ remains planar (see Figure~\ref{f:tree}).
  We can thus build a reachability oracle of Holm et al.~\cite{HolmRT15} for $G_1'$
  in linear time so that we can answer reachability
  queries in $G_1'$ in constant time.

\begin{figure}[ht]
  \centering
  \begin{tikzpicture}[scale=0.7]

  \node[draw,circle,fill,inner sep=1.0pt] (p1) at (1,0) {};
  \node[draw,circle,fill,inner sep=1.0pt] (p2) at (2,0.5) {};
  \node[draw,circle,fill,inner sep=1.0pt] (p3) at (3,0.9) {};
  \node[draw,circle,fill,inner sep=1.0pt] (p4) at (4,1.21) {};
  \node[draw,circle,fill,inner sep=1.0pt] (p5) at (5,1.5) {};
  \node[draw,circle,fill,inner sep=1.0pt] (p7) at (7,2.0) {};
  \node[draw,circle,fill,inner sep=1.0pt] (p6) at (6,1.8) {};
  \node[draw,circle,fill,inner sep=1.0pt] (p8) at (8,2.1) {};
  \node[draw,circle,fill,inner sep=1.0pt] (p9) at (9,2.1) {};
  \node[draw,circle,fill,inner sep=1.0pt] (p10) at (10,2.0) {};
  \node[draw,circle,fill,inner sep=1.0pt] (p11) at (11,1.8) {};
  \node[draw,circle,fill,inner sep=1.0pt] (p12) at (12,1.5) {};
  \node[draw,circle,fill,inner sep=1.0pt] (p13) at (13,1.21) {};
  \node[draw,circle,fill,inner sep=1.0pt] (p15) at (15,0.5) {};
  \node[draw,circle,fill,inner sep=1.0pt] (p14) at (14,0.9) {};
  \node[draw,circle,fill,inner sep=1.0pt] (p16) at (16,0) {};

\fill[gray!15] (p1.center) -- (p2.center) -- (p3.center) -- (p4.center) --
              (p5.center) -- (p6.center) -- (p7.center) -- (p8.center) --
              (p9.center) -- (p10.center) -- (p11.center) -- (p12.center) --
              (p13.center) -- (p14.center) -- (p15.center) -- (p16.center) to[out=-110,in=-70]
              (p1.center);
  \node[draw,circle,fill,inner sep=1.0pt,label=180:{$p_1$}] (p1) at (1,0) {};
  \node[draw,circle,fill,inner sep=1.0pt] (p2) at (2,0.5) {};
  \node[draw,circle,fill,inner sep=1.0pt] (p3) at (3,0.9) {};
  \node[draw,circle,fill,inner sep=1.0pt] (p4) at (4,1.21) {};
  \node[draw,circle,fill,inner sep=1.0pt] (p5) at (5,1.5) {};
  \node[draw,circle,fill,inner sep=1.0pt] (p7) at (7,2.0) {};
  \node[draw,circle,fill,inner sep=1.0pt] (p6) at (6,1.8) {};
  \node[draw,circle,fill,inner sep=1.0pt] (p8) at (8,2.1) {};
  \node[draw,circle,fill,inner sep=1.0pt] (p9) at (9,2.1) {};
  \node[draw,circle,fill,inner sep=1.0pt] (p10) at (10,2.0) {};
  \node[draw,circle,fill,inner sep=1.0pt] (p11) at (11,1.8) {};
  \node[draw,circle,fill,inner sep=1.0pt] (p12) at (12,1.5) {};
  \node[draw,circle,fill,inner sep=1.0pt] (p13) at (13,1.21) {};
  \node[draw,circle,fill,inner sep=1.0pt] (p15) at (15,0.5) {};
  \node[draw,circle,fill,inner sep=1.0pt] (p14) at (14,0.9) {};
  \node[draw,circle,fill,inner sep=1.0pt,label=0:{$p_{\ell}$}] (p16) at (16,0) {};

  \node[draw,fill,scale=0.4,label=180:{$p_{[1,2]}$}] (q12) at (1,3) {};
  \node[draw,fill,scale=0.4,label=200:{$p_{[3,4]}$}] (q34) at (3,3.9) {};
  \node[draw,fill,scale=0.4,label=100:{$p_{[5,6]}$}] (q56) at (5,4.5) {};
  \node[draw,fill,scale=0.4,label=100:{$p_{[7,8]}$}] (q78) at (7.3,5) {};

  \node[draw,fill,scale=0.4,label=80:{$p_{[9,10]}$}] (q910) at (9.7,5) {};
  \node[draw,fill,scale=0.4,label=87:{$p_{[11,12]}$}] (q1112) at (12,4.5) {};
  \node[draw,fill,scale=0.4,label=0:{$p_{[13,14]}$}] (q1314) at (14,3.9) {};
  \node[draw,fill,scale=0.4,label=0:{$p_{[15,16]}$}] (q1516) at (16.5,3) {};

  \node[draw,fill,scale=0.4,label=180:{$p_{[1,4]}$}] (q14) at (2.4,6.5) {};
  \node[draw,fill,scale=0.4,label=180:{$p_{[5,8]}$}] (q58) at (6.5,7.5) {};
  \node[draw,fill,scale=0.4,label=0:{$p_{[9,12]}$}] (q912) at (10.5,7.5) {};
  \node[draw,fill,scale=0.4,label=0:{$p_{[13,16]}$}] (q1316) at (14.6,6.5) {};
  
  \node[draw,fill,scale=0.4,label=180:{$p_{[1,8]}$}] (q18) at (6,10) {};
  \node[draw,fill,scale=0.4,label=0:{$p_{[9,16]}$}] (q916) at (11,10) {};
  
  \node[draw,fill,scale=0.4,label=90:{$p_{[1,16]}$}] (q116) at (8.5,11.5) {};
  
  \draw[-latex] (p1) -- (q12);
  \draw[-latex] (p2) -- (q12);
  
  \draw[-latex] (p3) -- (q34);
  \draw[-latex] (p4) -- (q34);
  \draw[-latex] (p5) -- (q56);
  \draw[-latex] (p6) -- (q56);
  \draw[-latex] (p7) -- (q78);
  \draw[-latex] (p8) -- (q78);
  
  \draw[-latex] (p9) -- (q910);
  \draw[-latex] (p10) -- (q910);
  \draw[-latex] (p11) -- (q1112);
  \draw[-latex] (p12) -- (q1112);
  \draw[-latex] (p13) -- (q1314);
  \draw[-latex] (p14) -- (q1314);
  
  \draw[-latex] (p15) -- (q1516);
  \draw[-latex] (p16) -- (q1516);

  \draw[-latex] (q12) -- (q14);
  \draw[-latex] (q34) -- (q14);
  \draw[-latex] (q56) -- (q58);
  \draw[-latex] (q78) -- (q58);
  
  \draw[-latex] (q910) -- (q912);
  \draw[-latex] (q1112) -- (q912);
  \draw[-latex] (q1314) -- (q1316);
  \draw[-latex] (q1516) -- (q1316);
  
  \draw[-latex] (q14) -- (q18);
  \draw[-latex] (q58) -- (q18);
  \draw[-latex] (q912) -- (q916);
  \draw[-latex] (q1316) -- (q916);

  \draw[-latex] (q18) -- (q116);
  \draw[-latex] (q916) -- (q116);

  \draw[-latex,dotted] (p1) -- (p2);
  \draw[-latex,dotted] (p2) -- (p3);
  \draw[-latex,dotted] (p3) -- (p4);
  \draw[-latex,dotted] (p4) -- (p5);
  \draw[-latex,dotted] (p5) -- (p6);
  \draw[-latex,dotted] (p6) -- (p7);
  \draw[-latex,dotted] (p7) -- (p8);
  \draw[-latex,dotted] (p8) -- (p9);
  \draw[-latex,dotted] (p9) -- (p10);
  \draw[-latex,dotted] (p10) -- (p11);
  \draw[-latex,dotted] (p11) -- (p12);
  \draw[-latex,dotted] (p12) -- (p13);
  \draw[-latex,dotted] (p13) -- (p14);
  \draw[-latex,dotted] (p14) -- (p15);
  \draw[-latex,dotted] (p15) -- (p16);

  \draw[-latex] (3,-1) -- (p3);
  \draw[-latex] (5,-2) -- (p3);

  \draw[-latex] (9,-2) -- (p12);
  \draw[-latex] (11,-3) -- (p12);
  \draw[-latex] (15,-1) -- (p12);

  \node at (8.5,-1) {$G_1'$};
  \end{tikzpicture}
  \caption{Auxiliary vertices corresponding to elementary intervals.}
  \label{f:tree}
\end{figure}
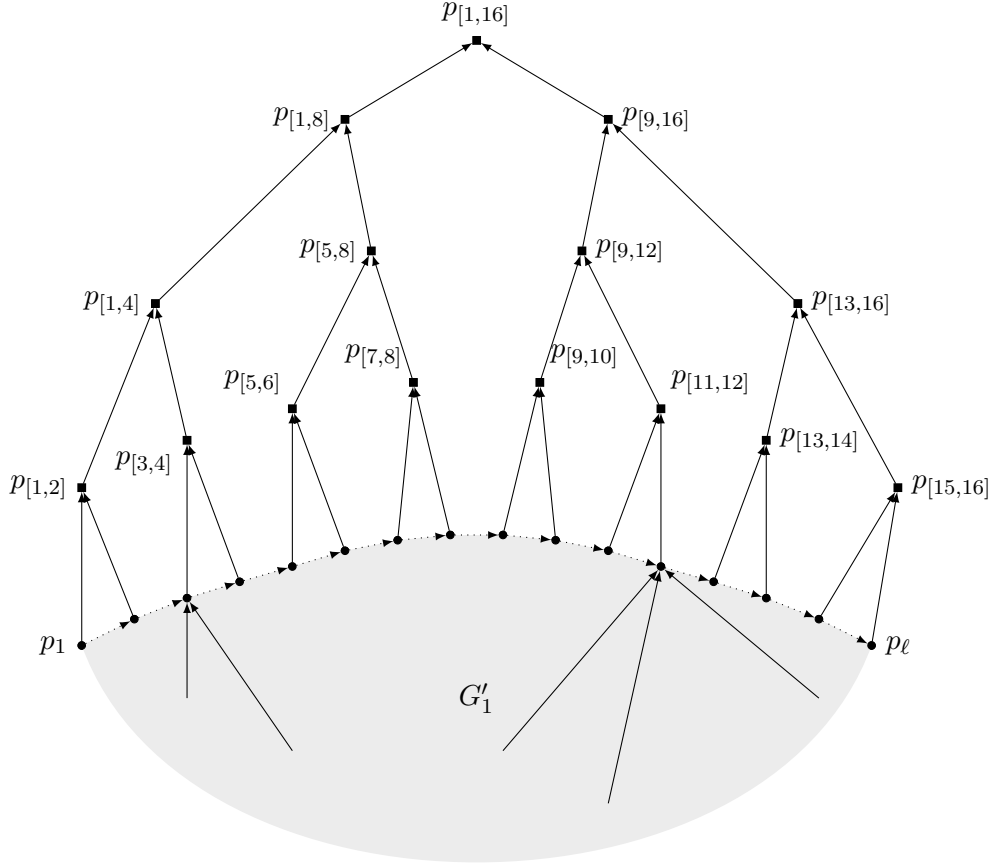
\newcommand{\nil}{\textbf{nil}}

  Note that we would be able to compute $\vfs_{G_1}(v,a,b)$ in
  $O(\log{(b-a+1)})$ time via binary search if only we could query in $O(1)$ whether a path
  from either of $p_x,\ldots,p_y$ to $v$  exists for an arbitrary interval $[x,y]$.
  However, in $O(1)$ time we can only handle such queries for $[x,y]\in \iset$.
  It is well-known that one can decompose an arbitrary $[x,y]$ into $O(\log{n})$ elementary
  intervals: this way, we could implement a single step of binary
  search in $O(\log{n})$ time and thus computing $\vfs_{G_1}(v,a,b)$
  would take $O(\log^2{n})$ time.
  
  In order to compute $\vfs_{G_1}(v,a,b)$ faster, we use a slightly more complicated binary-search-like recursive
  procedure as follows.
  The recursive procedure $\texttt{first}$ is passed a single parameter
  $[e,f]\in \iset$ and returns the earliest
  vertex from $p_{\max(a,e)},\ldots,p_{\min(b,f)}$ reachable from $v$
  in $G_1'$, or $\nil$ if such a vertex does not exits.
  This way, in order to find $\vfs_{G_1}(v,a,b)$ we compute $\texttt{first}([1,\ell])$.
  
  The procedure $\texttt{first}([e,f])$ is implemented as follows.
  \begin{enumerate}
    \item If $[a,b]\cap [e,f]=\emptyset$ then we return $\nil$.
    \item Otherwise, if $[e,f]\subseteq [a,b]$ and 
      there is no $v\to p_{[e,f]}$ path in $G_1'$
  (which can be decided in $O(1)$ time with a single reachability query on $G_1'$), we also return $\nil$.
    \item Otherwise if $e=f$ we return $p_e$.
    \item Finally, if none of the above cases apply,
  let $m=\left\lfloor\frac{e+f}{2}\right\rfloor$.
  If $\texttt{first}([e,m])=v$, we return~$v$.
  Otherwise (i.e., if $\texttt{first}([e,m])$ returned $\nil$), we compute
  and return $\texttt{first}([m+1,f])$.
  \end{enumerate}

  It is easy to see that the procedure is correct.
  Let us now analyze its running time.
  We say that a call to $\texttt{first}([e,f])$ is \emph{non-leaf}
  if it invokes the procedure recursively.
  Note that the query algorithm runs in time linear in the number of
  non-leaf calls $\texttt{first}([e,f])$.
  We now show that the number of non-leaf calls is $O(\log{n})$.

  First observe that for each non-leaf call we have
  $[e,f]\cap [a,b]\neq\emptyset$.
  Let us first count non-leaf calls such that $[e,f]\not\subseteq [a,b]$.
  Suppose for some $k$ there are at least
  3 such (pairwise disjoint) intervals $[e_i,f_i]$, for $i=1,2,3$ with $f_i-e_i+1=2^k$.
  Let $e_1<e_2<e_3$. Since $[a,b]$ intersects both $[e_1,f_1]$ and $[e_2,f_2]$
  but does not contain any of them $e_2\in [a,b]$.
  Similarly, since $[a,b]$ intersects both $[e_2,f_2]$ and $[e_3,f_3]$
  but does not contain any of them $f_2\in [a,b]$.
  Hence, $[e_2,f_2]\subseteq [a,b]$, a contradiction.
  So there are at most 2 such intervals per $k\in \{0,\ldots,\log_2\ell\}$,
  and conclude that there are $O(\log{n})$ non-leaf calls such that
  $[e,f]\not\subseteq [a,b]$.
  
  Now suppose $\texttt{first}([e,f])$ is a non-leaf call
  and $[e,f]\subseteq [a,b]$.
  Since this is a non-leaf call, there is a path from $v$ to
  some of $p_j$, where $j\in [e,f]$ in $G_1'$, and the call will
  return $p_j$.
  Consequently, the root call will also return $p_j$.
  We conclude that all such non-leaf calls return the same $p_j$.
  But there are only $O(\log{n})$ intervals $[e,f]$ such
  that $p_j\in [e,f]$.
\end{proof}

\begin{lemma}\label{l:scc-interval}
  After linear preprocessing, we can compute $\vf_{G-x}(v,a,b)$ ($\vl_{G-x}(v,a,b)$) for
  all $v,a,b,x\in V(P)$ such that $a\preceq v\preceq b\prec x$ or $x\prec a\preceq v\preceq b$
  in $O(\log{n})$ time.
\end{lemma}
\begin{proof}
We build the data structure of Theorem~\ref{t:scc-fail-query}
for graph $G$ in linear time.
  We only show how to compute $\vf_{G-x}(v,a,b)$, as $\vl_{G-x}(v,a,b)$ can be computed
completely analogously.

  Let $v'=\vf_{G-x}(v,a,b)$.
Observe that $v'\preceq v$ since a path $v\to v$ exists in $G-x$.
  Note that $P[a,b]$ is a path in $G-x$.
  Since a path $v'\to v$ exists in $G-x$ (since $P[v',v]$ is a subpath of $P[a,b]$),
  $v$ and $v'$ are in fact strongly connected in $G-x$.

  By Lemma~\ref{l:path-scc}, if $v'$ is strongly
  connected to $v$ in $G-x$, then all $y\in V(P)$, $v'\preceq y\preceq v$,
  are also strongly connected
  to $v$ in $G-x$.
  Consequently, we can find $v'$ by binary searching
  on the subpath $P[a,v]$ for the last vertex not strongly
  connected to $y$ in $G-x$.
  A single step of binary search is executed
  in $O(1)$ time by Theorem~\ref{t:scc-fail-query}.
\end{proof}

\newcommand{\vin}{\text{in}}
\newcommand{\vout}{\text{out}}

\begin{lemma}\label{l:compute-jumps}
  For a digraph $H$ such that $V(P)\subseteq V(H)$, the earliest and latest jumps of all $w\in V(P)$
  can be computed in linear time.
\end{lemma}
\begin{proof}
  We only show how to compute earliest jumps, as latest jump can be computed
  analogously.
  
  We first replace each $p_i\in V(P)$ in $H$ with two vertices $p_i^{\vin}$ and $p_i^{\vout}$.
  Next we change each original edge $p_iv$ of $H$ into $p_i^\vout v$,
  and subsequently turn each edge $up_i$ of $H$ into $up_i^\vin$.
  Observe that every $p_i^\vout$ has only outgoing
  edges, whereas every $p_i^\vin$ has only incoming edges.

  Observe that there is a 1-1 correspondence between $p_i^\vout\to p_j^\vin$ paths in $H$
  and $p_i\to p_j$ \emph{satellite} paths in $H$ before the transformation
  since no path in the current $H$ can go through either $p_l^\vin$ or $p_l^\vout$
  as an intermediate vertex.
  Therefore, the task of computing the earliest jumps can be reduced
  to computing, for each $i\in [1,\ell]$, the minimum $j$ such
  that $p_j^\vin$ is reachable from $p_i^\vout$ in the transformed $H$.
  This, in turn, can be done as follows, in an essentially the same way
  as we computed the layers in Section~\ref{s:satellite}.

  For each $i=1,\ldots,\ell$, we run a reverse graph search
  from $p_i^\vin$ in $H$ that only enters vertices that have not
  been visited so far.
  If $w\in V(H)$ gets visited in phase $i$, we set $\alpha_H^-(w)=i$.
  Observe that if a vertex $w$ is visited in phase $i$, this means
  that there exists a path $w\to p_i^\vin$ in $H$ and for any $j<i$
  there is no path $w\to p_j^\vin$ in $H$.
  This in particular applies to vertices $w$ of the form $p_l^\vout$.

  Finally, observe that the reverse graph search takes linear total time through all phases.
\end{proof}

\begin{lemma}
  The minimal detour pairs of all $x\in V(P)$ can be computed in linear time.
\end{lemma}
\begin{proof}
  From the proof of Lemma~\ref{l:two-detours} it follows that in
  fact there is at most one minimal detour pair for each $x\in V(P)$
  in each $G_i$, $i=1,2$.
  We thus compute this pair (if it exists) separately for $G_1$
  and $G_2$. In the following we focus on $G_1$ since $G_2$ can be handled identically.

  Let $x=p_k$.
  Suppose a minimal detour pair $(p_i,p_j)$ of $x$ in $G_1$ exists.
  Then the latest jump of $p_i$ in $G_1$ also exists and $x\prec p_j\preceq \jump^+_{G_1}(p_i)$.
  Let $a\in V(P)$ be the latest vertex $a\prec x$ such that
  $x\prec \jump^+_{G_1}(a)$.
  We conclude that $p_i\preceq a$.
  Similarly, the earliest (with respect to the original order
  of $P$) jump of $p_j$ in $\rev{G_1}$ also exists
  and $\jump^-_{\rev{G_1}}(p_j)\preceq p_i\prec x$.
  Let $b\in V(P)$ be the earliest vertex satisfying $x\prec b$ such
  that $\jump^-_{\rev{G_1}}(b)\prec x$.
  Observe that we have $b\preceq p_j$,
  and thus $p_i\preceq a\prec x\prec b\preceq p_j$.

  We now prove that in fact $p_i=a$ and $p_j=b$.
  First observe that 
  since there exist detours
  $D_b=\jump^-_{\rev{G_1}}(b)\to b$ and $D_a=a\to \jump^+_{G_1}(a)$ of $x$ in $G_1$,
  by the definition of $a$ we have $\jump^-_{\rev{G_1}}(b)\preceq a$
  and $b\preceq  \jump^+_{G_1}(a)$.
  So, $\jump^-_{\rev{G_1}}(b)\preceq a\prec x\prec b\preceq  \jump^+_{G_1}(a)$.
  But $\jump^-_{\rev{G_1}}(b),a,b,\jump^+_{G_1}(a)$ all lie on a single
  face of $G_1$ (in that order), so 
  $D_a$ and $D_b$ have a common vertex.
  Similarly as in the proof of Lemma~\ref{l:two-detours}, we can conclude
  that in fact there exists a $a\to b$ detour of $x$ in $G_1$.
  Hence, indeed $p_i=a$ and $p_j=b$.

  It remains to show how to compute the vertices $a,b$, as defined above,
  for all $x\in V(P)$.
  Let us focus on computing the values $a$; the values $b$ can be computed
  by proceeding symmetrically.
  We first compute the values $\jump^+_{G_1}(w)$
  for all $w\in V(P)$ in linear time using Lemma~\ref{l:compute-jumps}.
  
  Now, the problem can be rephrased in a more abstract
  way as follows:
  given an array $t[1..\ell]$, compute for each $i=1,\ldots,\ell$
  the value $h[i]=\max\{j:j<i\land t[j]>i\}$.
  We solve this problem by processing $t$ left-to-right
  and maintaining a certain subset $S$ of indices $j<i$ that surely
  contains all the values $h[i],\ldots,h[\ell]$ that are less than $i$.
  The indices of $S$ are stored in a stack sorted bottom-to-top.
  Additionally we maintain the invariant that for $S=\{i_1,\ldots,i_s\}$, where $i_1<\ldots<i_s$,
  we have $t[i_1]>\ldots>t[i_s]$.
  
  Initially, $S$ is empty. Suppose we process some $i$.
  By the invariants posed on $S$ we know that $h[i]\in S$.
  Let $i_s$ be the top element of the stack.
  If $t[i_s]>i$ then $h[i]=i_s$ since $i_s=\max{S}$.
  Otherwise $t[i_s]\leq i$.
  Hence, we also have $t[i_s]\leq j$ for all $j\geq i$,
  so $h[j]\neq i_s$.
  Consequently, we can remove $i_s$ from $S$ (by popping
  it from the top of the stack) without breaking
  the invariants posed on~$S$.
  We repeat the process until $S$ is empty
  or $t[i_s]>i$.
  If $S$ becomes empty, we set $h[i]=-\infty$,
  otherwise $h[i]=i_s$.
  At this point, $S$ contains values $h[j]$
  for all $j>i$ such that $h[j]<i$.

  In order to move to the next $i$, we need $S$
  to contain values $h[j]$ for all $j>i$ such that $h[j]\leq i$
  and the stack storing $S$ has to be sorted
  top to bottom according to the $t$-values.
  If $t[i_s]>t[i]$, we push $i$ to $S$ and finish.
  Suppose $t[i_s]\leq t[i]$.
  We show that $h[j]\neq t[i_s]$ for $j>i$ and thus
  we can safely remove $i_s$ from $S$.
  Assume the contrary, i.e., $h[j]=i_s$ for some $j>i$.
  But then $t[i]\geq t[h[j]]>j$ and $h[j]=\max\{j':j'<j\land t[j']>i\}\geq i>h[j]$,
  a contradiction.

  The running time of this algorithms is clearly linear
  in $\ell$ plus the number of stack operations.
  The total number of stack operations is linear in
  the number of indices pushed to to the stack, i.e., no more than $\ell$.
\end{proof}

\subsection{Simplifying paths}\label{s:ds2simplify}

In this section we devise a few lemmas that will
allow the query procedure 
to consider only paths of a very special form.

Let $R$ be any path. Let $\enter{R}$ be the first vertex of $R$ that additionally
lies on $P$.
Similarly, let $\leave{R}$ be the last vertex of $R$ that
additionally lies on $P$.
Let $\mnv{R}$ ($\mxv{R}$) be the earliest (latest) vertex of $V(R)\cap V(P)$ wrt. to $\preceq_P$.

\begin{lemma}\label{l:move-u-general}
  Let $a,b$ be some two vertices of $P$ such that
  $x\notin V(P[a,b])$.
  Let $Q=LR$ be any $u\to v$ path in $G-x$ satisfying $\enter{Q}\in V(P[a,b])$
  and such that $L$ is a satellite $u\to \enter{Q}$ path.
  Let $u'=\vf_{G-x}(\vfs_G(u,a,b),a,b)$.

  Then $Q$ does not go through the vertices of $P[a,b]$ earlier than $u'$.
  Moreover, 
there exists a $u\to v$ path $Q'=L'R$ in $G-x$ such that
      $\enter{Q'}\in V(P[a,b])$, $u'\in V(L')$ and $\leave{Q'}=\leave{Q}$.
\end{lemma}
\begin{proof}
  First note that by $\enter{Q}\in V(P[a,b])$, $\vfs_G(u,a,b)$ is in fact defined.
  Hence $u'\in V(P[a,b])$ also exists and $u'\preceq\vfs_G(v,a,b)$.

  Suppose $Q$ goes through a vertex $u''$ of $P[a,b]$ earlier than $u'$.
  Then there exists an $\enter{Q}\to u''$ path in $G-x$.
  We have $\vfs_G(u,a,b)\preceq \enter{Q}$.
  Since $x\notin V(P[a,b])$, there exists a $\vfs_G(u,a,b)\to \enter{Q}$ path (e.g., $P[\vfs_G(u,a,b),\enter{Q}]$)
  in $G-x$.
  It follows that there exists a $\vfs_G(u,a,b)\to u''$ path in $G-x$,
  which contradicts the definition of $u'$.

  To finish the proof, note that it is sufficient to set $L'$
  to be a concatenation of a satellite $u\to \vfs_G(u,a,b)$ path,
  any $\vfs_G(u,a,b)\to u'$ path in $G-x$ and $P[u',\enter{Q}]$.
  Note that $L'$ is not necessarily simple.
\end{proof}
\begin{lemma}\label{l:move-u}
  Let $Q=LR$ be any $u\to v$ path in $G-x$ satisfying $\enter{Q}\prec x$
  and such that $L$ is a $u\to \enter{Q}$ satellite path.
  Let $u'=\vf_{G-x}(\vfs_G(u))$.
  Then $u'\preceq \mnv{Q}$.
  Moreover, 
there exists a $u\to v$ path $Q'=L'R$ in $G-x$ such that
      $\enter{Q'}\prec x$, $\mnv{L'}=\mnv{Q'}=u'$, and $\leave{Q'}=\leave{Q}$.
\end{lemma}
\begin{proof}
  It is enough to apply Lemma~\ref{l:move-u-general} to $a=p_1$ and $b$ equal
  to the vertex preceding $x$ on $P$.
\end{proof}

\begin{lemma}\label{l:move-v}
  Let $Q=LR$ be any $u\to v$ path in $G-x$ satisfying $x\prec \leave{Q}$
  and such that $R$ is a $\leave{Q}\to v$ satellite path.
  Let $v'=\vl_{G-x}(\vls_G(v))$.
  Then $\mxv{Q}\preceq v'$.
  Moreover, 
there exists a $u\to v$ path $Q'=LR'$ in $G-x$ such that
      $x\prec\leave{Q'}$, $\mxv{R'}=\mxv{Q'}=v'$, $\enter{Q'}=\enter{Q}$.
\end{lemma}
\begin{proof}
  This lemma is symmetric to Lemma~\ref{l:move-u} and therefore can be proved
  analogously.
\end{proof}

\begin{lemma}\label{l:simplify}
  Suppose a directed path $Q=u\to v$ in $G-x$ such that
  $\mnv{Q}\prec x\prec \mxv{Q}$ can be represented as $Q_1RQ_2$, where
  $R$ is a $\mnv{Q}\to\mxv{Q}$ path.
  Then there is a $u\to v$ path $Q'$ in $G-x$ that can be represented
  as $Q_1P_1P_2P_3Q_2$, where $P_1$ and $P_3$ are (possibly empty) subpaths of $P$
  and $P_2$ is a minimal detour of $x$.
\end{lemma}
\begin{proof}
  Let $R=R_1R_2$, where $R_2$ starts at $y$, $\mnv{Q}\preceq y\prec x$ and $R_2$ does
  not go through any other vertices of $P[\mnv{Q},x]$.
  Such an $y$ exists since $\mnv{Q}\prec x\prec \mxv{Q}$.
  Observe that $R_2$ starts with some detour $D=y\to z$ of $x$,
  where $z\in P[x,\mxv{Q}]$.
  Hence $R$ can be rewritten as $R_1DR_3$, where $R_3=z\to \mxv{Q}$
  is possibly an empty path.
  If $D$ is not minimal, let $D'=y'\to z'$ be a minimal
  detour of $x$ such that $y\preceq y'\prec x\prec z'\preceq z$.
  Note that since $\mnv{Q}\preceq y\preceq y' \prec x$ and $x\prec z'\preceq z\preceq\mxv{Q}$,
  $P[\mnv{Q},y']D'P[z',\mxv{Q}]$ is a $\mnv{Q}\to \mxv{Q}$ path in $G-x$.
\end{proof}

\subsection{The query algorithm}\label{s:ds2query}
Finally we are ready to describe our query algorithm.
Recall that we want to efficiently check whether
there exists a $u\to v$ path $Q$ in $G-x$ that
goes through $V(P)$, for $x\in V(P)$.
We consider cases based on the configuration of
$\enter{Q}$, $\leave{Q}$ and $x$ on~$P$.

\begin{enumerate}
  \item Suppose there exists such a path $Q$ that $\enter{Q}\prec x$
and $x\prec \leave{Q}$.
    Let $u'=\vf_{G-x}(\vfs_G(u))$ as in Lemma~\ref{l:move-u}.
    Since $\vfs_G(u)\prec x$, $u'$ can be computed
    in $O(\log{n})$ time by first using Lemma~\ref{l:reach-interval}
    to compute $\vfs_G(u)$ and then using Lemma~\ref{l:scc-interval}.

By Lemma~\ref{l:move-u}, there exists a $u\to v$ path $Q'=L'R$ in $G-x$ such that
$\enter{Q'}\prec x$, $\mnv{L'}=\mnv{Q'}=u'$ and $x\prec \leave{Q'}$.
    Let $v'=\vl_{G-x}(v,\vls_G(v))$.
    Similarly, since $x\prec \leave{Q'}\preceq \vls_G(v)$, $v'$ can be computed
    in $O(\log{n})$ time by Lemmas~\ref{l:reach-interval}~and~\ref{l:scc-interval}.
By Lemma~\ref{l:move-v}, there exists a $u\to v$ path $Q''=L'R'$ in $G-x$ such
that $\enter{Q''}\prec x$, $x\prec \leave{Q''}$, $\mxv{Q''}=v'$.
    Since $\mnv{L'}=u'$ and $Q''$ cannot go through a vertex of $P$
    earlier then $u'$ (by Lemma~\ref{l:move-u}), we also have $\mnv{Q''}=u'$.

As we have $u'=\mnv{Q''}\prec x\prec \mxv{Q''}=v'$,
we can apply Lemma~\ref{l:simplify}.
In order to find a $u\to v$ path in $G-x$ such that $\enter{Q}\prec x$
and $x\prec \leave{Q}$, we only need
to check whether there exits a path $P_1P_2P_3=u'\to v'$
such that $P_1$ and $P_3$ are subpaths
of $P$ and $P_2$ is a minimal detour of $x$.
By Lemma~\ref{l:two-detours}, there are only two
minimal detour pairs $(e,f)$ of $x$.
Assuming $P_2=e\to f$, $P_1$ exists
if and only if $u'\preceq e$.
Similarly $P_3$ exists if and only if $f\preceq v'$.
Hence, each of at most two minimal detour pairs
can be checked in constant time.

\item Suppose that $\enter{Q}\prec x$ and $\leave{Q}\prec x$.
  Similarly as in the previous case, by Lemma~\ref{l:move-u} (and putting $Q:=Q'$)
  we can assume that the sought path $Q$ goes through $u'=\vf_{G-x}(\vfs_G(u))$
  and $\mnv{Q}=u'$.
  As before, $u'\prec x$ can be reached from $u$ in $G-x$.
    Now, since $u'=\mnv{Q}\preceq \leave{Q}\prec x$, we can limit
    our attention to such paths $Q$ that
    the subpath from $u'$ to $\leave{Q}$
    of $Q$ is actually a subpath of $P$.
    Hence, in order to check if such $Q$ exists,
    we only need to check whether
    there is a satellite path $z\to v$, where $u'\preceq z\prec x$.
    In other words we can check
    whether $\vls_{G}(v,u',x')$, 
    where $x'$ is a predecessor of $x$ on $P$, exists.
    This can be done in $O(\log{n})$ time
    using Lemma~\ref{l:reach-interval}.

\item Suppose that $x\prec \enter{Q}$ and $x\prec \leave{Q}$.
  This case is symmetric to the preceding case:
  we analogously check whether $u$ can reach $v'$
    by a satellite path and a subpath of $P$.

\item Finally, suppose $x\prec \enter{Q}$ and $\leave{Q}\prec x$.
      Let $y$ be the vertex following $x$ on $P$.
      Let us now apply Lemma~\ref{l:move-u-general} for $a=y$ and $b=p_\ell$,
      and let $u'=\vf_{G-x}(\vfs_G(u,y,p_\ell),y,p_\ell)$ be as in Lemma~\ref{l:move-u-general}.
      $u'$ can be again computed in $O(\log{n})$ time
      by first using Lemma~\ref{l:reach-interval} to obtain
      $\vfs_G(u,y,p_\ell)$ and then applying Lemma~\ref{l:scc-interval}.
      We conclude that there exists a $u\to v$ path $Q'$ in $G-x$
      that goes through $u'$, does not go through
      vertices of $P$ between $x$ and $u'$, and satisfies
      $\leave{Q'}\prec x$.
      Note that since $Q'$ last departs from $V(P)$ at a vertex
      earlier than $x$, and simultaneously goes through $u'\in V(P[y,p_\ell])$, $Q'$
      can be written as $Q'=YRST$, where $R=u'\to s$, $S=s\to t$,
      $s\in V(P[y,p_\ell])$, $t\in V(P)$, $t\prec x$,
      and $S$ is a satellite path.

      Let $t'$ be the earliest vertex of $P$ such that $t'$ can
      be reached from $u'$ by a path of the form $R'S'$, where $R'$
      is a subpath of $P$ and $S'$ is a satellite path.
      Observe that we have $t'\preceq t$, since the path
      $R$ does not go through any vertex between $x$ and $u'$ on $P$
      and thus $t$ can be reached from $u'$ by a path $P[u',s]S$.
      Note also that there
      exists a $u\to v$ path $Q''$ avoiding
      $x$ such that $Q''=YR'S'T'$: if we set $T'=P[t',t]T$,
      clearly all $Y,R',S',T'$ avoid $x$.

      Furthermore, note that we could replace path $T'$ with
      any $t'\to v$ path $T''$ in $G-x$ satisfying $\leave{T''}\prec x$
      and would still obtain a $u\to v$ path in $G-x$.
      Note that $\enter{T''}=t'$, so clearly $\enter{T''}\prec x$.
      Checking if any such path $T''=t'\to v$ exists can be
      performed as in case 2.

      It remains to show how to compute the earliest possible
      vertex $t'$ given $u'$.
      Observe that we want to pick such $s\in P[u',p_\ell]$
      that can reach the earliest vertex of $P$ using
      a satellite path.
      To this end we need to find
      the earliest out of the earliest backwards jumps of all vertices of $P[u',p_\ell]$.
      This can be computed in $O(1)$ time,
      after preprocessing the earliest backwards jumps
      for all suffixes of $P$ in linear time.

\end{enumerate}
To sum up, the query procedure tries to find
a $u\to v$ path $Q$ in $G-x$ for each of the above
four configurations of $\enter{Q}$, $\leave{Q}$ and $x$.
Each of the cases is handled in $O(\log{n})$ time.

\section{2-Reachability Queries}\label{s:2reach}

To answer 2-reachability queries we reuse the recursive approach
and extend the data structure of Section~\ref{s:failures-overview}.
Suppose we want to know whether there exist two vertex-disjoint
paths from $u$ to $v$ in some graph $G$ with suppressed set $A$ and separator $S_{ab}$
that arises in the recursive decomposition.
Recall that $S_{ab}-A$ can be decomposed into $O(1)$ simple
directed paths $P_1,\ldots,P_k$ in $G-A$.

First suppose that there is no $P_i$ such that
some $u\to v$ path goes through $V(P_i)$ in $G-A$.
Then either no $u\to v$ path in $G-A$ exists at all, or
it is contained in at most one child subgraph $G[V_i]-A$ ($i\in \{1,2\}$), 
where $V_1,V_2$ are the subsets of $V(G)$ strictly on one
side of the separator.
In this case, in $O(1)$ time we reduce our problem to searching for two vertex-disjoint
$u\to v$ paths in a graph
$G'[V_i\cup\{r'\}]$ with suppressed set $(A\cap V_i)\cup\{r'\}$,
as described in Section~\ref{s:thorup}.

Otherwise, there exists a $u\to v$ path in $G-A$ that goes
through $V(S_{ab})$. We will handle this case without
further recursive calls (so, the time needed to reduce
the original problem to this case is clearly $O(\log{n})$).
Observe that this guarantees that, in fact, no $u\to v$ path
goes through a vertex of $A$ in $G$. 
This is because the suppressed set can only
contain (possibly contracted) vertices of the separators
in the ancestors of $G$ in the recursion tree,
and no $u\to v$ path could go through these separators.
Moreover, every $u\to v$ path in the input graph $G_0$ (i.e., $G_0$ denotes the ancestor
of $G$ in the recursion tree that is the root) is preserved in $G$.
Consequently, for any $x\in V(G)$ we can test whether
there exists a $u\to v$ path in $G-x$ by issuing the $(u,v,x)$
query to the 1-sensitivity oracle built for the input graph $G_0$.

In fact, we will focus on finding
some $x\in V(G)$ that lies on all $u\to v$ paths in $G$ (equivalently, in $G-A$).
In other words, we will be looking for a \emph{separating vertex} certifying that
$v$ is not 2-reachable from $u$.
There are a few cases to consider.

First, assume that we can find two paths $P_i,P_j$, $i\neq j$
such that there exists both a $u\to v$ path through $P_i$ and a $u\to v$ path through
$P_j$ in $G$.
If some vertex $x\in V(G)$ lies on all $u\to v$ paths in $G$,
then either $x\in V(P_i)\cap V(P_j)$ or $x\in V(G)\setminus V(P_i)$
or $x\in V(G)\setminus V(P_j)$.
We can check all $x\in V(P_i)\cap V(P_j)$ in $O(\log{n})$ time using only $O(1)$ queries to our 1-sensitivity reachability oracle of Theorem~\ref{t:mainresult}, since 
$|V(P_i)\cap V(P_j)|=O(1)$.
In Section~\ref{s:2reach-not-on-path} we show that
the two latter cases, i.e., finding a separating vertex outside a path through
which one can reach $v$ from $u$, can be handled separately in $O(\log{n})$ time
after additional linear preprocessing.

Finally, suppose there is a unique path $P_i$
such that there exists a $u\to v$ path through $V(P_i)$ in $G$.
As we mentioned above, if some vertex $x\in V(G)\setminus V(P_i)$ lies on all $u\to v$
paths in $G$, it can be found as described in Section~\ref{s:2reach-not-on-path}.
Suppose that $x\in V(P_i)$.
Note that by the uniqueness of $i$, all $u\to v$ paths in $G$ are
preserved in $G[\bar{V_i}]$, where
$\bar{V_i}=V(G)\setminus V(S_{ab})\cup V(P_i)$.
Similarly to Section~\ref{s:failures-overview},
we can assume that $u$ and $v$ lie in a single connected component
of $G[\bar{V_i}]$ and the endpoints of $P_i$
lie on a single face of that component.
In Section~\ref{s:2reach-on-path} we show
that these assumptions enable us to find a separating vertex
on a path through which one can reach $v$ from $u$
in $O(\log^{2+o(1)}{n})$ time after additional $O(n\log^{5+o(1)}{n})$ preprocessing
and using $O(n\log^{2+o(1)}{n})$ space (per node of the recursion tree).
The total preprocessing time and space consumption can be analyzed
analogously as in Section~\ref{s:thorup}.
The following theorem summarizes our data structure.
\thmtworeach*
\subsection{Looking for a separating vertex outside the path}\label{s:2reach-not-on-path}

Let $G$ be a plane digraph and 
let $p_1\ldots p_\ell=P\subseteq G$ be a directed path.
Our goal in this section is to support queries of the following
form: for any $u,v\in V(G)$ such that there exists a $u\to v$ path going
through $V(P)$ in $G$, find a vertex $x\in V(G)\setminus V(P)$
such that all $u\to v$ paths in $G$ go through $x$, or decide there is no such vertex.
We start by fixing notation similar to that from Section~\ref{s:ds1}

For $i=\ell,\ldots,1$
let $L_i$ be the subset of $V(G)\setminus V(P)$
reachable from $p_i$ by a satellite path (wrt.~$P$), minus $\bigcup_{j=i+1}^\ell L_j$.
Denote by $D_i$ the dominator tree of $G[p_i\cup L_i]$.

Analogously, for $i=1,\ldots,\ell$, let
$\rev{L}_i$ be the subset of $V(G)\setminus V(P)$
that can reach $p_i$ by a satellite path, minus $\bigcup_{j=1}^{i-1} \rev{L}_j$.
Denote by $\rev{D}_i$ the dominator tree of $\rev{G}[p_i\cup \rev{L}_i]$.

All the sets $L_i,\rev{L}_i$ and all dominator trees $D_i,\rev{D}_i$
can be computed in linear time similarly as in Section~\ref{s:ds1}.
Let the functions $\vf_G,\vl_G,\vfs_G,\vls_G$ be defined (wrt. $P$) as in Section~\ref{s:ds1}.
For brevity we use $\prec$ to denote the order $\prec_P$.

We start with the first key lemma that allows us to reduce the
case when the vertices $\vfs_G(u)$ and $\vfs_G(v)$ are not strongly connected
to querying for an existence of a $u\to v$ path in $G-x$ for only two possible vertices $x$ and using the 1-sensitivity data structure that we have developed.

\begin{lemma}\label{l:2reach-not-on-path}
  Let $u,v\in V(G)$. Suppose the following conditions are satisfied:
  \begin{enumerate}
    \item there exists a $u\to v$ path in $G$ that goes through $V(P)$,
    \item $\vfs_G(u)=p_i$ and $\vls_G(v)=p_j$ are not strongly connected in $G$,
    \item there exists $x\in V(G)\setminus V(P)\setminus \{u,v\}$
    such that all $u\to v$ paths in $G$ go through $x$.
  \end{enumerate}
  Then either all $u\to v$ paths in $G$ go through
  the parent of $u$ in $\rev{D}_i$, or all of them go through the parent of $v$ in $D_j$.
\end{lemma}
\begin{proof}
  Since $u$ can reach $v$ through $V(P)$, $\vfs_G(u)=p_i$ and $\vls_G(v)=p_j$ indeed exist.
  Moreover, for the same reason some path $Q$ from $\vf_G(u)$ to $\vl_G(v)$ exists.
  Clearly, there exist some paths $Q_u=\vfs_G(u)\to \vf_G(u)$ and $Q_v=\vl_G(v)\to \vls_G(v)$
  in $G$.
  If we had $\vls_G(v)\preceq \vfs_G(u)$, then $Q_v\cdot P[\vls_G(v),\vfs_G(u)]\cdot Q_u\cdot Q$
  would form a cycle in $G$ and thus $\vls_G(v)$ and $\vfs_G(u)$ would be strongly connected,
  a contradiction.
  Hence, we obtain $\vfs_G(u)\prec \vls_G(v)$.

If there were satellite paths $R=u\to \vfs_G(u)$ and $T=\vls_G(v)\to v$
such that $x$ is a vertex of neither $R$ nor $T$, 
then $RP[\vfs_G(u),\vls_G(v)]T$ would be a $u\to v$ path
not going through $x$, which is impossible.
Assume wlog. that $x$ lies on all
$\vls_G(v)\to v$ satellite paths in~$G$
(the proof when $x$ lies on all $u\to \vfs_G(u)$ satellite paths is symmetric).

  Let $w$ be the parent of $v$ in $D_j$.
  Note that since $x\notin V(P)$ and $x$ lies on all $\vls_G(v)\to v$ satellite paths,
  we have $\vls_G(x)=\vls_G(v)$ and hence $x\in L_j$.
  Moreover, since $x\neq v$, $x$ is a non-root ancestor of $v$ in $D_j$.
  It follows that $w\neq \vls_G(v)$ and, by Lemma~\ref{lemma:paths-through-SAP}, any $x\to v$ satellite path goes through $w$.

We will now prove that all $u\to v$ paths in $G$ go through $w$.
  To this end, suppose the contrary, that there exists a $u\to v$
path $Y$ that does not go through $w$.
  First consider the case when $V(Y)\cap V(P)\neq \emptyset$.
$Y$ can be expressed as $Y'Z$, where $Z$ is a satellite path $p\to v$,
  for some $p\in V(P)$.
By the definition of $\vls_G(v)$, $p\preceq \vls_G(v)$.
  If $x\in V(Z)$, then $w\in V(Z)$, since any $x\to v$ satellite
  path goes through $w$, a contradiction.
  Hence $x\notin V(Z)$ and thus $p\neq\vls_G(v)$, which implies $p\prec \vls_G(v)$.

  But $Y$ has to go through $x$, so there exists a path $x\to p$ in $G$.
Since $x$ is reachable from $\vls_G(v)$ (as $x\in L_j$) and $\vls_G(v)$ is reachable from $p$,
we conclude that $p$ and $\vls_G(v)$ are strongly connected.
  If we had $p\preceq \vfs_G(u)$, then Lemma~\ref{l:path-scc} would
  imply that $\vfs_G(u)$ and $\vls_G(v)$ are strongly connected, a contradiction.
  Hence $\vfs_G(u)\prec p$.
  
  Take any satellite path $S=u\to \vfs_G(u)$.
  If $x\notin V(S)$, then $S\cdot P[\vfs_G(u),p]\cdot Z$ is a $u\to v$
  path not going through $x$, a contradiction.
  Otherwise, $x$ can reach $\vfs_G(u)$, $\vfs_G(u)$ can reach $\vls_G(v)$,
  and $\vls_G(v)$ can reach $x$.
  As a result, $\vfs_G(u)$ and $\vls_G(v)$ are strongly connected, a contradiction.

  Finally, suppose $V(Y)\cap V(P)=\emptyset$. Consider the $x\to v$ subpath $Z'$ of $Y$.
  Since $Z'$ is a satellite path, we showed that it has to go through $w$, a contradiction.
\end{proof}

We now turn to the case when $\vfs_G(u)$ and $\vls_G(v)$ are
strongly connected.
Recall that by Lemma~\ref{l:path-scc}, for any $p_j\in V(P)$, the
vertices of $P$ that are strongly connected to $p_j$ constitute a
subpath $P[p_i,p_k]$, where $i\leq j\leq k$.
Denote by $p_j^*$ the latest vertex of $P$ that is
strongly connected to $p_j$, i.e., $p_j^*:=p_k$.
Define $W_j=\bigcup_{l=i}^k \left(\{p_l\}\cup L_l\cup \rev{L}_l\right)$.
Let $D_j^*$ ($\rev{D_j^*}$) be the dominator tree of $G[W_j]$ ($\rev{G}[W_j]$, resp.)
with start vertex $p_j^*$.
Observe that the total number of different sets $W_j$ equals the
number of strongly connected components of $G$ that have a non-empty intersection with $V(P)$,
and we have $W_a=W_b$ if and only if $p_a$ and $p_b$ are strongly connected in $G$.
Each vertex $v\in V(G)$ belongs to at most two sets $W_j$:
if $v\in V(P)$, then it belongs to a unique set $W_j$, whereas
if $v\in V(G)\setminus V(P)$, then there is at most one $L_a$ containing
$v$ and at most one $\rev{L}_b$ with $v\in \rev{L}_b$.

Therefore, the total size of all distinct graphs $G[W_j]$ is linear
in the size of $G$. Consequently, we can compute all dominator
trees $D_j^*,\rev{D_j^*}$ in linear time.

\begin{lemma}\label{l:2reach-scc}
  Let $u,v\in V(G)$. Suppose the following conditions are satisfied:
  \begin{enumerate}
    \item there exists a $u\to v$ path in $G$ that goes through $V(P)$,
    \item $\vfs_G(u)=p_i$ and $\vls_G(v)=p_j$ are strongly connected in $G$.
  \end{enumerate}
  Then for any path $Q=u\to v$ such that $V(P)\cap V(Q)\neq\emptyset$,
  $Q\subseteq G$ if and only if $Q\subseteq G[W_i]$.
  Moreover, $u$ can reach $p_i^*$ in $G[W_i]$ and
  $v$ can be reached from $p_i^*$ in $G[W_i]$.
\end{lemma}
\begin{proof}
  We prove that no $u\to v$ path $Q$ going through $V(P)$ in $G$ goes through
  the vertices of $V(G)\setminus W_i$.
  This will imply that in fact any such $Q$ exists in $G$
  if and only if it exists in $G[W_i]$.

  We start by noting that $V(Q)\subseteq V(P)\cup \bigcup_{i=1}^\ell (L_i\cup \rev{L}_i)$.
  Otherwise $Q$ would go through a vertex $w$ that lies neither on $P$ nor in
  any of $L_i$ or $\rev{L}_i$.
  Since $Q$ also goes through some $p\in V(P)$, $w$ either can reach $p$,
  or is reachable from $p$, a contradiction.

  Observe that $Q$ cannot have
  a satellite subpath $z\to v$, where $z\in L_k$
  and $p_i^*\prec p_k$.
  If this was the case, $p_k$ could reach $v$ by a satellite path
  and $\vls_G(v)\preceq p_i^*\prec p_k$,
  which would contradict the definition of $\vls_G(v) = p_i$.
  Similarly we can prove that $Q$ does not contain
  any $u\to z$ satellite subpath such that
  $z\in \rev{L_l}$, where $p_l\prec p_i^*$ and
  $p_l$ is not strongly connected to $p_i^*$.

  Suppose that for some $p_k\in V(P)$, $p_i^*\prec p_k$, we have $p_k\in V(Q)$.
  Then the path $Q$ can be expressed as $Q=Q'T$,
  where $T=p_t\to v$ is a (possibly $0$-edge, if $v\in V(P)$) satellite path.
  If $v\in V(P)$, then $\vls_G(v)=v$ and $v\preceq p_i^*\prec p_k$.
  But $p_k$ can reach $v$, so $p_i^*$ and $p_k$ are strongly connected,
  which contradicts the definition of $p_i^*$.
  So we have $v\notin V(P)$.
  Let $T'$ be the $p_k\to p_t$ suffix subpath of $Q'$.
  If $p_k\neq p_t$, then $p_t\preceq \vls_G(v)\preceq p_i^*\prec p_k$, 
  so $p_t$ can reach $p_k$ and thus $p_k,p_i^*,\vls_G(v),p_t$
  are all strongly connected, which contradicts the definition of $p_i^*$.
  Therefore, $p_k=p_t$.
  But this means that $p_k$ can reach $v$ by a satellite
  path and $\vls_G(v)\preceq p_i^*\prec p_k$ which in turn contradicts
  the definition of $\vls_G(v)$.
  Symmetrically we prove that $Q$ cannot go through a vertex
  $p_l\in V(P)$ that appears on $P$ earlier than the earliest
  vertex of $P$ strongly connected to $p_i^*$.
  We conclude that each $p\in V(Q)\cap V(P)$ is strongly connected
  to $p_i^*$.

  Finally, suppose that $Q$ goes through a vertex $z\in L_k$, where $p_i^*\prec p_k$.
  We have already proven that $Q$ cannot contain a $z\to v$
  satellite subpath.
  So the $z\to v$ subpath of $Q$ has to go through a vertex $p\in V(P)$.
  But we have already proven that $p$ has to be strongly connected
  to $p_i^*$. Hence $p\preceq p_i^*$.
  So there exist paths $p_k\to z$, $z\to p$, $p\to p_i^*$, and $p_i^*\to p_k$.
  Thus, $p_i^*$ is strongly connected to $p_k$, which contradicts
  the definition of $p_i^*$.
  Similarly we prove that $Q$ cannot go through a vertex $z\in \rev{L_k}$ where
  $p_k\prec p_i^*$ and $p_k$ is not strongly connected to $p_i^*$.
  We conclude $Q\subseteq G[W_i]$.

  To see that $u$ can reach $p_i^*$ ($p_i^*$ can reach $v$) in $G[W_i]$ it is enough to note
  that $\vls_G(p_i^*)=\vfs_G(p_i^*)=p_i^*$ and apply
  what we have proved for $v:=p_i^*$ ($u:=p_i^*$ respectively).
\end{proof}

To make use of Lemma~\ref{l:2reach-scc}, we will need one more result.

\begin{lemma}\label{l:2reach-source}
  Let $G$ be a digraph and let $u,v,s$ be pairwise distinct vertices of $G$.
  Let $H\subseteq G$ be such that $H$ contains (preserves) all $u\to s$ paths in $G$
  and all $s\to v$ paths in $G$.
  Let $D$ ($\rev{D}$) be a dominator tree
  with start vertex $s$ in $H$ ($\rev{H}$, resp.).

  Let $d_u$ be the parent of $u$ in $\rev{D}$ and let
  $d_v$ be the parent of $v$ in $D$.
  Then $v$ is 2-reachable from $u$ in $G$ if and only if
  $u$ can reach $v$ in each of the graphs $G-d_u$, $G-d_v$ and $G-s$.
\end{lemma}
\begin{proof}
  The ``$\impliedby$'' direction is trivial so we only need
  to prove that existence of a path $u\to v$ in each of the graphs $G-d_u$, $G-d_v$ and $G-s$
  implies 2-reachability in $G$.
  Suppose there exists some $x\in V(G)\setminus\{u,v\}$ lying on all $u\to v$
  paths in $G$.
  If $x=s$, then there is no $u\to v$ path in $G-s$.

  Otherwise, if $x$ is neither an ancestor of $u$ in $\rev{D}$
  nor the ancestor of $v$ in $D$, then
  there exist a path $u\to s$ in $H-x$ and there exists
  a path $s\to v$ in $H-x$.
  Since $H-x\subseteq G-x$, $u$ can reach $v$ in $G-x$,
  a contradiction.

  Therefore, $x$ either lies on all $u\to s$ paths in $H$,
  or it lies on all $s\to v$ paths in $H$.
  Wlog. assume the former case, as the proof in the latter
  is symmetric.
  Since $H$ preserves all $u\to s$ paths of $G$,
  $x$ lies on all $u\to s$ paths in $G$.
  If $x=d_u$, we are done.
  Otherwise, $x$ is a proper ancestor of $d_u$ in $\rev{D}$,
  so, by Lemma~\ref{lemma:paths-through-SAP},
  $d_u$ lies on all $u\to x$ paths in $H$.
  Observe that in fact $d_u$ lies on all $u\to x$ paths in $G$.
  If this was not the case, there would be a $u\to s$ path
  in $H$ avoiding $d_u$, a contradiction.
  Therefore, since $d_u$ lies on all $u\to x$ paths in $G$,
  and $x$ lies on all $u\to v$ paths in $G$,
  $d_u$ also lies on all $u\to v$ paths in $G$.
  Equivalently, $u$ cannot reach $v$ in $G-d_u$.
\end{proof}

Now, Lemmas~\ref{l:2reach-scc}~and~\ref{l:2reach-source} allow us
to handle 2-reachability queries between pairs of vertices $u,v$
such that $\vfs_G(u)$ is strongly connected to $\vls_G(v)$.
Indeed, if $\vfs_G(u)=p_i$ and $\vls_G(v)=p_j$, then, by Lemma~\ref{l:2reach-scc}, $G[W_i]$
preserves all $u\to p_i^*$ and $p_j^*\to v$ paths in $G$.
Hence, by Lemma~\ref{l:2reach-source}, it is sufficient to 
look for a $u\to v$ path in $G-x$ using the 1-sensitivity data structure
for $w\in \{p_i^*,d_u,d_v\}$, where
$d_u$ is the parent of $u$ in $\rev{D_i^*}$, and $d_v$ is the parent
of $v$ in $D_j^*$.

\subsection{Looking for a separating vertex on the path}\label{s:2reach-on-path}

Let $G$ be a plane digraph and 
let $p_1\ldots p_\ell=P\subseteq G$ be a directed path
whose endpoints lie on a single face of $G$.
Our goal is to support queries of the following
form: for any $u,v\in V(G)$ such that there exists a $u\to v$ path going
through $V(P)$ in $G$, find a vertex $x\in V(P)$
such that all $u\to v$ paths in $G$ go through $x$, or decide there is none.

Observe that in the previous section, the query algorithm
using Lemmas~\ref{l:2reach-scc}~and~\ref{l:2reach-source}
for the case when $\vfs_G(u)$ and $\vls_G(v)$ are strongly connected
actually did not require the separating vertex $x$ that we were looking
for to lie outside of $P$.
We will leverage this fact so that we can additionally
assume that $\vfs_G(u)$ and $\vls_G(v)$ are not strongly connected.

We can focus on the case when $u$ cannot reach $v$ in $G-V(P)$,
since otherwise clearly no failure of a vertex of $P$ can make
$v$ unreachable from $u$.

Similarly to the proof of Lemma~\ref{l:2reach-not-on-path},
from the fact that $\vfs_G(u)$ and $\vls_G(v)$ are
not strongly connected we can conclude that $\vfs_G(u)\prec \vls_G(v)$.
Define $l$ to be the latest vertex of $P$ such that
$l\preceq \vls_G(v)$ and there exists a satellite path $u\to l$.
Note that $l$ exists and $\vfs_G(u)\preceq l$ by $\vfs_G(u)\prec\vls_G(v)$.
Similarly, let us define $f$ to be the earliest vertex of $P$
such that $\vfs_G(u)\preceq f$ and there exists a $f\to v$ satellite path.
Analogously, $f$ exists and $f\preceq \vls_G(v)$.
Note that by using the data structure of Lemma~\ref{l:reach-interval}, the vertices $f,l$
can be computed in $O(\log{n})$ time after linear preprocessing.

\begin{lemma}\label{l:2reach-aux1}
  Suppose for some $x\in V(P)$ there is no $u\to v$ path in $G-x$.
  Then $l\preceq x\preceq f$.
\end{lemma}
\begin{proof}
  We prove that if $x\prec l$, then there is a $u\to v$ path in $G-x$.
  The proof if $f\prec x$ is analogous.
  Take any satellite path $R=u\to l$ and any $\vls_G(v)\to v$ satellite path $T$.
  Then, since $x\prec l\preceq \vls_G(v)$, $R\cdot P[l,\vls_G(v)]\cdot T$
  is a $u\to v$ path in $G-x$, a contradiction.
\end{proof}

By Lemma~\ref{l:2reach-aux1}, if $f\prec l$, no vertex of $P$
lies on all $u\to v$ paths in $G$.
Hence, in the following let us assume $l\preceq f$.
Moreover, since in $O(\log{n})$ time we can check whether a single vertex
lies on all $u\to v$ paths in $G$, without loss of generality
we can assume that none of the vertices $f,l$
lies on all $u\to v$ paths in $G$.
Consequently, we assume $l\prec x\prec f$.

Let $x\in V(P(l,f))$ and consider any path $Q_x=u\to v$ in $G-x$.
Clearly, $Q_x$ can be expressed as $RST$, where $R$ is a $u\to \enter{Q_x}$
satellite path and $T$ is a $\leave{Q_x}\to v$ satellite path.

We now prove a few structural lemmas.

\begin{lemma}\label{l:2reach-enter-leave}
  We have:
  \begin{itemize}
    \item either $\vfs_G(u)\preceq \enter{Q_x}\preceq l$ or $\vls_G(v)\prec \enter{Q_x}$.
    \item either $\leave{Q_x}\prec \vfs_G(u)$ or $f\preceq \leave{Q_x}\preceq \vls_G(v)$.
  \end{itemize}
\end{lemma}
\begin{proof}
  We only prove the former item, as the proof of the latter is completely analogous.
  $\enter{Q_x}\prec \vfs_G(u)$ is impossible by the definition of $\vfs_G(u)$.
  On the other hand, $l\prec \enter{Q_x}\preceq \vls_G(v)$ would contradict
  the definition of $l$.
\end{proof}
\begin{lemma}\label{l:2reach-impossible}
  $\vls_G(v)\prec \enter{Q_x}$ and $\leave{Q_x}\prec\vfs_G(u)$ cannot hold simultaneously.
\end{lemma}
\begin{proof}
  Suppose $\vls_G(v)\prec \enter{Q_x}$ and $\leave{Q_x}\prec\vfs_G(u)$ at the same time. Then since $\vfs_G(u)\prec \vls_G(v)$, there
  is a $\leave{Q_x}\to\vfs_G(u)\to\vls_G(v)\to\enter{Q_x}$ path in $G$.
  But by the definition of $\enter{Q_x}$ and $\leave{Q_x}$, there is also
  a $\enter{Q_x}\to \leave{Q_x}$ path in $G$.
  Hence, $\vfs_G(u)$ and $\vls_G(v)$ are strongly connected in $G$, a contradiction.
\end{proof}

\begin{lemma}\label{l:2reach-g}
  Suppose $\vfs_G(u)\preceq \enter{Q_x}\preceq l$ and $\leave{Q_x}\prec \vfs_G(u)$.
  Then there exists a $u\to v$ path $Q_x'$ in $G-x$ such that
  $\enter{Q_x'}=\vfs_G(u)$ and $\leave{Q_x'}=g$, where $g$ is the latest
  vertex $g\prec \vfs_G(u)$ such that there exists a $g\to v$ satellite path in $G$.
\end{lemma}
\begin{proof}
  First note that $g$ indeed exists since $\leave{Q_x}\prec \vfs_G(u)$.
  Let $R'$ be any $u\to \vfs_G(u)$ satellite path and let $T'$ be any
  $g\to v$ satellite path in $G$.
  Then $R'\cdot P[\vfs_G(u),\enter{Q_x}]\cdot S\cdot P[\leave{Q_x},g]\cdot T'$
  forms a desired $u\to v$ path $Q_x'$ in $G-x$.
\end{proof}

The proof of the following lemma is symmetric to the proof of the above Lemma,
so we omit it.
\begin{lemma}\label{l:2reach-h}
  Suppose $f\preceq \leave{Q_x}\preceq \vls_G(v)$ and $\vls_G(v)\prec \enter{Q_x}$.
  Then there exists a $u\to v$ path $Q_x'$ in $G-x$ such that
  $\leave{Q_x'}=\vls_G(v)$ and $\enter{Q_x'}=h$, where $h$ is the earliest
  vertex $\vls_G(v)\prec h$ such that there exists a $u\to h$ satellite path in $G$.
\end{lemma}

\begin{lemma}\label{l:2reach-detour}
  Suppose $\vfs_G(u)\preceq \enter{Q_x}\preceq l$ and $f\preceq \leave{Q_x}\preceq \vls_G(v)$.
  Let $u'= \vf_{G-x}(\vfs_G(u))$, $v'=\vl_{G-x}(\vls_G(v))$.
  Then:
  \begin{enumerate}
    \item $u'\preceq \mnv{Q_x}$ and $\mxv{Q_x}\preceq v'$.
    \item There exists a $u\to v$ path in $Q_x'=Q_1P_1P_2P_3Q_2$ in $G-x$, where
      $Q_1$ is a $u\to u'$ path in $G-x$,
      $Q_2$ is a $v'\to v$ path in $G-x$,
      $P_1$ and $P_3$ are possibly empty subpaths of $P$,
      and $P_2$ is a minimal detour of $x$.
    \item If $x$ is not strongly connected to $\vfs_G(u)$, then $u'=\vf_G(\vfs_G(u))$.
      Similarly,  if $x$ is not strongly connected to $\vls_G(v)$, then $v'=\vl_G(\vls_G(v))$. \end{enumerate}
\end{lemma}
\begin{proof}
  Clearly, $u'\preceq \vfs_G(u)$ and $\vls_G(v)\preceq v'$.
  Since
  ${\vfs_G(u)\preceq \enter{Q_x}\preceq l}$ and $f\preceq \leave{Q_x}\preceq \vls_G(v)$,
  we have   $u'\preceq \vfs_G(u)\prec x\prec \vls_G(v)\preceq v'.$

  Item (1) follows by applying Lemmas~\ref{l:move-u}~and~\ref{l:move-v} to $Q_x$.

  Let $Q_1$ be the $u\to u'$ path and let $Q_2$ be any
  $v'\to v$ in $G-x$.
  Consider the path $Q_x''=Q_1\cdot P[u',\enter{Q_x}]\cdot S\cdot P[\leave{Q_x},v']\cdot Q_2$
  in $G-x$.
  The path $Q_x'=Q_1P_1P_2P_3Q_2$
  as required by item (2) can be obtained
  by applying Lemma~\ref{l:simplify} to $Q_x''$.

  Finally to prove the first part of item (3), suppose that $\vf_G(\vfs_G(u))\prec u'$.
  Then every $\vfs_G(u)\to \vf_G(\vfs_G(u))$ path in $G$ goes through $x$.
  Since $\vf_G(\vfs_G(u))\preceq \vfs_G(u)\prec x$, there is also a $\vf_G(\vfs_G(u))\to \vfs_G(u)$ path in $G$.
  Hence, $\vfs_G(u)$ and $x$ are strongly connected in $G$.
  The proof that $\vl_G(\vls_G(v))=v'$ if $x$ is not strongly connected to $\vls_G(v)$ is analogous.
\end{proof}

Having proven the above properties, we move to describing
the query procedure that finds some $x\in V(P)$ that
lies on all $u\to v$ paths, if one exists.
Suppose there exists some $x\in V(P(l,f))$ lying on all $u\to v$ paths in $G$.
Let $g$ be defined as in Lemma~\ref{l:2reach-g}, and $h$ be defined as in Lemma~\ref{l:2reach-h}.
Such an $x$, in order to be a separating vertex for the pair $u,v$, has to satisfy the following conditions.
\begin{enumerate}[label={(\arabic*)}]
\item  By Lemma~\ref{l:2reach-g}, either $g\prec \vfs_G(u)$ does not exist
or there is no path $\vfs_G(u)\to g$ in $G-x$.
    Observe that since $g$ (if exists) satisfies $g\prec \vfs_G(u)\prec x$,
    then there is no $\vfs_G(u)\to g$ path in $G-x$ if and only
    if $g$ is not strongly connected to $\vfs_G(u)$ in $G-x$.
\item Similarly, by Lemma~\ref{l:2reach-h}, either $h$, $\vls_G(v)\prec h$ does not exist
or $\vls_G(v)$ is not strongly connected to $h$ in $G-x$.
\item By Lemma~\ref{l:2reach-detour}, for each minimal detour pair $(a_x,b_x)$ of $x$
  we should have either $a_x\prec \vf_{G-x}(\vfs_G(u))$ or $\vl_{G-x}(\vls_G(v))\prec b_x$.
\end{enumerate}
By Lemmas~\ref{l:2reach-enter-leave}~and~\ref{l:2reach-impossible}, if all conditions (1)-(3)
are satisfied, then we rule out all possible configurations of $\enter{Q_x}$ and $\leave{Q_x}$ with
respect to $x$ on $P$.
In other words, if all 
of these conditions hold simultaneously, $x$ indeed lies on all $u\to v$ paths in~$G$.
Recall that $\vfs_G(u)$ and $\vls_G(v)$ are not strongly connected,
so $x$ cannot be simultaneously strongly connected to both $\vfs_G(u)$ and $\vls_G(v)$.
By Lemma~\ref{l:path-scc}, the subpath $P(l,f)$ can be split into three possibly
empty subpaths $P(l,f)=P_u,P_0,P_v$ (in this order), such that $x\in V(P(l,f))$ is strongly connected
to $\vfs_G(u)$ in $G$ iff $x\in V(P_u)$, to $\vls_G(v)$
in $G$ iff $x\in V(P_v)$, and to neither $\vfs_G(u)$ nor $\vls_G(v)$ iff $x\in V(P_0)$.

Suppose first that $x\in V(P_u)$, i.e, $x$ is strongly connected
to $\vfs_G(u)$, but not to $\vls_G(v)$ (in $G$).
The case when $x\in V(P_v)$ is symmetric.
Since $x$ is not strongly connected to $\vls_G(v)$,
$h$ and $\vls_G(v)$ are strongly connected
in $G-x$ if and only if they are strongly connected in $G$.
Hence, condition (2) does not depend on $x$ and implies that either $h$ does not exist
or it has to be strongly connected to $\vls_G(v)$ in~$G$.
Luckily, strong-connectivity in $G$ can be tested in $O(1)$ time.

By condition (3) we obtain that no minimal detour pair $(a_x,b_x)$ of $x$
can satisfy \linebreak $\vf_{G-x}(\vfs_G(u))\preceq a_x\preceq b_x\preceq \vl_G(\vls_G(v))$.
This is because $x$ not strongly connected to $\vls_G(v)$, 
implies, by Lemma~\ref{l:2reach-detour}, that $\vl_{G-x}(\vls_G(v))=\vl_G(\vls_G(v))$.

Let $(a_x^1,b_x^1),(a_x^2,b_x^2)$ be the two minimal detour
pairs of $x$ (if there is only one, put $(a_x^2,b_x^2):=(a_x^1,b_x^1)$; if there are none, set $(a_x^1,b_x^1):=(a_x^2,b_x^2):=(x,x)$).
Denote by $\bar{a}_x^i$ the earliest vertex of $P$ such
that $a_x^i\prec \bar{a}_x^i\prec x$ and $a_x^i$ is not strongly connected
to $\bar{a}_x^i$ in $G-x$.
If no such $\bar{a}_x^i$ exists, put $\bar{a}_x^i=x$.
For a fixed $x$, each $\bar{a}_x^i$ can be found in $O(\log{n})$ time by 
combining binary-search with the strong-connectivity
under failures data structure as we did in Lemma~\ref{l:reach-interval}.
Therefore, the vertices $\bar{a}_x^i$ for all $x\in V(P)$ and $i$
can be computed in $O(n\log{n})$ time.

Wlog. assume that $a_x^1\preceq a_x^2$ and $b_x^1\preceq b_x^2$.
Observe that $\bar{a}_x^1\preceq \bar{a}_x^2$. \emph{Not} satisfying the condition
$$\vf_{G-x}(\vfs_G(u))\preceq a_x^i\preceq b_x^i\preceq \vl_G(\vls_G(v))$$
can be restated as $\bar{a}_x^i\preceq \vfs_G(u)\lor \vl_G(\vls_G(v))\prec b_x^i$.
So, we obtain that $$\forall_{i\in\{1,2\}}(\bar{a}_x^i\preceq \vfs_G(u)\lor \vl_G(\vls_G(v))\prec b_x^i)$$
has to be satisfied if $x\in V(P_u)$ lies on all $u\to v$ paths in $G$. 
By applying elementary transformations and using
the inequalities $\bar{a}_x^1\preceq \bar{a}_x^2$ and $b_x^1\preceq b_x^2$, the above formula can be rewritten as $F_x^1\lor F_x^2\lor F_x^3$, where
\begin{align*}
  F_x^1 &=\bar{a}_x^2\preceq \vfs_G(u)),\\
  F_x^2 &= \vl_G(\vls_G(v))\prec b_x^1,\\
  F_x^3 &= \bar{a}_x^1\preceq\vfs_G(u)\land \vl_G(\vls_G(v))\prec b_x^2.
\end{align*}

By condition (1) we obtain that if $x\in V(P_u)$ lies on all $u\to v$ paths in $G$, either $g$ does not exist
or $x$ has to be such that $g$ and $\vfs_G(u)$ are not strongly connected in $G-x$.
To characterize all such vertices $x$ in an algorithmically useful
way, we need the following structural lemma.
\begin{lemma}\label{l:2reach-separate-paths}
  Let $G$ be a digraph.
  There exists two trees $T_1,T_2$ with the following property:
  for any vertices $u,v$ strongly connected in $G$,
  there exists at most four paths $Q^{u,v}_1,\ldots Q^{u,v}_k$,
  such that each $Q^{u,v}_i$ is a path in either $T_1$ or $T_2$,
  and all $x\in V(G)$ such that $u$ and $v$ are not strongly connected
  in $G-x$ are precisely the vertices of $Q^{u,v}_1\cup\ldots\cup Q^{u,v}_k$.
  
  The trees $T_1,T_2$ can be computed in linear time.
  Given $u,v$, the endpoints of paths
  $Q^{u,v}_1,\ldots Q^{u,v}_k$ can be computed in constant time.
\end{lemma}
\begin{proof}
  The lemma is essentially proved in~\cite[Lemma~8.1]{GeorgiadisIP15}. 
  There, it is shown that if $G$ is strongly connected, the following holds.
  Let $H$ and $\rev{H}$ are some loop-nesting trees of $G$ and $\rev{G}$ respectively.
  Let $D$ and $\rev{D}$ are dominator trees with the same source vertex $s$
  in $G$ and $\rev{G}$.
  Let $w$ ($\rev{w}$) be the nearest common ancestor of $u$ and $v$ in $H$ ($\rev{H}$, resp.).
  Then, the set $X_{uv}$ of vertices $x\in V(G)\setminus\{u,v\}$ such that $u$ and $v$ are not strongly connected
  in $G-x$ is equal to the set
  of proper ancestors of $u$ or $v$ in $D$ which are not proper ancestors of
  $w$ in $D$, plus the set of proper ancestors of $u$ and $v$ in $\rev{D}$ which are not proper
  ancestors of $\rev{w}$ in $\rev{D}$.
  Observe that $X_{uv}$ can be represented as a union
  of two paths in $D$ and two paths in $\rev{D}$.
  The endpoints of these paths can be computed in $O(1)$ time
  after preprocessing $D,\rev{D},H,\rev{H}$ in linear time so that
  nearest common ancestor queries on these trees are supported in constant time.

  The required tree $T_1$ ($T_2$) can be formed by arbitrarily connecting
  the dominator trees $D$ ($\rev{D}$) of individual strongly connected components of $G$.
\end{proof}
Hence in order for $x$ to separate $g$ and $\vfs_G(u)$ in $G$, $x$ has to lie
on one of at most $4$ paths $Q_1,\ldots,Q_k$ obtained
from Lemma~\ref{l:2reach-separate-paths} for $u:=g$ and $v:=\vfs_G(u)$.

To summarize, provided that $h$ either does not exist
or is not strongly connected to $\vls_G(v)$, $x\in V(P_u)$ lies on
all $u\to v$ paths in $G$ if and only if the following hold at the same time:
\begin{enumerate}[label={(\roman*)}]
  \item $F_x^1\lor F_x^2\lor F_x^3$,
  \item $x\in V(Q_1\cup \ldots\cup Q_k)$,
\end{enumerate}
whereas if additionally $g$ does not exist or is not strongly connected
to $\vfs_G(u)$, then it is enough that $x$ satisfies $F_x^1\lor F_x^2\lor F_x^3$.
Let us focus on the most involved case when $g$ is strongly connected
to $\vfs_G(u)$, as otherwise $x$ has less constraints to satisfy
and dummy constraints can be added to treat the problem analogously.
Note that (i) and (ii) hold simultaneously for $x\in V(P_u)$
if and only if $F_x^i\land (x\in V(Q_j))$ holds for some
of $O(1)$ pairs $i,j$.

Let us describe how to find some $x$ (if it exists) satisfying $F_x^3\land (x\in V(Q_j))$ only,
since this is clearly the most involved case.
We use the following data-structural lemma.

\begin{lemma}\label{l:2reach-tree-range}
  Let $T$ be an $n$-vertex tree and suppose each $v\in V(T)$ has
  assigned a $3$-dimensional label $(c^v_1,c^v_2,c^v_3)\in [n]^3$.
  Then, in $O(n\log^{5+o(1)}{n})$-time one can construct an
  $O(n\log^{2+o(1)}{n})$-space
  data structure 
  answering the following queries in $O(\log^{2+o(1)}{n})$ time:
  given $u,v\in V(T)$ and a $3$-dimensional rectangle
  $[a_1,b_1]\times [a_2,b_2]\times [a_3,b_3]$, find some vertex $w\in V(T)$
  lying on the $u\to v$ path in $T$ satisfying
  $a_i\leq c^w_i\leq b_i$ for $i=1,2,3$.
\end{lemma}
\begin{proof}
  We reduce our problem to a 4-dimensional orthogonal range reporting
  problem. It is known that in $O(n\log^5{n}\cdot \log\log{n})$ time one can construct
  an $O(n\log^{2}{n}\cdot (\log\log{n})^3)$-space
  data structure that can report some $k$ points (if they exist) in a $4$-dimensional
  rectangle in $O(\log{n}(\log\log{n})^3+k\log\log{n})$ time~\cite{KarpinskiN09}.
  We will use this data structure for $k=1$, i.e., for reporting at most one point.

  In our problem, we are given a tree and 3-dimensional labels.
  With so-called heavy-path decomposition~\cite{SleatorT83}, one can partition
  a tree $T$ into a set of vertex-disjoint paths $\mathcal{P}$ so that each 
  path in $T$ can be expressed, in $O(\log{n})$ time, as a concatenation of $O(\log{n})$
  subpaths of paths in $\mathcal{P}$.
  Equivalently, by concatenating $\mathcal{P}$
  into a single path $\mathcal{P}^*$, and ordering the vertices of $T$
  according to their position on $\mathcal{P}^*$,
  each path in $T$ can be expressed as a union of $O(\log{n})$
  intervals of the obtained order.
  We make this order our fourth dimension.
  This way, we can reduce our original query to
  $O(\log{n})$ 4-dimensional orthogonal range reporting queries (with pairwise-disjoint fourth dimension ranges),
  each of which can be processed in $O(\log{n}(\log\log{n})^3)$ time, since
  we only want to find one vertex or decide there is none.
  Hence $O(\log^2{n}(\log\log{n})^3)=O(\log^{2+o(1)}{n})$ query time.
\end{proof}

Now, to find $x\in V(P_u)=V(P[r_1,r_2])$ satisfying $F_x^2\land (x\in Q_j)$, suppose
$Q_j$ is a path in one of the two trees $T$ from Lemma~\ref{l:2reach-separate-paths}.
Let $\alpha(p_i):=i$ for all $p_i\in V(P)$.
We assign each vertex $p\in V(T)\cap V(P)$ the label $(\alpha(p),\alpha(\bar{a}_{p_i}^1),\alpha(b_{p_i}^2))$,
whereas for each vertex $v\in V(T)\setminus V(P)$ we use label $(0,0,0)$.
Given this, observe that we can find $x$ by searching for
a vertex on the path $Q_j$ of $T$ with
a label in the rectangle $[\alpha(r_1),\alpha(r_2)]\times [0,\alpha(\vfs_G(u))]\times [\alpha(\vl_G(\vls_G(v))),n]$.
By Lemma~\ref{l:2reach-tree-range}, such a search can be performed
in $O(\log^{2+o(1)}{n})$ time using a $O(n\log^{2+o(1)}{n})$-space
data structure that can be constructed in $O(n\log^{5+o(1)}{n})$ time.

Finally suppose that $x\in V(P_0)$.
This case is in fact easier than the previous one,
which can be explained as follows.
When $x$ is strongly connected to neither $\vfs_G(u)$ nor $\vls_G(v)$,
conditions (1) and (2) boil down to first checking
whether the vertices $g,h$ do not exist at all
or are at least are not strongly connected to $\vfs_G(u)$ and $\vls_G(v)$
respectively.
Then, we only need to check whether there exists
$x\in V(P_0)$ for which $F_x^1\lor F_x^2\lor F_x^3$ holds.
As already discussed, this can be achieved with a three-dimensional
orthogonal range reporting query.

\section{Proof of Theorem~\ref{thm:find-maxID-in-SCC}}\label{s:max-scc-failure}

Before we proceed with the proof, we first review the notion of loop nesting forest and a useful characterization from \cite{GeorgiadisIP17} of the SCCs of a graph under vertex failures.
\subsection{Loop nesting trees}
Let $G=(V,E)$ be a directed graph.
A \emph{loop nesting forest} represents a hierarchy of strongly connected subgraphs of $G$~\cite{st:t}, and is defined with respect to a dfs tree $T$ of $G$ as follows.
For any vertex $u$, the \emph{loop} of $u$, denoted by $\mathit{loop}(u)$, is the set of all descendants $x$ of $u$ in $T$ such that there is a path from $x$ to $u$ in $G$ containing only descendants of $u$ in $T$. Vertex $u$ is the \emph{head} of $\mathit{loop}(u)$. Any two vertices in $\mathit{loop}(u)$ reach each other. Therefore, $\mathit{loop}(u)$ induces a strongly connected subgraph of $G$; it is the unique maximal set of descendants of $u$ in $T$ that does so.
The $\mathit{loop}(u)$ sets form a laminar family of subsets of $V$:
for any two vertices $u$ and $v$, $\mathit{loop}(u)$ and $\mathit{loop}(v)$ are either disjoint or nested (i.e., one contains the other).
The above property allows us to define the \emph{loop nesting forest} $H$ of $G$, with respect to $T$, as the forest in which the parent of any vertex $v$, denoted by $h(v)$, is the nearest proper ancestor $u$ of $v$ in $T$ such that $v \in \mathit{loop}(u)$ if there is such a vertex $u$, and $\textbf{nil}$ otherwise.
Then $\emph{loop}(u)$ is the set of all descendants of vertex $u$ in $H$, which we will also denote as $H[u]$ (the subtree of $H$ rooted at vertex $u$).
Since $T$ is a dfs tree, every cycle contains a back edge \cite{dfs:t}.  More generally, every cycle $C$ contains a vertex $u$ that is a common ancestor of all other vertices $v$ of $T$ in the cycle \cite{dfs:t}, which means that any $v \in C$ is in $\mathit{loop}(u)$. Hence, every cycle of $G$ is contained in a loop.
A loop nesting forest can be computed in linear time~\cite{dominators:bgkrtw,st:t}.
In a strongly connected graph, each vertex is contained in a loop, so $H$ is a tree.
Since here we are going to apply the loop nesting forest on a strongly connected graph, we will refer to $H$ as the \emph{loop nesting tree} of $G$.

\subsection{SCCs in digraphs under failures}

\begin{theorem}[\cite{GeorgiadisIP17}]
	\label{cor:scc}
	Let $G$ be a strongly connected graph, $u$ be a vertex such that $G-u$ is not strongly connected, 
	and let $s$ be an arbitrary vertex in~$G$. Moreover, let $D$ (resp., $D^R$) to be the dominator tree of $G$ (resp., $G^R$), and $H$ (resp., $H^R$) be the loop nesting tree of $G$ 
	(resp., $G^R$), all rooted at $s$.
Let $C$ be a strongly connected component of $G-u$. Then one of the following cases holds:
\begin{itemize}
  \item[(a)] If $u\neq s$ is a non-leaf vertex in $D$ but a leaf in $D^R$ then either $C \subseteq D[u] \setminus \{u\}$ or $C = V \setminus D[u]$.
  \item[(b)] If $u\neq s$ is a leaf in $D$ but a non-leaf vertex in $D^R$ then either $C \subseteq D^R[u]\setminus \{u\}$ or $C = V \setminus D^R[u]$.
	\item[(c)] If $u\neq s$ is a non-leaf vertex in both $D$ and $D^R$ then either $C \subseteq D[u] \setminus D^R[u]$, or $C \subseteq D^R[u] \setminus D[u]$, or $C \subseteq D[u] \cap D^R[u]$, or $C = V \setminus \big( D[u] \cup D^R[u] \big)$.
	\item[(d)] If $u = s$ then $C \subseteq D[u] \setminus u$.
\end{itemize}
  Moreover, if  $C \subseteq D[u] \setminus \{u\}$ (resp., $C \subseteq D^R[u] \setminus \{u\}$) then $C=H[w]$ (resp., $C=H^R[w]$) where $w$ is a vertex in $D[u]\setminus \{u\}$ (resp., $D^R[u] \setminus \{u\}$) such that $h(w) \not \in D[u] \setminus \{u\}$ (resp.,  $h^R(w) \not \in D^R[u] \setminus \{u\}$).
\end{theorem}

\subsection{Decremental range minimum queries}

In the \emph{range minimum query} (RMQ) problem we are given an array $A$ of size $n$ with values in $\mathbb{R}$, and after bounded preprocessing time, for query indices $i,j$ we want to answer the minimum value in  $\{A[i],\dots,A[j]\}$. In the dynamic version of the problem one wants to support an intermixed sequence of queries and updates that alter the values on the array $A$. In the case where the updates can only increase (resp., decrease) the values of $A$ the algorithm is called incremental (resp., decremental). We use the following result from \cite{Wilkinson2014Amortized}.

\begin{theorem}[\cite{Wilkinson2014Amortized}]
\label{thm:RMQ}
There exists a data structure for decremental RMQ that requires $O(n)$ space, $O((\log n \log \log n)^{2/3})$ update time, and $O((\log n \log \log n)^{1/3})$ query time. 
\end{theorem}

\subsection{Proof of Theorem~\ref{thm:find-maxID-in-SCC}}
Let us recall the theorem we are going to prove.
\thmfindmax*

Observe that instead of computing the maximum-labeled vertex, we can instead look of the minimum label inside the SCC of $v$ in $G-x$,
as the maximum can be computed using the algorithm for minimum by assigning labels $f'(v) = -f(v)$.
In the following we are interested in vertices with minimum label $f(v)$.
	
To solve our problem, we will leverage Theorem \ref{cor:scc}.
Without loss of generality assume that $G$ is strongly connected, as otherwise, we can proceed with each SCC separately
since two vertices $u,v$ that are not strongly connected in $G$ are not strongly connected
in $G-x$ either.
Each query regarding the strongly connected component of $v$ is answered using the data structure of the SCC of $G$ containing $v$.
We choose an arbitrary start vertex $s$. 
We first handle the case $x=s$ in linear time by computing the strongly connected components of $G-s$
and preprocessing answers to all possible
queries.

Suppose $x\neq s$.
By Theorem \ref{cor:scc}, any strongly connected component $C$ of $G-x$ satisfies $V(S)\subseteq D[x]\cup D^R[x]$,
or is equal to the SCC of $s$ in $G-x$.
We consider those two cases separately.
	
	\paragraph{SCCs of $s$ in $G-x$.}
	We first deal with the case when the SCC $C$ containing $v$ in $G-x$ also contains $s$.
	Our goal is to precompute all answers for all possible failed vertices $x$.
	By Theorem \ref{cor:scc}, the SCC of $s$ contains precisely all vertices in $V\setminus \{D[x]\cup D^R[x]\}$.
	Hence, the minimum value in the SCC of $s$ in $G-x$ is the minimum $f(w)$ among all vertices $w$ such that $w\notin D[x]$ and $w\notin D^R[x]$.
	Let $\pord:V\to [1..n]$ be some preorder of the dominator tree $D$, where $n=|V|$.
  Let $\tsz(v)=|V(D[v])|$.
	Similarly let  $\pordr:V\to [1..n]$ be some preorder of the dominator tree $D^R$.
  Let $\tszr(v)=|V(D^R[v])|$.
	Clearly, all $\pord, \pordr, \tsz, \tszr$  can be computed in linear time.
	
	Notice that we have $u\notin D[x]$ iff $\pord(u)\notin [\pord(x),\pord(x)+\tsz(x)-1]$. Similarly, $u\notin D^R[x]$ iff $\pordr(x)\notin [\pordr(x),\pordr(x)+\tszr(x)-1]$.
	
	We map each vertex $u\in V$ to a point $(\pord(u),\pordr(u))$ on the two dimensional grid $[n]\times[n]$.
	Let $A$ be the set of obtained points; clearly $|A|=n$.
	Based on our previous discussion it is clear that  $u\notin \{D[x]\cup D^R[x]\}$ if an only if $$(\pord(x),\pordr(x)) \in \left( [1,\pord(x)-1] \cup [\pord(x)+\tsz(x),n]\right) \times \left([1,\pordr(x)-1] \cup [\pordr(x)+\tszr(x),n] \right).$$
  Hence, computing the minimum value $f(u)$ among all vertices $u\notin \{D[x]\cup D^R[x]\}$ can be reduced to finding
  the minimum-labeled point in the following four two-dimensional rectangles: 
	\begin{itemize}
		\item $[1,\pord(x)-1] \times [1,\pordr(x)-1]$.
		\item $[1,\pord(x)-1]\times[\pordr(x)+\tszr(x),n]$.
		\item $[\pord(x)+\tsz(x),n]\times[1,\pordr(x)-1]$.
		\item $[\pord(x)+\tsz(x),n]\times[\pordr(x)+\tszr(x),n]$.
	\end{itemize} 
  Observe that all the above rectangles are actually 2-sided, in
  the sense that each side has at most one endpoint
  that is not equal to $1$ or $n$, i.e., not equal to the
  minimum/maximum possible coordinate.

	In order to precompute the minimum value in the SCC of $s$ for all possible failures $x$,
  we can simply find the minimum-labeled points in $4(n-1)$ 2-sided rectangles.
  We will only show how to compute minimum-labeled points
  in the rectangles of the form $[1,\pord(x)-1] \times [1,\pordr(x)-1]$.
  The remaining cases can be handled analogously.
	Notice that in our case, we do not need compute
  minimum-labeled points in rectangles in an
  online manner, as we already know	the predefined set of $n-1$ queries to execute. 
	Hence, we can equivalently express our task as the following
  offline problem.
  Assume we have a set $R$ of red points (corresponding to our original points)
  and a set of $B$ blue points (corresponding to the rectangles) on a two dimensional grid $[n]\times[n]$,
  and for each blue point $b=(b_x,b_y)\in B$ we want to identify the red point 
  $r=(r_x,r_y)\in R$ such that $r_x\leq b_x$ and $r_y\leq b_y$
  with the minimum label $f(r)$.
	
	We solve this problem as follows. We first sort the points $B\cup R$ lexicographically so that red points appear earlier than blue points with the same coordinates; this can be done in $O(|B\cup R|+n)$ time by radix-sort. 
  We initialize an array $M[1..n]$, and fill $M$ with values $\infty$ (or a sufficiently large finite value).
  An instance of the decremental RMQ algorithm of Theorem~\ref{thm:RMQ} is applied on~$M$. 
	We process all points in $B\cup R$ in lexicorgaphical order.
  We will maintain the invariant that $M[y]$ is equal to the minimum
  label among all processed red points with second coordinate equal to $y$.
	
  Suppose we process the current point $p=(p_x,p_y)$.
  If $p$ is red and $f(p)<M[p_y]$ then we update $M[p_y]:=f(p)$; otherwise, if $f(p)\geq M[p_y]$, we do nothing.
	If the current point $p=(p_x,p_y)$ is blue, then we ask the RMQ data structure for the minimum value of $M$ in the range $[1,p_y]$.
	Notice that at the time the query is made, we have already processed precisely the points in $R$ whose coordinates
  are lexicographically no more
  than the coordinates of the blue point $p$.
  Moreover, $M[i]$ contains the smallest label of a red point with second coordinate $i$ that we have already processed, and hence, the minimum value in $\{M[1],\cdots, M[p_y]\}$ is the minimum value among the points $p'$ that we have processed and which have their second coordinate no larger than $p_y$.
	We conclude that the range minimum query on $M$ will correctly return the minimum label among all red
  points in the range $[1,1]\times [p_x,p_y]$. 
	
  Since we make $O(n)$ updates to $M$ (and thus to the RMQ data structure) and perform $O(n)$ queries to
  the RMQ data structure, the algorithm runs in
  $O(n (\log n \log \log n)^{2/3})$ time.

	\paragraph{SCCs in $D[x]\cup D^R[x]$.}
	We preprocess the loop-nesting-tree $H$ and we assign a weight to each edge $h(w)w$, to be equal to the minimum label $f(v)$ where  $v\in H[w]$. This can be easily done in linear time via a simple bottom-up visit of $H$. By Theorem \ref{cor:scc}, for each vertex $v$ that belongs to an SCC $C \subseteq D[x]$ there exists a $w$, such that $v\in H[w]$, for which it holds that $h(w)\notin D[x]$ and $H[w] = C$. Our goal is to prove that the weight of $h(w)w$ is the minimum weight in $C$ (the SCC of $v$ in $G-x$), and also the minimum weight on the path $nca_H(x,v)$. This will allow us to retrieve the weight of $h(w)w$ efficiently.
	
	Next, we prove the following two claims for a failing vertex $x$.
	\begin{itemize}
		\item [i)] If $h(w)\notin D[x]$, the nearest common ancestor $nca_{H}(w,x)$ of $w$ and $x$ in $H$ is $h(w)$.
		\item[ii)] The minimum label in the SCC of $v$ in $G-x$ equals to the minimum weight on the path from $v$ to $nca_H(x,v)$ in $H$. 
	\end{itemize}
	
	We first prove i).
	Consider the DFS tree $T$ that generated the loop-nesting-tree $H$.
	By the definition of $H$, $h(w)$ is the nearest ancestor of $w$ in $T$ such that $w$ and $h(w)$ are strongly connected in $G[T[h(w)]]$, that is, in the subgraph induced by the descendants of $h(w)$ in $T$.
	By Lemma \ref{lemma:paths-through-SAP}, since $h(w) \notin D[x]$,  all paths from $h(w)$ to $w$ in $G$ contain $x$, and in particular the path from $h(w)$ to $w$ in $T$.
	That means, $x$ lies on the path from $h(w)$ to $w$ in $T$.
	The existence of the path on $T$ from $h(w)$ to $x$ to $w$ implies that $x$ is strongly connected with $h(w)$ and $w$ in $G[T[h(w)]]$, and hence $h(w)$ is an ancestor of $x$ in $H$, by the definition of the loop nesting tree. Since $w\not = x$, the nearest ancestor of $x$ and $w$ is $h(w)$.
	
	Now we prove ii). 
  By Theorem~\ref{cor:scc}, the SCC of $v$ in $G-x$ is a subtree of $H$ rooted at a vertex $w$ such that $h(w)\notin \{D[x] \setminus \{x\}\}$. 
	Hence, the minimum label in the SCC of $v$ equals to the weight of the edge $h(w)w$, by claim i). 
	Moreover, the weight of each edge $e$ is larger than the label of the edges of $H$ incident to all descendants of both endpoints of edge $e$.
	
	Given claims i) and ii), in the case where $v$ is in $D[x]$, the minimum label in the SCC of $v$ is the minimum weight on the path between $v$ and $nca_{H}(x,v)$ in $H$. 
	We can identify the minimum weight on the path of a static tree in constant time, after $O(|E(H)|)=O(n)$ time preprocessing~\cite{DemaineLW14}.
  The nearest common ancestor queries can be answered in constant time after linear preprocessing as well~\cite{BenderF00}.
	We can do this analogously in the case where $v \in D^R[x]$, by preprocessing $D^R$ and $H^R$.
	
  Overall, we showed that we can answer the queries of the statement constant time, after\linebreak
  $O(m+n (\log n \log \log n)^{2/3})$ time preprocessing.

\newpage 
\bibliographystyle{plainurl}

{\small \bibliography{references}}

\end{document}